\documentclass[article]{s-artmod}
\usepackage{fourier}

\usepackage{amsmath, amssymb}
\usepackage{wrapfig}
\usepackage{braket,amsfonts}

\usepackage{array}

\usepackage[caption=false]{subfig}
\usepackage{pgfplots}

\newsiamthm{claim}{Claim}
\newsiamremark{rem}{Remark}
\newsiamremark{expl}{Example}
\newsiamremark{hypothesis}{Hypothesis}
\crefname{hypothesis}{Hypothesis}{Hypotheses}

\usepackage{algorithmic}

\usepackage{graphicx,epstopdf}

\Crefname{ALC@unique}{Line}{Lines}

\numberwithin{theorem}{section}

\usepackage{amsfonts}

\ifpdf
  \DeclareGraphicsExtensions{.eps,.pdf,.png,.jpg}
\else
  \DeclareGraphicsExtensions{.eps}
\fi

\usepackage{framed}

\newcommand{\TheTitle}{The Winner Takes All: Volume-Scavenging Populations of Networked Droplets} 
\newcommand{\TheAuthors}{Thomas C. Hagen and Paul H. Steen}

\headers{The Winner Takes All: Volume-Scavenging Droplets}{\TheAuthors}

\title{{\TheTitle}\thanks{arXiv submission of revision: 02/15/2019.  This manuscript is an annotated version of a related work on ``volume scavenging.'' It highlights results of a  physics-based model from different viewpoints, suggesting parallels to the socio-economic context among others. It also offers ``asides'' and annotations. The related work can be found at 
\url{https://doi.org/10.1016/j.physd.2019.01.005}. 
}}

\author{Thomas C. Hagen%
\thanks{Department of Mathematical Sciences,
    The University of Memphis,  Memphis, TN 38152, USA (\email{thagen@memphis.edu}).}%
\and
 Paul H. Steen%
\thanks{ School of Chemical and Biomolecular Engineering and Center for Applied Mathematics, Cornell University, Ithaca,  NY 14853, USA (\email{PHS7@cornell.edu}).}
}

\ifpdf
\hypersetup{
  pdftitle={\TheTitle},
  pdfauthor={\TheAuthors}
}
\fi
\definecolor{shadecolor}{gray}{0.9}

\begin{document}

\maketitle

\renewcommand{\abstractname}{Overview} 
\begin{snugshade}
\begin{abstract} In this work we present and analyze a fluid-mechanical model of competition (scavenging) amongst $N$ liquid droplets (individual competitors).
The eventual outcome of this competition depends sensitively on the average resource (volume) per individual $\overline{V}$.
 For  abundant resource, $\overline{V}>1$,  there is one winner only and that winner eventually scavenges all or most of the resource. In the socio-economic realm this is is known as the  ``winner-take-all" outcome:
A disproportionately large reward falls to one or a few winners, even though 
other competitors start out with comparable (or even slightly more) resource and perform only 
marginally worse. The losing competitors are not rewarded.
For less than  abundant
 resource,  $\overline{V}<1$, an outcome with  resource that is evenly partitioned amongst the $N$ droplets becomes possible. This is the ``all-share-evenly" or egalitarian outcome. For sufficiently scarce resource the egalitarian outcome is the only one that can occur. In addition to predicting what kind and how many winners, our  analysis shows that once an individual's resource (droplet volume) falls below a fixed threshold, that individual can neither recover nor emerge as winner.  This is 
the ``once down-and-out, always down-and-out" outcome. 
Selected simulations suggest that the winner depends
 sensitively on population size and the ``trading friction" or  inefficiency of resource exchange between individuals (liquid rheology). 
Of all feasible rest states (equilibria), only certain ones are reachable (stable equilibria). Friction turns out
 to strongly influence
 the time to reach an end state (stable equilibrium), in surprising ways.  Besides the end states, our 
analysis reveals an array of rest states, ordered in hierarchies of more versus less costly (energetic) outcomes.
\end{abstract}
\end{snugshade}
\renewcommand{\abstractname}{Abstract} 
\begin{abstract}
 A system of $N$ spherical-cap fluid droplets protruding from circular openings on a plane is  connected through channels.
This system is governed by surface tension acting on the droplets and viscous
stresses inside the fluid channels. The fluid rheology is given by the Ostwald-de Waele power law, thus permitting shear thinning.
The pressure acting on each droplet is caused by capillarity and given in terms of the droplet volume via the Young-Laplace law.
Liquid is exchanged along the network of fluid conduits due to an imbalance of the Laplace pressures between the droplets. In this way some droplets gain volume at the expense of others. This mechanism,
christened ``volume scavenging,'' leads to interesting dynamics. 

Numerical experiments show that an initial  droplet configuration is driven to a stable equilibrium exhibiting $1$  super-hemi\-spher\-i\-cal droplet and $N-1$ sub-hemi\-spher\-i\-cal ones when the initial droplet  volumes are large. The selection of this ``winning'' droplet depends not only on the channel network and the fluid volume, but also notably on the fluid rheology. The rheology is also observed to drastically change the transition to equilibrium.
 For smaller droplet volumes the long-term behavior is seen to be more complicated since the types of equilibria differ from those arising for larger volumes. These observations motivate our analytical study of equilibria and their stability for the corresponding nonlinear dynamical system.  The identification of equilibria is accomplished by locating the zeros of a mass polynomial, defined through the constant volume/mass constraint. The key tool in our stability analysis is a pressure-volume work functional, related to the total surface area, which serves as a Lyapunov function for the dynamical system. This functional is useful since equilibria are typically not hyperbolic and linearization techniques not available. Equilibria will be shown to be hierarchically organized in terms of size of the pressure-volume work functional. For larger droplet volumes this ordering exhibits one hierarchy of equilibria. Two hierarchies exist when the volumes are smaller. The minimizing equilibria in either case are asymptotically stable. 
\end{abstract}

\begin{keywords}
Winner-take-all system, equilibria, stability and bifurcation, nonlinearity,  gradient dynamical system, shear thinning, fluid droplet, surface tension, constrained optimization
\end{keywords}

\begin{AMS}
  70K50, 37C75, 34C23, 76D45
\end{AMS}

\section{Introduction}

\label{sec_in}

We consider a system of $N$ spherical-cap fluid droplets protruding from uniform, circular orifices on a flat surface. The droplets are connected through
straight fluid channels of uniform circular cross section and possibly variable length. The fluid is homogeneous and incompressible. The rheology of the fluid is  governed by the Ostwald-de Waele power law, thus, in particular, permitting shear thinning which is  often observed in complex (e.g.\,biological) fluids. The flow dynamics are dominated by surface tension acting on the droplets and viscous forces within the network channels. The pressure acting on each droplet is given by the Young-Laplace relation. Our  physics-based model accounts for scavenging amongst neighboring spherical-cap droplets owing to the action of capillary pressures due to surface tension. Volume exchange arises by pressure (curvature) differences that drive liquid from one to another droplet along a network of interconnected channels. In this way, certain droplets gain volume at the expense of others. This mechanism, christened ``volume scavenging,'' leads to interesting dynamics, driving an initial  droplet configuration to a stable equilibrium.
\begin{figure}[tbhp]
\centering
    \includegraphics[width=0.7\textwidth]{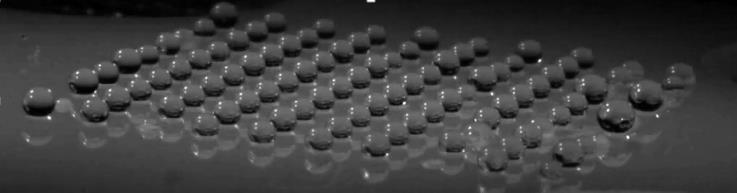}
\caption{An array of $10\times 10$ fluid droplets protruding from holes of $0.5$ millimeter diameter and connected  through a reservoir beneath, from \cite{VoSt}}
\label{drop-array}
\end{figure}
An array of fluid droplets connected to a common reservoir is displayed in \cref{drop-array}.

In the Newtonian case a small perturbation of  an initial configuration of super-hemispherical drop\-lets of equal size is known to evolve over time  into a configuration consisting of one ``winning'' super-hemi\-spher\-i\-cal droplet with all other droplets being  sub-hemispherical of equal size. Since this evolution is driven by a minimization of the total surface area, volume scavenging  is a coarsening process like Ostwald ripening  \cite{RaVo}. 
\begin{snugshade}
\noindent {\em Note to the reader}: Throughout this work we will highlight results of our physics-based model from different viewpoints, suggesting parallels to the socio-economic context among others. We will also offer ``asides."
\end{snugshade}
\begin{snugshade}
\noindent {\em Material science/fluid dynamics aside}: 
Coarsening phenomena are observed in  a large class of problems in material science and fluid dynamics, including binary mixtures and alloys, oil-water emulsions, epitaxial growth and re-crystallization.
\end{snugshade}
\begin{snugshade}\noindent {\em Socio-economic view}: Interestingly, coarsening has also been observed in the social sciences:  Schelling's agent-based model of segregation, possibly the best known representative of segregation models,  impressively demonstrates that the local interaction of individuals can lead to unexpected aggregate behavior (coarsening) \cite{Sc-art,Sc}. 
\end{snugshade}

The mathematical model we study is a generalization of the Newtonian volume scavenging model introduced by van Lengerich, Vogel and Steen \cite{LeEA1,LeEA2}. That model was used to explain the grab-and-release mechanism for  a
capillary adhesion device, described  by Vogel and Steen \cite{VoSt}, and  was motivated by a study of Eisner and Aneshansley \cite{EiAn} about the tarsal adhesion of the Florida tortoise beetle ({\em Hemisphaerota cyanea}) to defend against predators. In this work we consider the generalization of the previous model to power-law fluids. This generalization  has far-reaching consequences, both for the analytical treatment of the governing equations and for the observed dynamics of solutions. Our analysis, while extending    some of the previous results of van Lengerich et.\,al.\,\cite{LeEA1,LeEA2}, will do so without regard of the underlying conduit networks and without use of linearization techniques which are generally not available in the case of the power-law model. We will give a complete  classification of possible equilibria and their stability with the dimensionless average droplet volume $\overline{V}$ serving as a bifurcation parameter. The {\em full} range
$\overline{V}>0$ will be discussed.  One of our central results will be an {\emph {exhaustive}}, analytical identification of  hierarchies of stationary droplet configurations, classified by the size and number of large droplets versus small ones.

Volume scavenging has been the focal point of several important studies. Let us mention a few: Wente \cite{We} analyzed the two- and three-droplet regime from the vantage point of catastrophe theory. His seminal work  made use of an energy formulation to classify the stability of  equilibria.  Slater and Steen \cite{SlSt} discussed  the inviscid case of $N$ spherical-cap droplets under the symmetry assumption of equivariance with respect to  the permutation group $S_N$. Stability results  were also given  by van Lengerich et.\,al.\,\cite{LeEA1,LeEA2} for the Newtonian case with constant viscosity and large average volumes $\overline{V}>1$.  This was achieved by linearization about hyperbolic equilibria for a conduit network where the linearized equations turned out to be particularly simple.  The existence of a network-independent Lyapunov function made it then possible to extend the stability results to general networks. These works are the starting point  of our own study. Some of the  findings reported there will be included as special cases
 in this article. Yet far from being narrowly tailored, our analysis provides a template of arguments which is likely to generalize and be of use beyond this specific problem. 

We base our governing equations on  the dimensionless model introduced in \cite{LeEA1}. In this framework the uniform radius $R$ of each circular opening is rescaled and non-di\-men\-sion\-al\-ized to be $1$. The same length scale is used to rescale  the radius $r$ and the height $h$ of a spherical-cap droplet protruding from an orifice. Droplet volume $v$ is normalized such that the volume of a hemispherical droplet is $1$. Then radius $r$, height $h$ and volume $v$ of the  droplet are related by
\begin{equation}\label{randv}
	{2}\,r = h+\displaystyle{1\over h}\quad \text{and}\quad v = \displaystyle{1\over 4}\,h\,(h^2 +3).
\end{equation}
\begin{wrapfigure}{r}{0.46\textwidth}
\vspace*{-0.5cm}
\centering
    \includegraphics[width=0.45\textwidth]{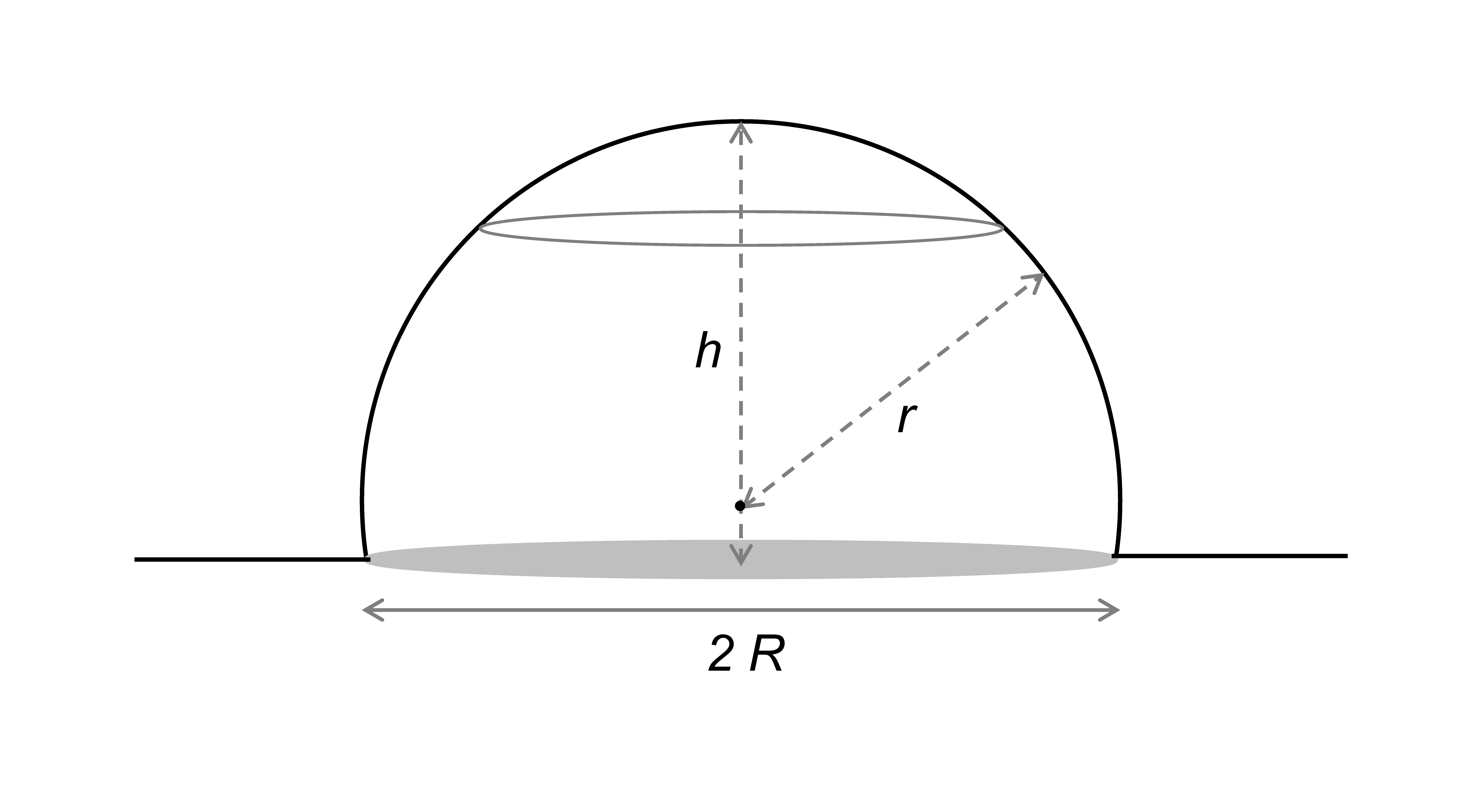}
\vspace*{-0.7cm}
 \caption{A spherical-cap droplet ($h>R$)}
  \label{droplet}
\end{wrapfigure}
A schematic of the spherical-cap droplet geometry is given in Figure~\ref{droplet}. It is also worthwhile to record the corresponding normalized surface area $s_A=s_A(h)$ of the droplet  in terms of its height $h$:
\begin{equation}\label{Adefn}
	s_A={3\over 2}\,\left(h^2+1\right).
\end{equation}
Note that $v=v(h)$ is invertible for $h\in \mathbb R$  (with $h\geq 0$ being of physical importance in our situation) and that, due to the chosen scalings,  $v=1$ if and only if $h=1$. Moreover, we have, of course,
\begin{equation}\label{vhasymp}
	|h| = O\left(|v|^{1/3}\right)\quad \text{as $|v|\rightarrow \infty$.}
\end{equation}
The  pressure $p$ acting on the spherical-cap droplet is caused by surface tension and given by the Young-Laplace law (after  non-dimensionalization):
\begin{equation}\label{ph}
	p = {2\over r}={{4\,h}\over {h^2+1}}.
\end{equation}
Hence the capillary pressure is the same for spherical-cap droplets of the same radius of curvature $r$. More precisely,  if $h_0$ is a solution of the equation $p(h)=\lambda$ with $\lambda\in {\mathbb R}_+$, then by \eqref{randv}, \eqref{ph}, $h_0^{-1}$ is also a solution. In fact, $h_0$ and $h_0^{-1}$ are the only solutions. They are distinct if and only if $\lambda\not= 1$. Consequently, similar to  the terminology used in \cite{LeEA1,LeEA2},  we call a spherical-cap droplet of height $h$ ``large'' if  $h>1$  ($v>1$) and ``small'' if $0\leq h\leq 1$ ($0\leq v\leq 1$).

Because of the relations between radius $r$, height $h$ and volume $v$ given above, we can  write  the pressure $p=p(h)$ in terms of the droplet volume  by setting
\begin{equation}
	P(V) = p(h)\quad \text{whenever $V=v(h)$.}
\end{equation}
In this way droplet pressure $P$ is defined as a function of droplet volume $V$. In a similar way, we obtain the surface area of a spherical-cap droplet as a function of its volume via
\begin{equation}
	S(V) = s_A(h)\quad \text{whenever $V=v(h)$.}
\end{equation}

\subsection{Networks of interaction}
\label{networks}
\begin{figure}[tbhp]
\centering
\vspace*{-0.3cm}
\subfloat[Complete network]{\label{graph1}
    \includegraphics[width=0.25\textwidth]{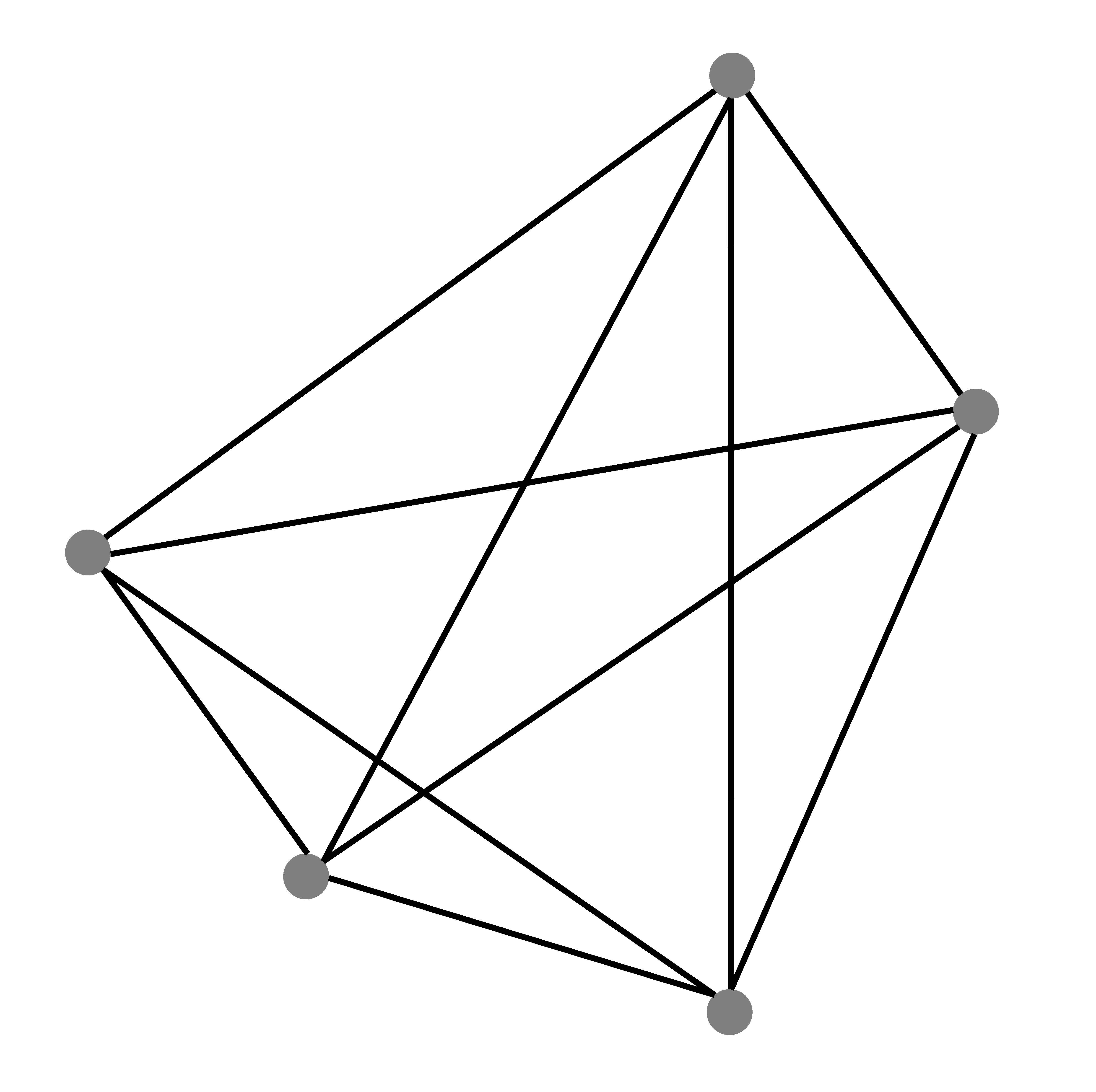}}
\subfloat[Star network]{\label{graph2}
 \includegraphics[width=0.25\textwidth]{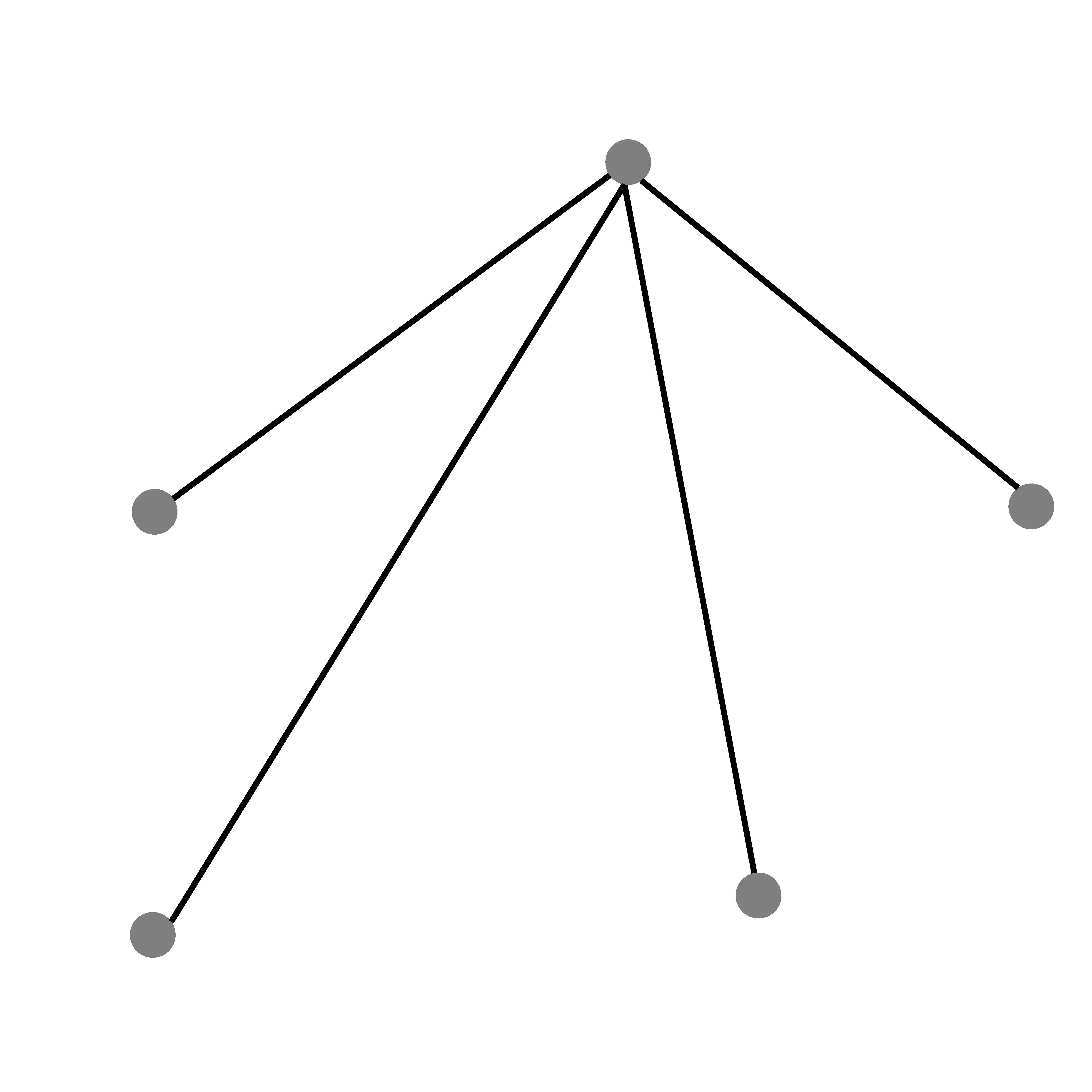}}
 \subfloat[Linear Network]{\label{graph3}
    \includegraphics[width=0.25\textwidth]{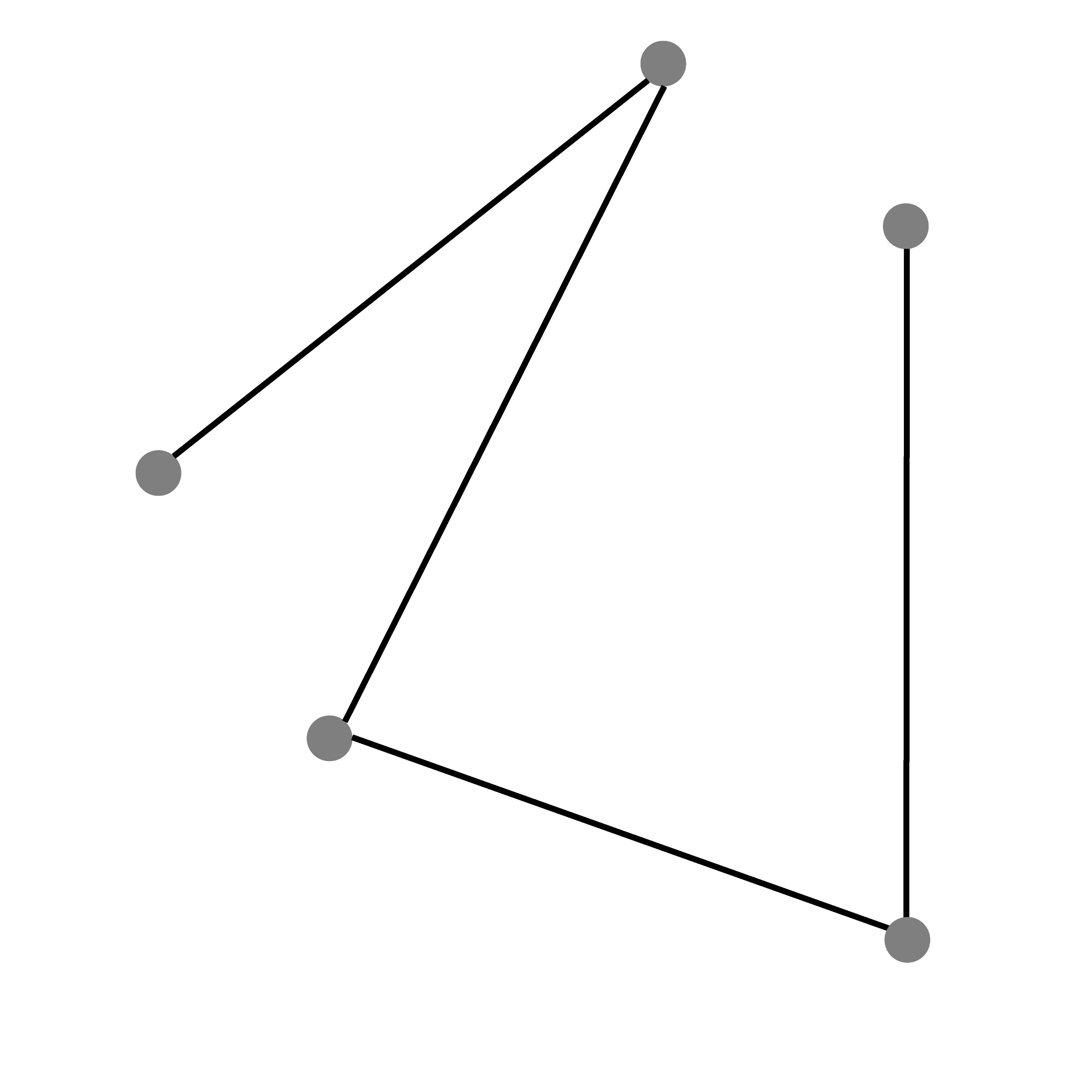}}
\caption{Simple, connected graphs with 5 vertices}
\label{graphs}
\end{figure}
We now consider a network of conduits of  circular cross section (with uniform radius) and possibly variable length connecting $N$ ($\geq 2$) spherical-cap fluid drop\-lets. The network is described by a simple, connected graph whose adjacency relation is given through a weighted adjacency matrix $(c_{i,j})$: For $1\leq i, j\leq N$,
\begin{align}
	 & c_{i,j} = c_{j,i}>0 \quad \text{if there is a channel connecting droplets $i$ and $j$ ($i\not= j$),}\\
	 & c_{i,j} = 0 \quad \text{in all other cases.}
\end{align}
The size of $c_{i,j}$ may be interpreted as a measure for the (inverse)  length of the channel between droplets $i$ and $j$. In a uniform network with conduits of equal length we may assume $c_{i,j}\in \{0,1\}$ after  rescaling. Examples of networks (graphs) with $N=5$ are given in Figure~\ref{graphs}. Important networks include the following simple, connected graphs:
\begin{itemize}
	\item {\em Complete network:} Each vertex is adjacent  to every other vertex.
	\item {\em Star network:} Exactly one vertex (star center)  is adjacent to every other vertex. No other edges are present.
	\item {\em Linear network:} Each vertex is adjacent to no more than two other vertices. Exactly two vertices are adjacent to only one other vertex.
\end{itemize}
\begin{snugshade}
\noindent {\em Socio-economic view}:  The network topology determines the level and efficiency of the $N$  individuals (droplet volumes) interacting with each other. Both local and global interaction of individuals can be modeled. 
\end{snugshade}

\subsection{Rates of exchange}
\label{exchange}
We denote the volume of droplet $j$ at time $t$ by $V_j=V_j(t)$ and  the volumetric flow rate  from droplet $i$ to droplet $j$ at time $t$ by $q_{i,j}=q_{i,j}(t)$. Then we obtain the
 change in volume $V_j$ at time $t$ from conservation of mass (volume):
\begin{equation}
	{d\over {dt}} V_j = \sum_{i=1}^N q_{i,j},\quad 1\leq j\leq N.
\end{equation}
For a power-law fluid the volumetric flow rate $q_{i,j}$ is assumed to be given by the dimensionless closed-form flow rate--pressure change   expression
\begin{equation}\label{fr-pd}
	q_{i,j} = c_{i,j}\, \Delta_s P_{i,j},\quad 1\leq i, j\leq N,
\end{equation}
where
\begin{equation}\label{defDs}
	\Delta_s P_{i,j} = |P(V_i)-P(V_j)|^s\,\text{sgn} \left(P(V_i)-P(V_j)\right)
\end{equation}
and the power-law parameter $s>0$ is the reciprocal of the power-law index of the fluid, see \cite{BiEA}. The case $s=1$ reduces this relationship to laminar pipe flow of
a Newtonian liquid. This is the situation studied in \cite{LeEA1,LeEA2}. For $s>1$ the fluid is shear thinning, i.e.\,the apparent viscosity decreases with increased stress. For $0<s<1$ the fluid is shear thickening where the opposite flow behavior is observed. Shear thinning is commonly seen  for real fluids (including biological fluids, paints and foods), while shear thickening is rare. Observed values of the power-law index in the case of shear-thinning fluids are listed in \cite{BiEA,Ta}. The corresponding values of the 
power-law parameter $s$ fall roughly  in the interval $1\leq s\leq 5$. In contrast, rheological data for shear-thickening behavior are scarce. In fact, we will observe in the next section that the shear thickening parameter range $0<s<1$ is problematic in our model.
\begin{snugshade}
\noindent {\em Socio-economic view}:  The liquid rheology determines the  ``trading friction" or inefficiency of resource exchange between the individuals (apparent viscosity). In the Newtonian regime ($s=1$) the trading friction is constant. In the shear thinning regime ($s>1$)  the trading friction is smaller when  the resource exchange is faster. 

It is worthwhile to note that all competitors perform according to a fixed set of rules (governing equations) and the results are deterministic.  There are no choices or decisions. While this rigidity disallows individuals to change course, one may interpret the resulting behavior as being optimal under the given rules.
\end{snugshade}


\begin{snugshade}
\noindent {\em Socio-economic view}: It is insightful to relate volume scavenging to the socio-economic concept of the 
``winner-take-all'' system (or market), a notion  systematically explored by Frank and Cook in their seminal work \cite{FrCo-art, FrCo}. Such a system is
 characterized by a concentration phenomenon where a disproportionately large reward falls to one or a few winners, even though other competitors start out with comparable (or even slightly more) resources and perform only marginally worse 
than the winner. The losing competitors are not rewarded. Winner-take-all systems are observed in sports, the entertainment industry, 
the rise of multi-national companies and elections. The ``egalitarian'' or ``all-share-evenly" system exhibits the opposite outcome. Here, comparable initial  resources and similar performance leads to comparable outcomes. Remarkably, both winner-take-all and all-share-evenly behaviors  can be  observed within our physics-based model.   This is an unusual feature  since  models which
 follow the  ``cumulative advantage principle"  typically exhibit  concentrated (coarsened) outcomes like the winner-take-all outcome only.  In our model the distinction of winner-take-all outcome versus egalitarian outcome depends on the average resource per individual 
$\overline{V}$ (more precisely: the average droplet volume $\overline{V}$). 
\end{snugshade}

\section{The Competition: Motivation and Numerical Results}
\label{sec_mot}

The flow is modeled by the system
\begin{equation}\label{sys}
	{d\over {dt}} V_j = \sum_{i=1}^N c_{i,j}\, \Delta_s P_{i,j},\quad 1\leq  j\leq N
\end{equation}
for the unknown droplet volumes $V_j=V_j(t)$
together with the initial volume configuration
\begin{equation}\label{ini}
	V_j(0) = v_j, \quad 1\leq  j\leq N.
\end{equation}
The initial value problem \cref{sys,ini} has a local--in time solution for every $s>0$ (and every choice of initial data). Solutions are unique for every $s\geq 1$. In the case $0<s<1$, solutions might lose uniqueness at points where the right-hand side of \cref{sys} is not Lipschitz continuous.  Indeed, the occurrence of  non-unique solutions is readily confirmed when  $0<s<1$ and  $N=2$.  

Non-uniqueness is clearly an unphysical artifact of our model. Instead of excluding the case $0<s<1$ outright, we carry it along as a mathematical curiosity and address  some resultant challenges of analytical interest.

Let us give  an elementary characterization of the equilibria of system \cref{sys}. The connectedness of the conduit network immediately implies:
\begin{proposition}\label{stat}
 For\, ${\mathbf V}^* = \left(V^*_1,\ldots,V^*_N\right) \in {\mathbb R}^N$ to be an equilibrium of the system \cref{sys}, it is necessary and sufficient that
$
	P\left(V^*_i\right) = P\left(V^*_j\right)$,  $1\leq i, j\leq N.
$
\end{proposition}
\begin{snugshade}
\noindent {\em Socio-economic view}: At equilibrium the ``push and pull" by capillary pressures are the same for all competitors.
\end{snugshade}
For $0<s<1$, Lipschitz continuity of the right-hand side of \cref{sys} is lost at equilibria.

The model equations given by \cref{sys} are based on conservation of volume (mass). Indeed, we obtain immediately from the symmetry of the matrix $(c_{i,j})$:

\begin{proposition}\label{masscon}
	Any  solution ${\mathbf V}={\mathbf V}(t)=\left(V_1(t),\ldots,V_N(t)\right)$ of the system
	\cref{sys}  for $0\leq t\leq T$ satisfies
	$\sum_{j=1}^N V_j(t) =\sum_{j=1}^N V_j(0)$.
\end{proposition}
Hence the average droplet volume 
\begin{equation}
	\overline{V} = {1\over N}\,\sum_{j=1}^N V_j(t)
\end{equation}
is an invariant of the system.
Let us explicitly state the underlying
 (physically relevant)
\begin{equation}
	\text{\it{Standing Assumption}:}\qquad \overline{V}>0
\end{equation}
and define:

\begin{definition}
	For given $\overline{V}$, let
\begin{equation}
	{\mathbb V}\left(\overline{V}\right) = \left\{{\mathbf v} = (v_1,\ldots,v_N)\in {\mathbb R}^N\,\Biggr|\, {1\over N}\,\sum_{j=1}^N v_j = \overline{V}\right\}.
\end{equation}
\end{definition}
By \cref{masscon}, the set ${\mathbb V}\left(\overline{V}\right)$ is forward invariant in the sense that any solution ${\mathbf V}={\mathbf V}(t)$
of  \cref{sys,ini} for $0\leq t\leq T$ with initial data in ${\mathbb V}\left(\overline{V}\right)$ takes values  in ${\mathbb V}\left(\overline{V}\right)$ for $0\leq t \leq T$.

\begin{figure}[tbhp]
 \centering
\subfloat[$s=1.0$]{\label{s10}
    \includegraphics[width=0.48\textwidth]{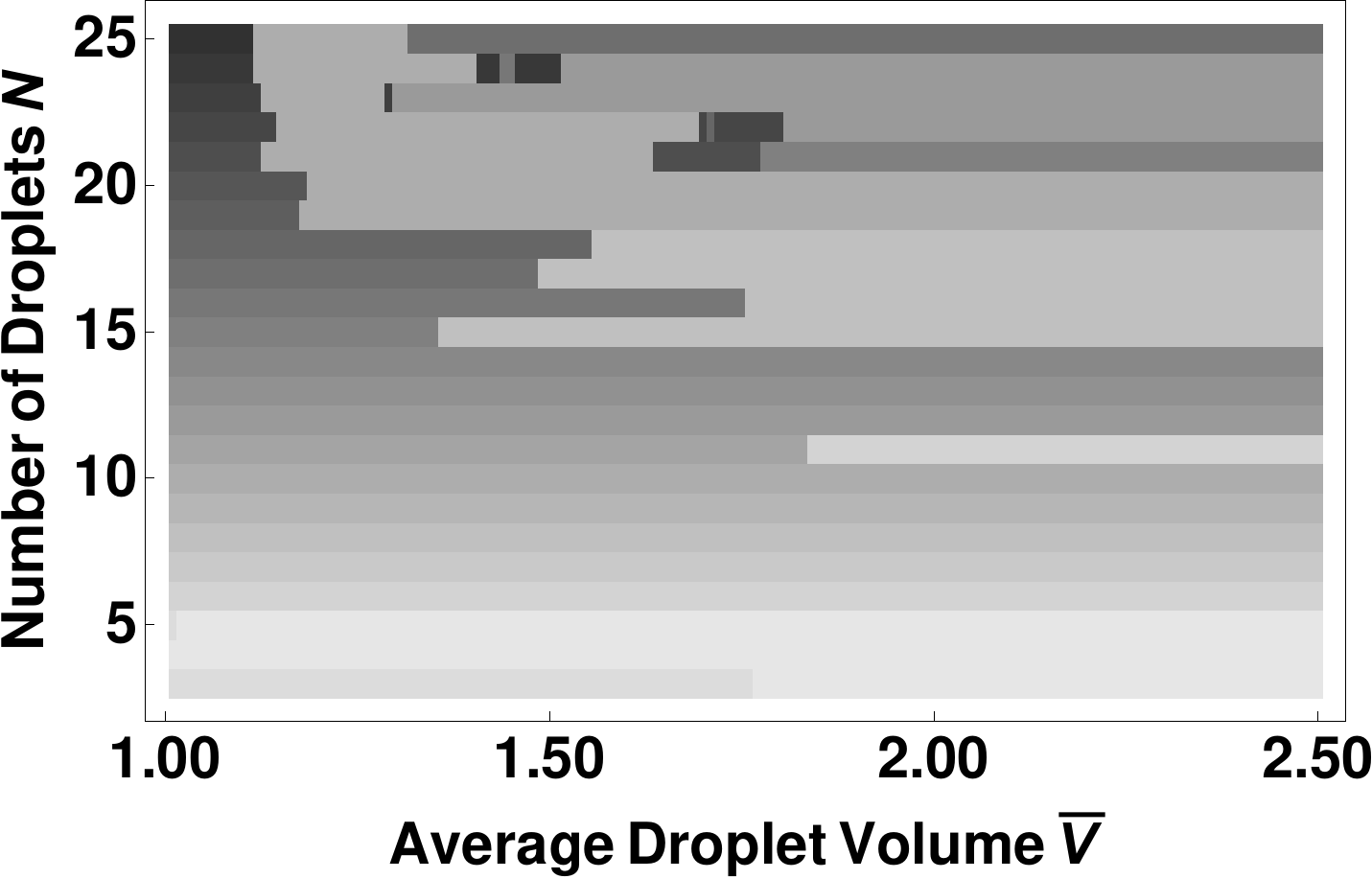}}
 \subfloat[$s=1.1$]{\label{s11}
    \includegraphics[width=0.48\textwidth]{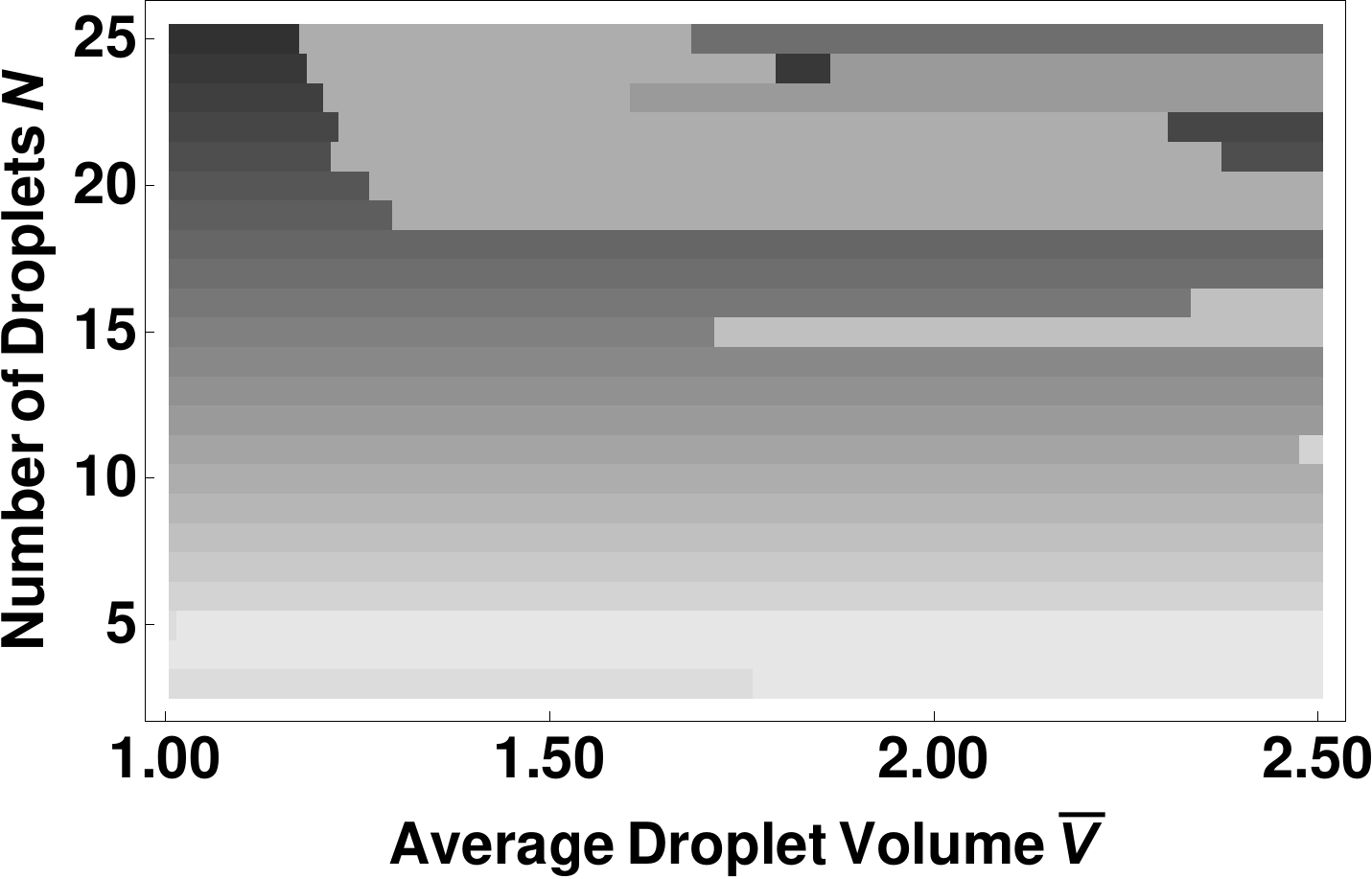}}\vspace*{-0.4cm}
 \newline
\centering
\subfloat{
    \includegraphics[width=0.9\textwidth]{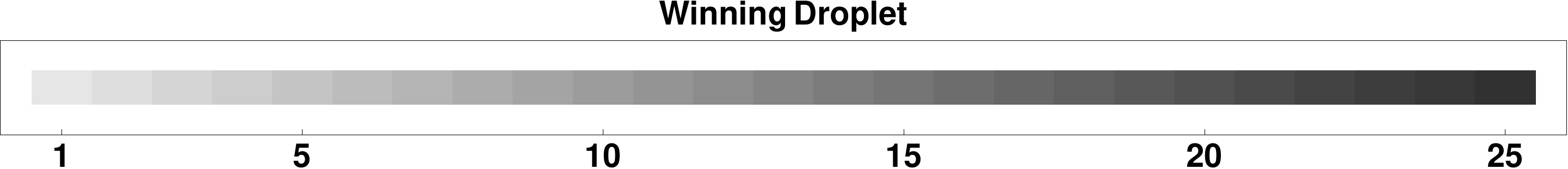}}
  \newline
\setcounter{subfigure}{2}
 \subfloat[$s=1.2$]{\label{s12}
      \includegraphics[width=0.48\textwidth]{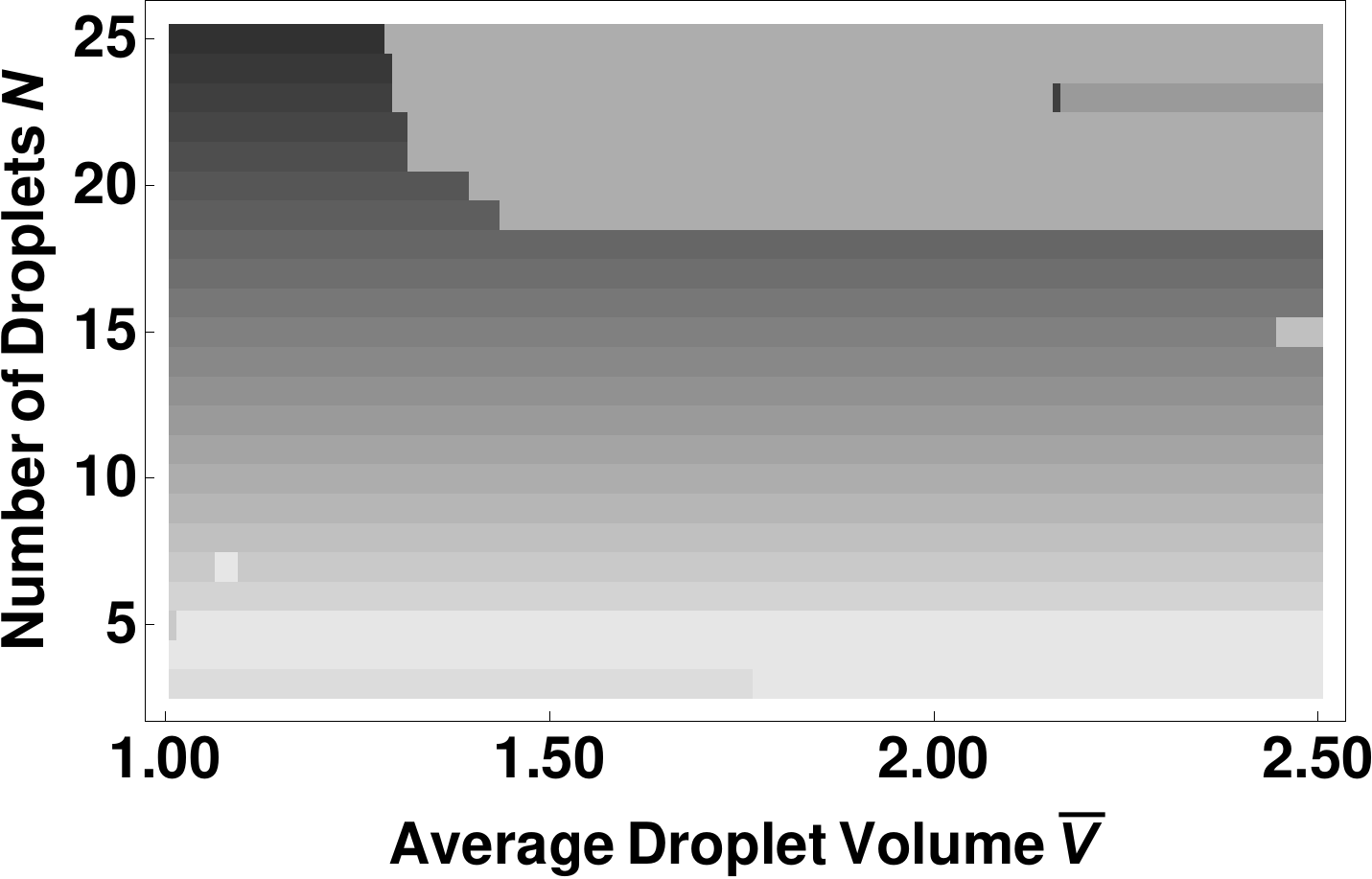}}
 \subfloat[$s=1.3$]{\label{s13}
      \includegraphics[width=0.48\textwidth]{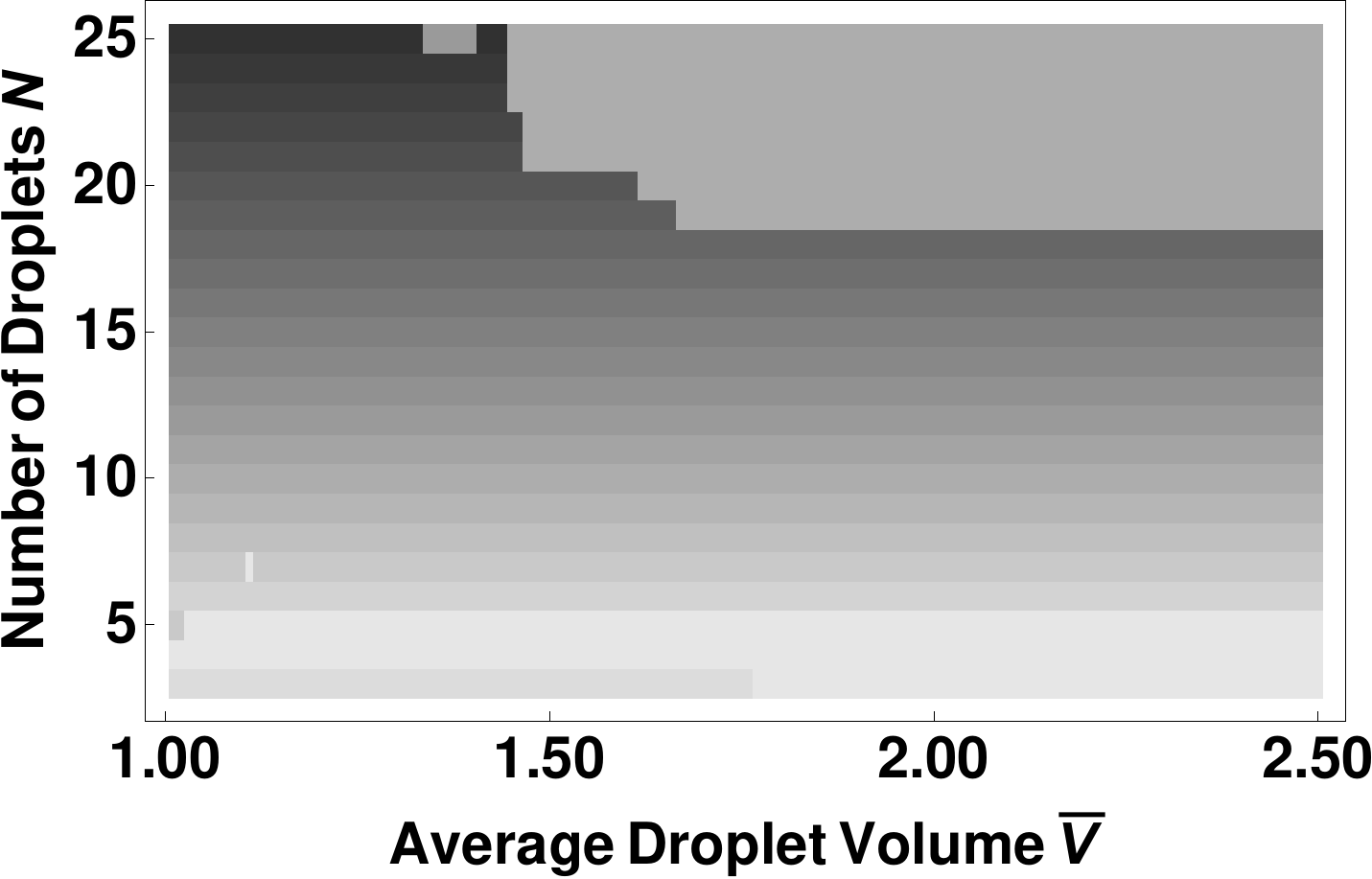}}
\caption{Winning droplets for (a) $s=1.0$, (b) $s=1.1$, (c) $s=1.2$, (d) $s=1.3$. Winning droplets with larger index $j$, $1\leq j\leq N$, are depicted by darker grays.}
\label{winners1}
\end{figure}
Having introduced the relevant parameters $N$, $\overline{V}$ and $s$, we now report some numerical experiments to document  their impact on the flow dynamics.
First we consider solutions of the system \cref{sys} with $1<\overline{V}\leq 2.5$ for a uniform linear network with $3\leq N\leq 25$ for various values of $s\geq 1$.  Specifically, we take $c_{i,i+1}=1=c_{i+1,i}$, $1\leq i\leq N-1$, and $c_{i,j}=0$ in all other cases.
As in \cite{LeEA1} we choose as initial data a small perturbation of the equilibrium $\left(\overline{V},\ldots,\overline{V}\right)$ consisting of droplets of equal size:
\begin{equation}
	v_j = \overline{V}+ 10^{-3}\,{{2\,j-(N+1)}\over {N-1}}, \quad 1\leq j\leq N.
\end{equation}
The long-term behavior of solutions is computed by using standard numerical integrators. We start with a value of $\overline{V}=1.01$ and increase it successively  to $\overline{V}=2.5$ in increments of $0.01$. Similarly to the Newtonian situation in \cite{LeEA1}, volume scavenging for $s>1$ leads to a depletion of volumes for all but one droplet that emerges as ``winner.'' In fact, $N-1$ small droplets remain with the winning droplet being large. As we will see later (and was shown in \cite{LeEA1} for the case $s=1$), this configuration corresponds to a stable equilibrium of the system. The winners of volume scavenging are displayed in  \cref{winners1} for $s=1.0$, $s=1.1$, $s=1.2$ and $s=1.3$. Winning droplets with larger index $j$, $1\leq j\leq N$, are depicted by darker grays. \Cref{s10} reproduces the Newtonian ``change-of-winner" result of \cite{LeEA1} with more details. 
\begin{snugshade}
\noindent {\em Physical mechanism}: The physical mechanism behind volume scavenging, independent of the particular fluid model,  was made plausible in \cite{SlSt} in the case $N=2$: When two droplets of equal volumes are large, adding a small volume to one droplet removes the same amount from the other one (mass conservation). While the pressure in the first droplet decreases, it increases in the second droplet, pushing even more volume toward the first droplet. Hence this configuration is unstable.
Likewise, when both droplets  are small of equal volumes less than 1, the configuration is stable. A similar behavior is also observed in  soap films \cite{Bo} and is reminiscent of the ``two-balloon experiment" \cite{DrEA,WeBa}.
\end{snugshade}
\noindent When $N>2$, the situation is more complicated as we will demonstrate later in this work. Animations of volume scavenging can be found online at
\url{https://www.dropbox.com/s/0t4ytk5jkx5uxd0/VAllLog.mpg?dl=0} with a legend given at 
 \url{https://www.dropbox.com/s/8teepdlybd2j3op/Legend.pdf?dl=0}.

\begin{snugshade}
\noindent {\em Computational aside}: In our computation we have noticed  a change in the dynamics with respect to the power-law parameter $s$: The transient times to equilibrium increase by several orders of magnitude as $s$ increases. This phenomenon occurs as the apparent viscosity (friction) of the fluid decreases with increasing shear rates. 
Hence the longer transients are surprising, at least at first sight, since, na\"{i}vely, we anticipate a shear-thinning fluid ($s>1$) to become ``thinner'' than a Newtonian one  ($s=1$). However, on second thought, one realizes that the apparent viscosity of shear-thinning fluids actually increases when shear rates become smaller. In volume scavenging an increase in the apparent viscosity occurs when pressure drops are small. Hence in low-pressure drop regions a shear-thinning fluid   is expected to flow slower than a Newtonian one. This slow-down manifests itself in the flow dynamics, especially near equilibria. 
To rationalize the observed transient behavior further, we refer to the similarity with the  ``toy problem'' $x^\prime = -x^s$, $x(0)=x_0>0$. The solution $x(t)$ decays to the equilibrium $0$ for every $s\geq 1$. The decay is exponential when $s=1$. For $s>1$, however,  the decay is only algebraic like $t^{-1/(s-1)}$. 
\end{snugshade}

\begin{figure}[tbhp]
\centering
\subfloat[$s=2.0$]{ \label{graph20}
    \includegraphics[width=0.48\textwidth]{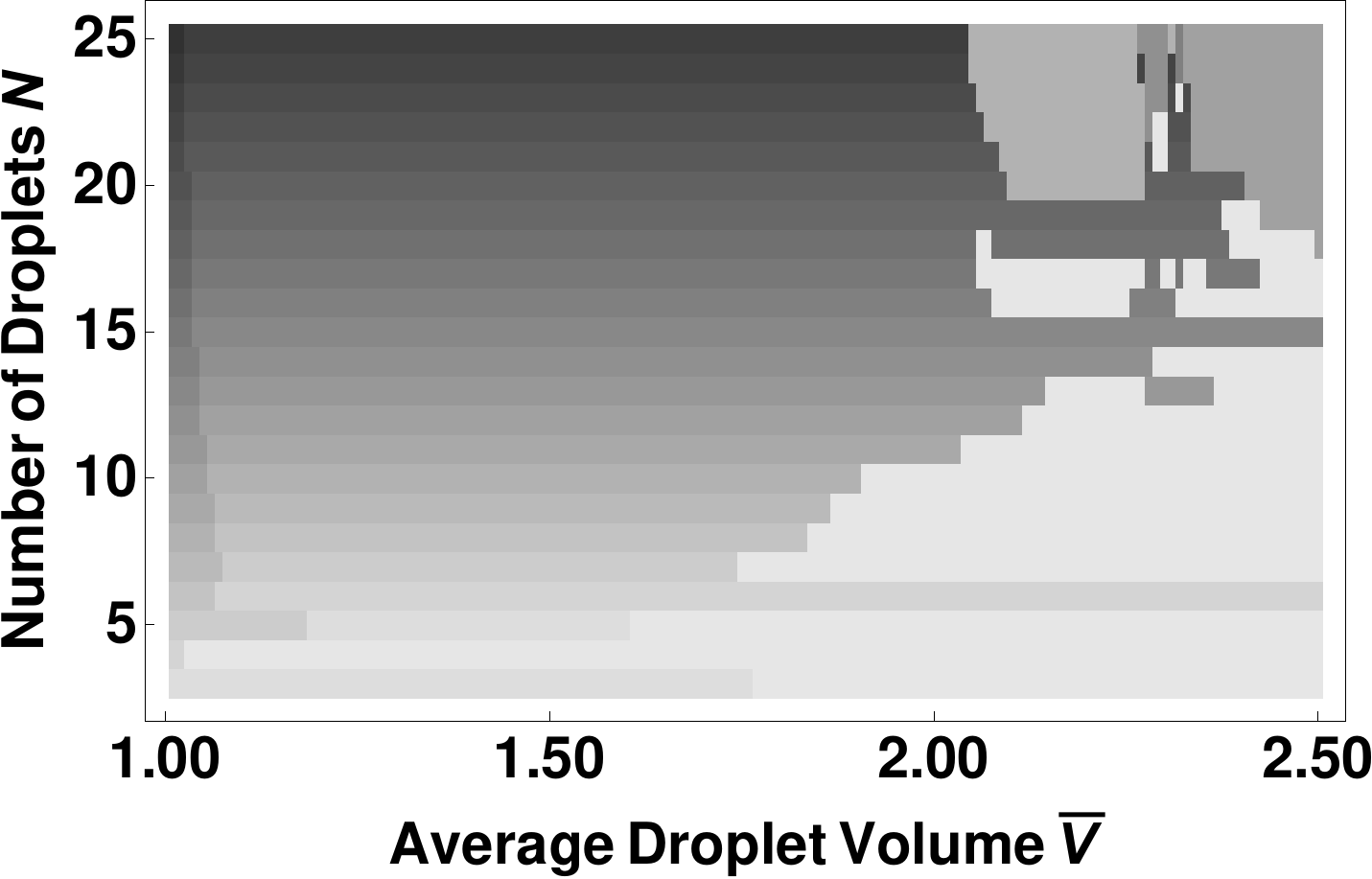}}
 \subfloat[$s=3.0$]{  \label{graph30}
    \includegraphics[width=0.48\textwidth]{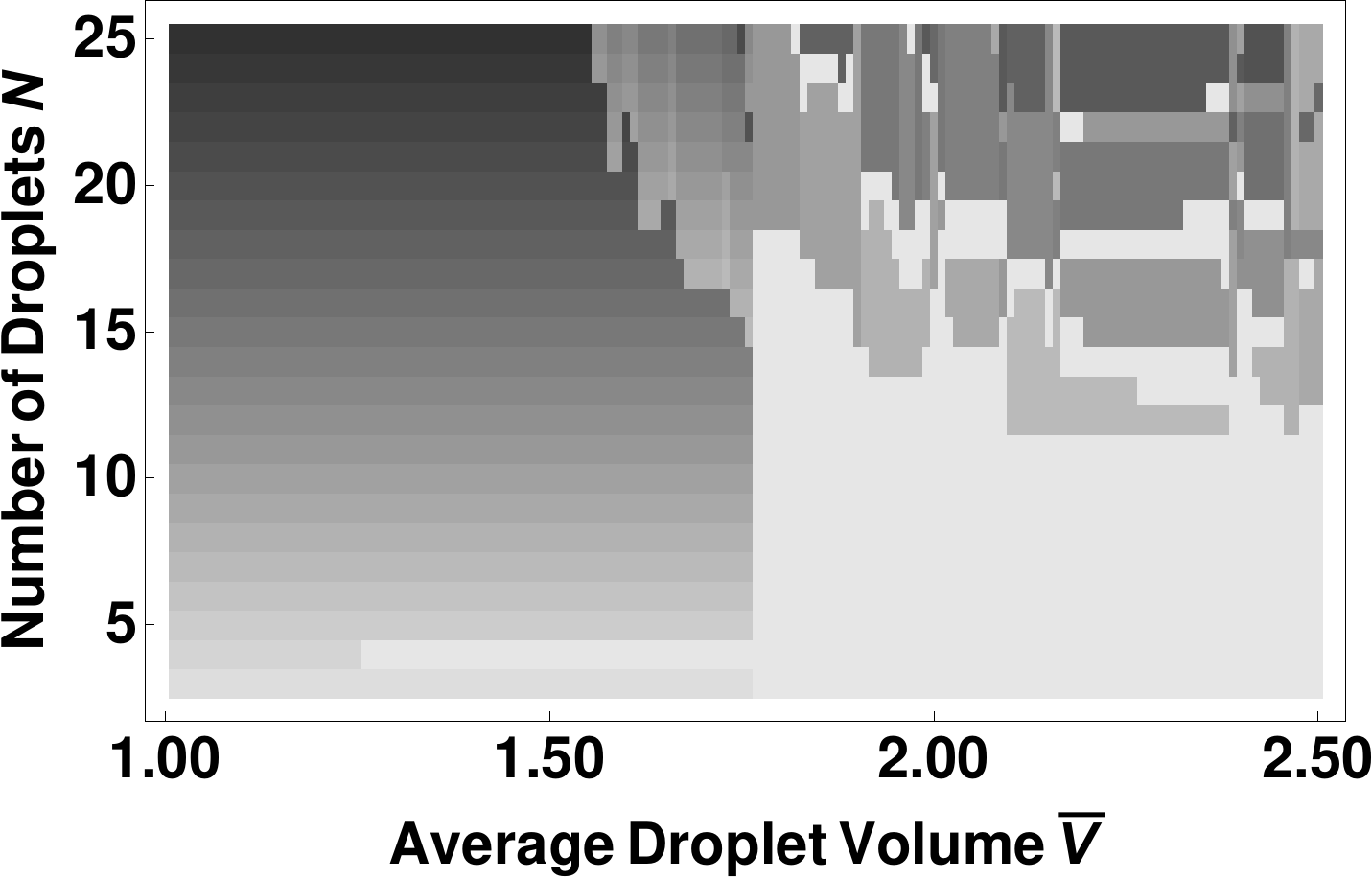}}\vspace*{-0.4cm}
 \newline
\subfloat{\includegraphics[width=0.9\textwidth]{winner.pdf}}
\caption{Winning droplets for (a) $s=2.0$, (b) $s=3.0$}
\label{winners2}
\end{figure}
\Cref{winners1} shows that even small variations in $s$ lead to notable changes in the winning droplet configuration. This fact is particularly visible when $N\geq 15$. It appears  that the basins of attraction of the corresponding stable equilibria not only vary drastically with $N$ and $\overline{V}$ (as indicated in \cite{LeEA1}), but also with $s$. These changes, it seems,  become    more pronounced for larger values of $s$ as seen in \cref{winners2}. The exact mechanisms underlying the selection of winning droplets remain largely unknown, although some  explanations to this end are given in \cite{LeEA1} for networks with $N=3$ and $s=1$.
\begin{snugshade}
\noindent {\em Engineering aside}:  As suggested by our model, it might be possible to identify (or narrow down)  an unknown power-law parameter $s$ from a small number of
change-of-winner data. In this way such information could prove useful  to replace shear stress versus shear rate measurements  in conventional rheometers.
 \end{snugshade}
\begin{figure}[tbhp]
\centering
\subfloat[$n=0$]{\label{n0}
    \includegraphics[width=0.45\textwidth]{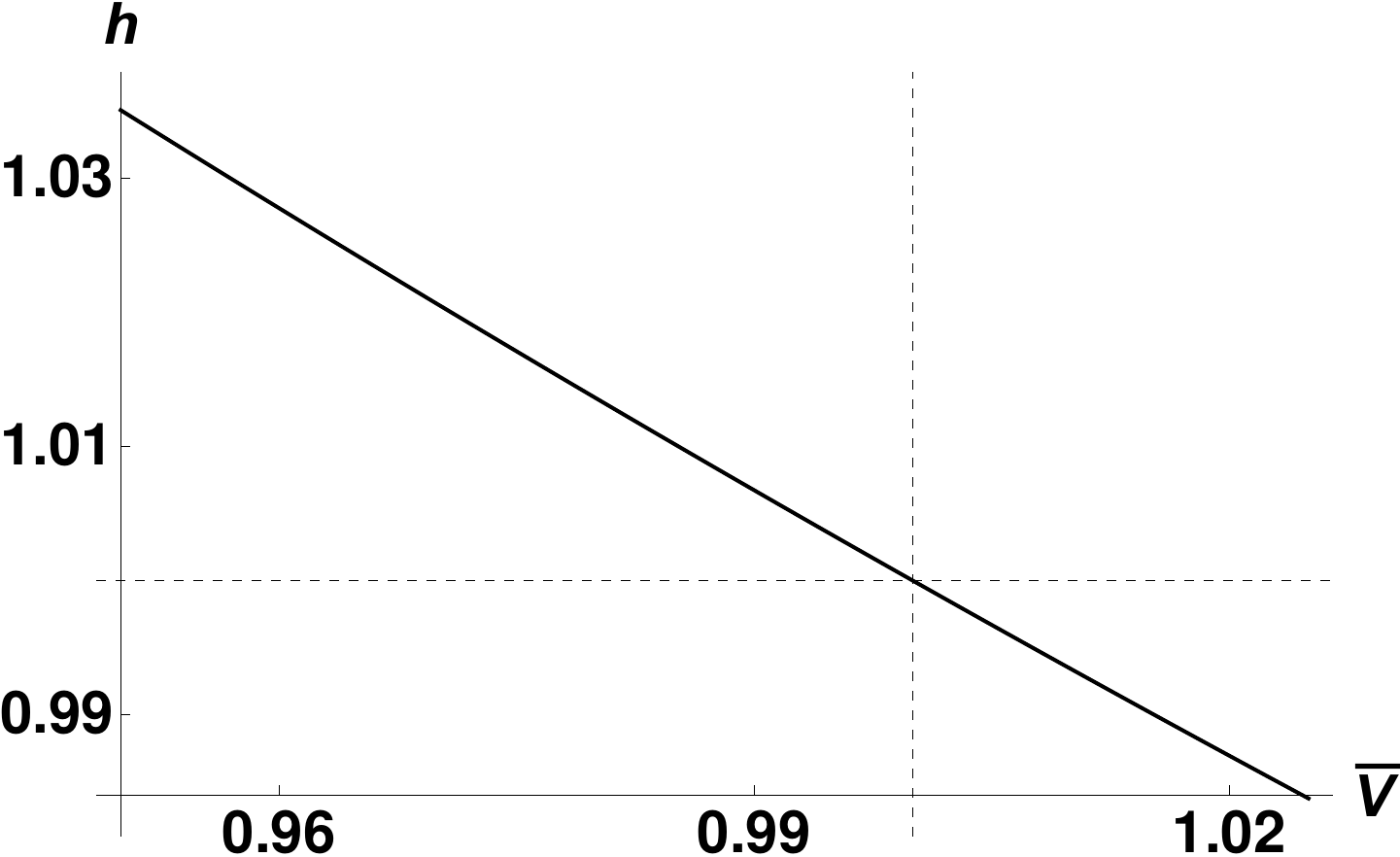}}
\subfloat[$n=1$]{\label{n1}
    \includegraphics[width=0.45\textwidth]{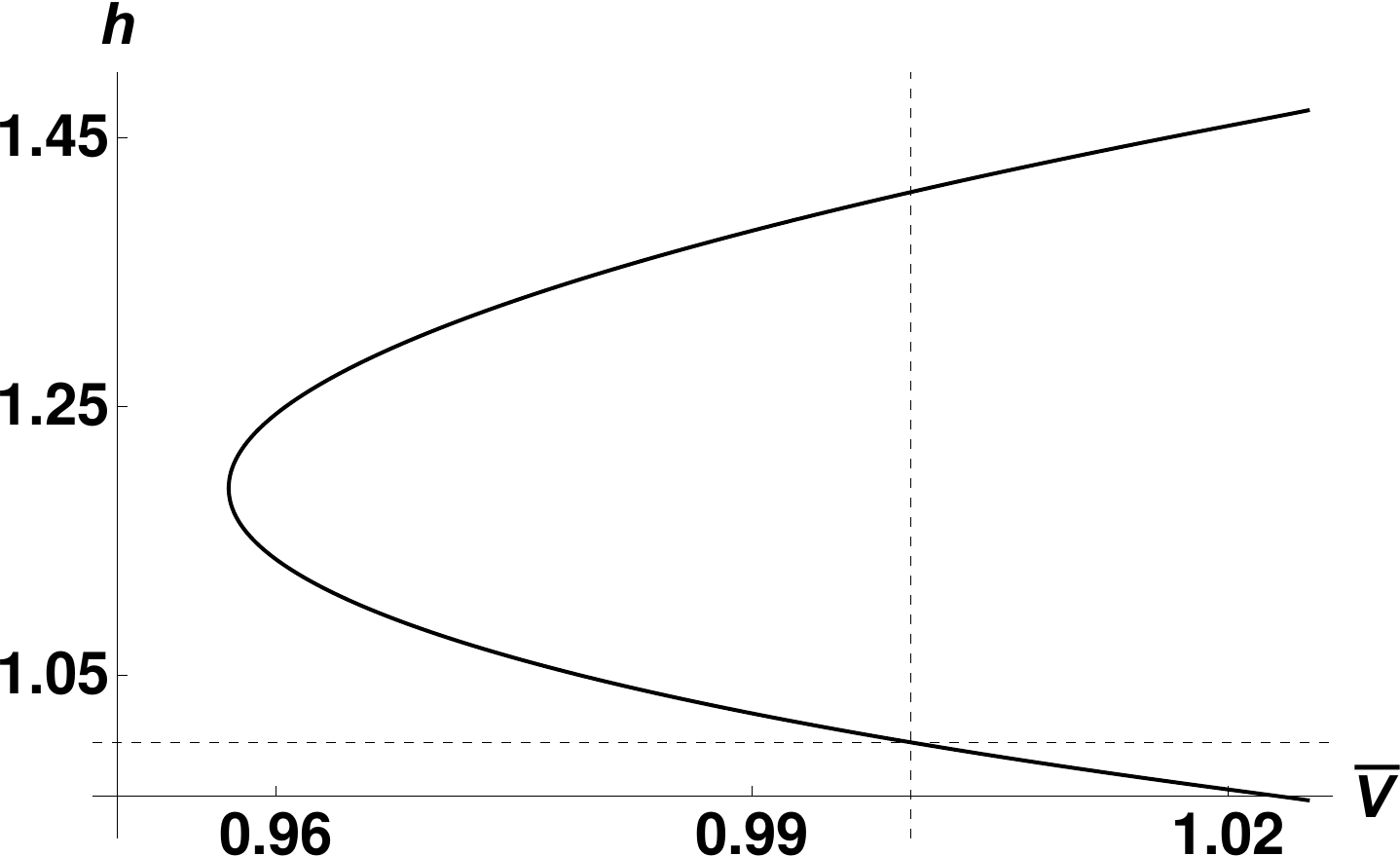}}
\newline
\subfloat[$n=2$]{\label{n2}
      \includegraphics[width=0.45\textwidth]{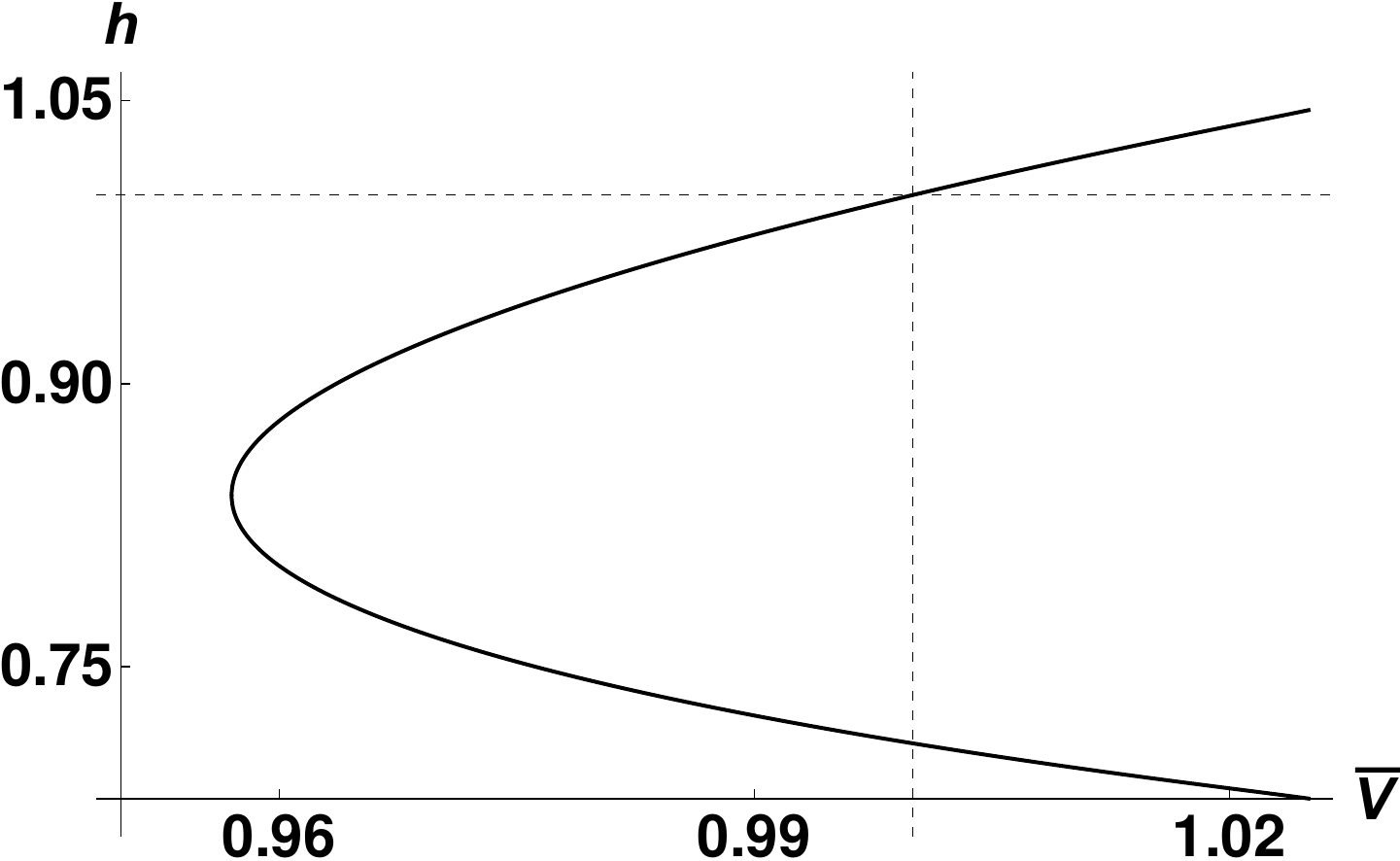}}
\subfloat[$n=3$]{\label{n3}
      \includegraphics[width=0.45\textwidth]{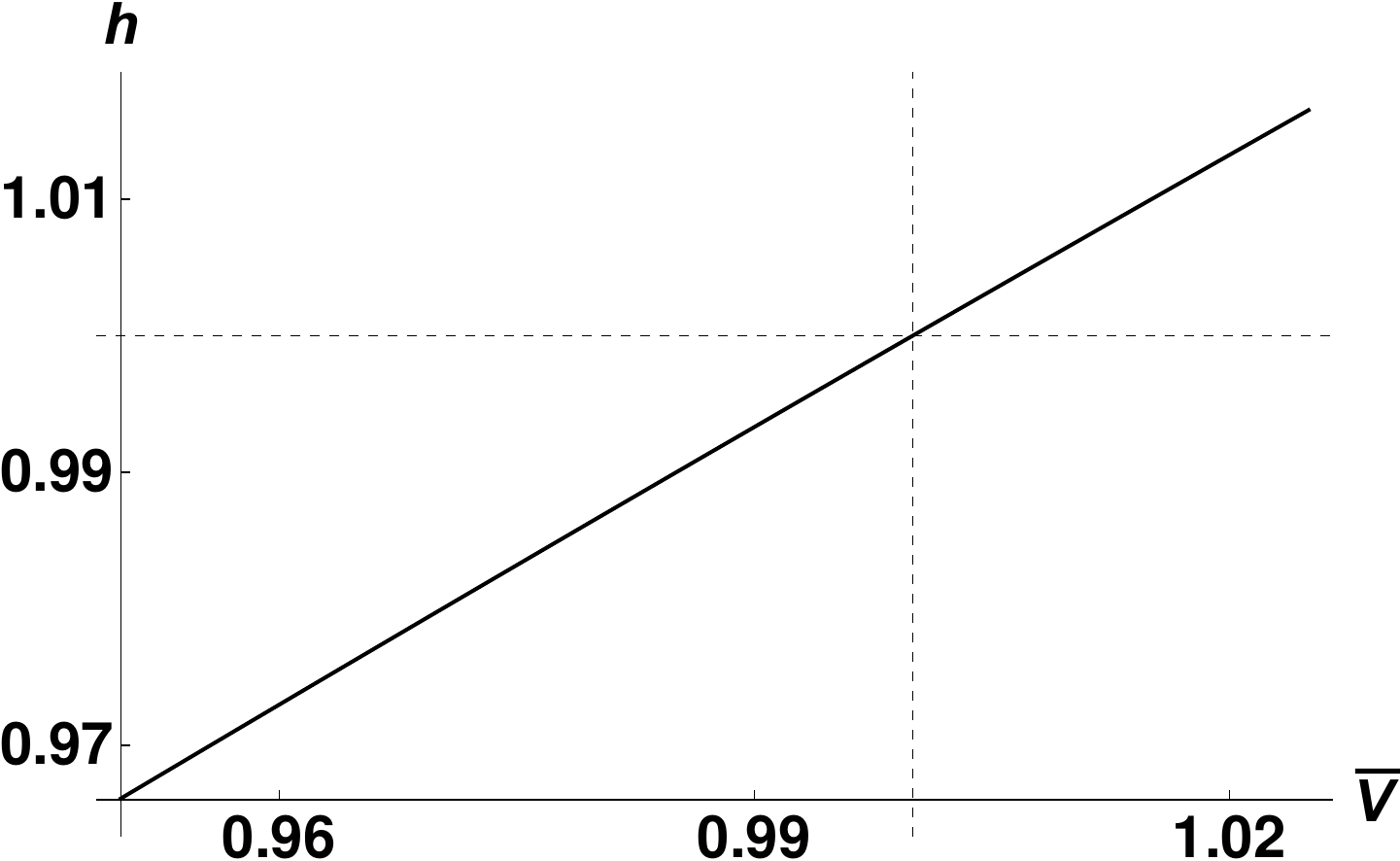}}
\caption{Positive solutions $h$ of Equation \cref{n3eq} in dependence on $\overline{V}$}
\label{zerosh}\vspace*{-0.4cm}
\end{figure}
Our second concern is to  highlight possible issues with the occurrence of equilibria in dependence on $\overline{V}>0$. We concentrate on the case $N=3$. In light of \cref{stat}, any equilibrium of the system \cref{sys} must consist of $n$ ($0\leq n\leq 3$) large droplets of height $h>1$ and $3-n$ small droplets of height $h^{-1}$ such that the capillary pressure in all droplets is equal. By mass conservation, for a given average droplet volume $\overline{V}>0$ we have
\begin{equation}
	n\,v(h) + (3-n)\,v\left(h^{-1}\right) = 3\,\overline{V},
\end{equation}
or equivalently,
\begin{equation}\label{n3eq}
	n\,h^6+3\,n\,h^4-12\,\overline{V}\,h^3+3\,(3-n)\,h^2+3-n=0.
\end{equation}
\Cref{zerosh} shows the positive solutions $h$ of this equation for varying values of $\overline{V}$ with $n\in\{0, 1, 2, 3\}$. We focus  on solutions  near $h=1$, $\overline{V}=1$ (intersection point of dashed lines). For given $n$, the solutions of interest are those with $h>1$. As expected, in \cref{n0,n3} we encounter equilibria consisting of droplets of the same size. They are the only equilibria which arise for $0<\overline{V}<1$ with $n=0$ (three small droplets of height $h^{-1}$ and zero large droplets) and for $\overline{V}>1$ with $n=3$ (three  large droplets of height $h>1$ and zero small droplets). As seen in \cref{n2}, equilibria with exactly  two large droplets occur for $\overline{V}>1$, but not for $0<\overline{V}<1$. The most interesting case is $n=1$. Here, \cref{n1} confirms  that for $\overline{V}>1$, we have equilibria consisting of exactly one large droplet (and two small ones). If, however,  $\overline{V}$ falls in the  interval $(0.957,1)$, two types of equilibria consisting of exactly one  large  and two  small droplets  exist. For smaller values of $\overline{V}$, there are no equilibria with $n=1$. We illustrate all possible equilibria (up to permutation of the droplets) with one large and two small droplets for $\overline{V}=0.98$ by volume percentages of the total droplet volume $3\,\overline{V} = 2.94$ in \cref{pies}. 
\begin{figure}[tbhp]\vspace*{-0.7cm}
\centering
 \subfloat[Volumes $V_L$, $V_S$, $V_S$]{\label{LSS}
    \includegraphics[width=0.32\textwidth]{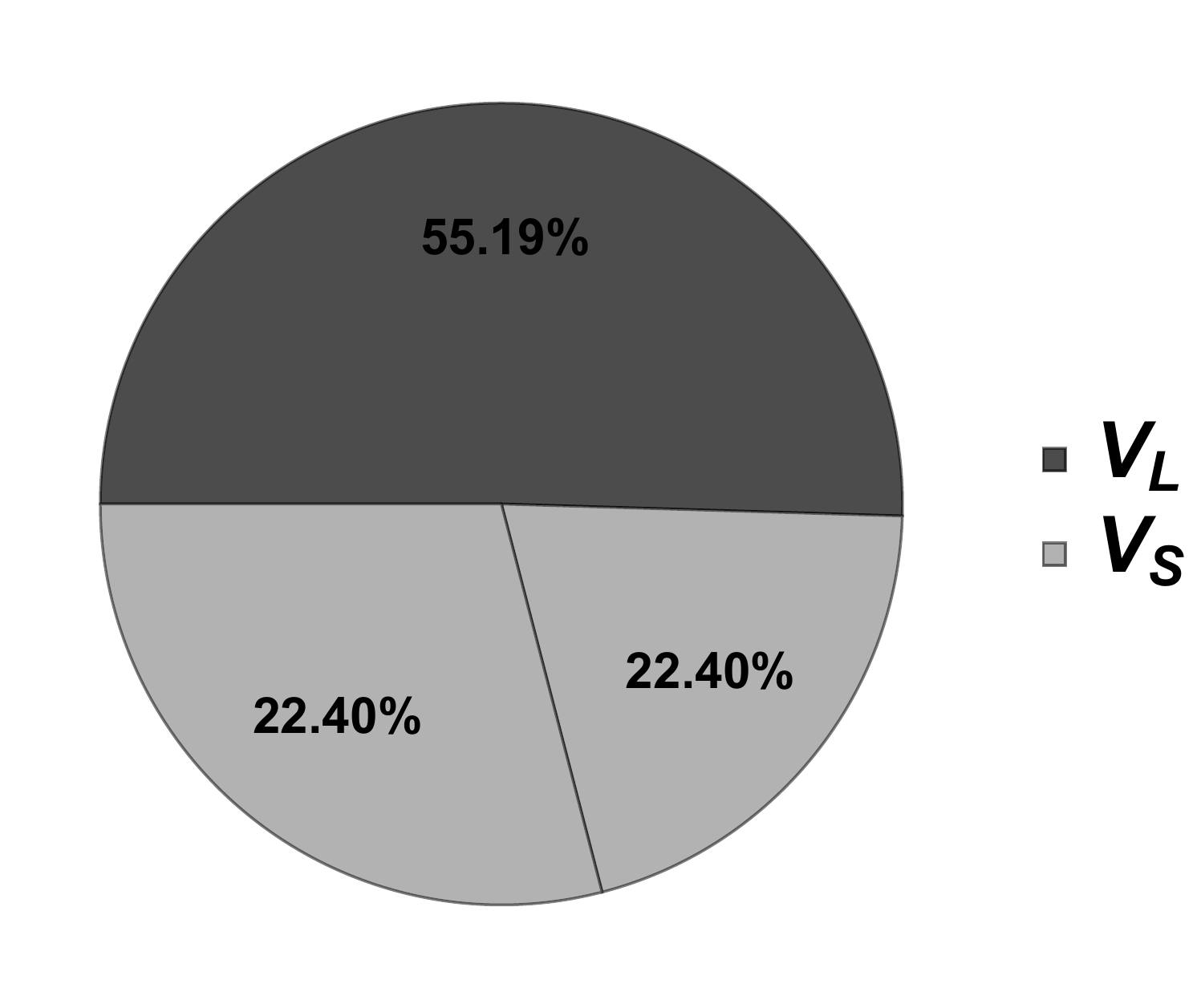}}
\subfloat[Volumes $V_l$, $V_s$, $V_s$]{\label{lss}
    \includegraphics[width=0.32\textwidth]{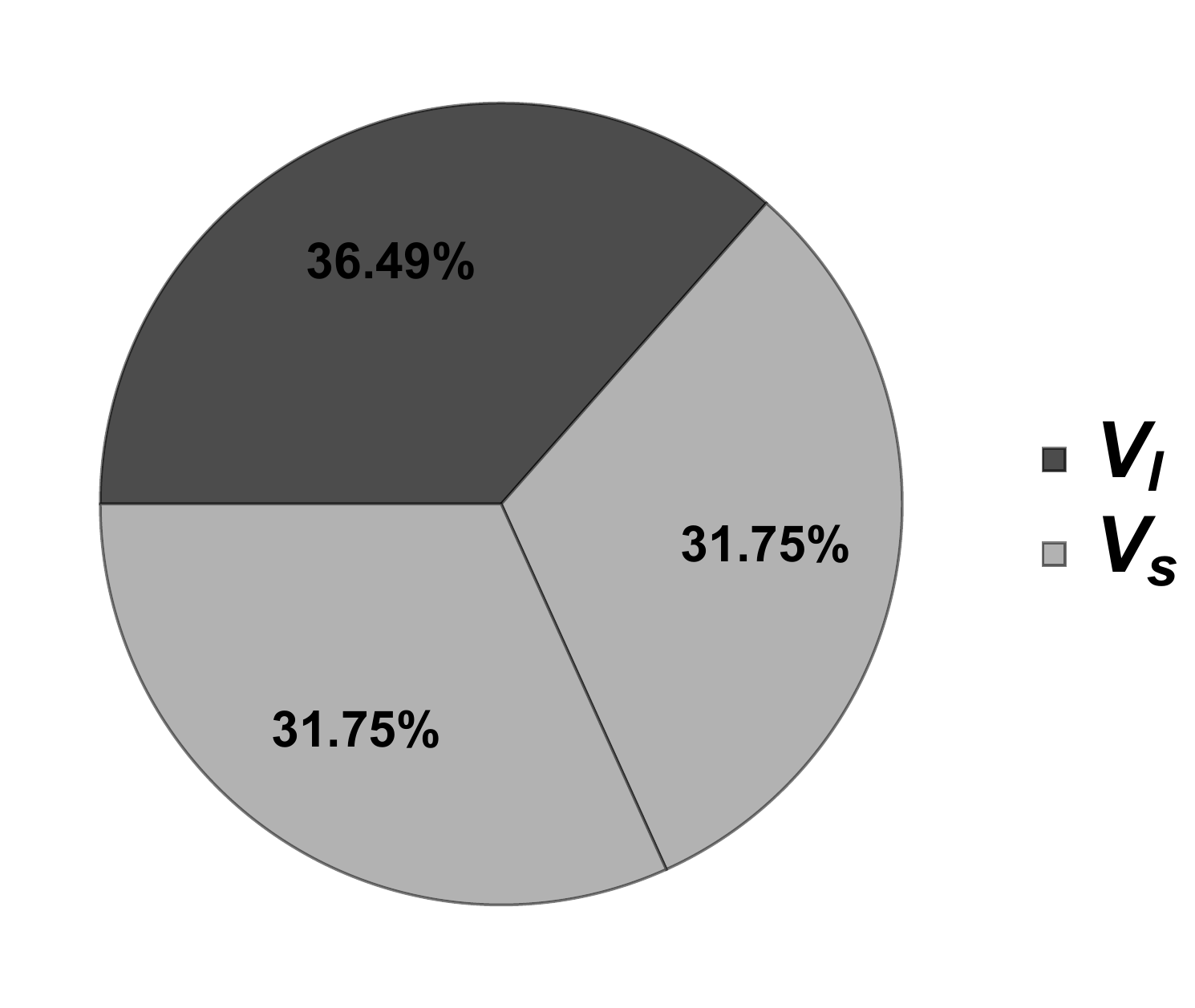}}
  \subfloat[Volumes $\overline{V}$, $\overline{V}$, $\overline{V}$]{\label{sss}
      \includegraphics[width=0.32\textwidth]{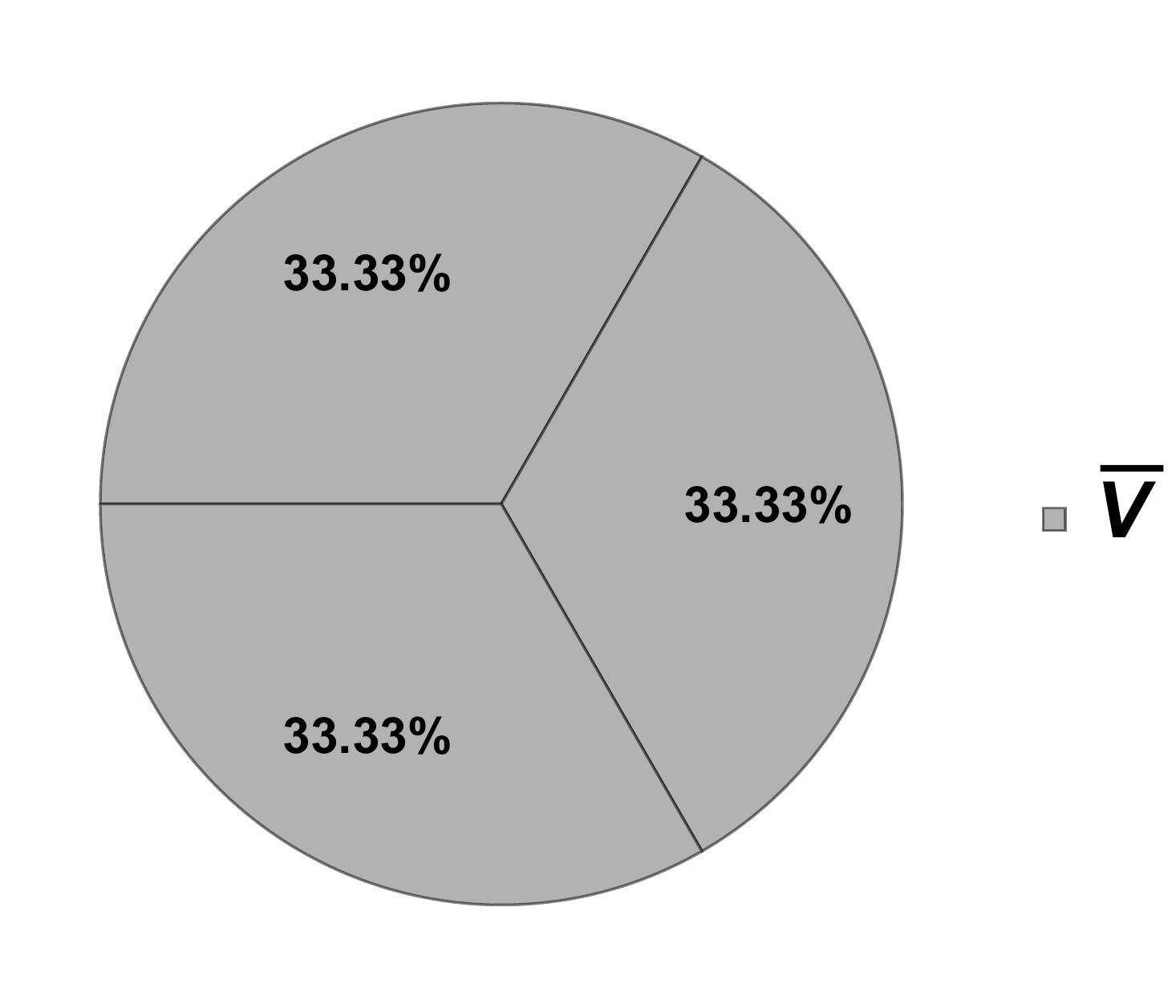}}
\caption{Volume percentages (rounded) for equilibria with $N=3$, $\overline{V}=0.98$}
\label{pies}
\end{figure}
Here, the two possible heights of the large spherical-cap droplet are $h_L = 1.348$ and $h_l = 1.047$ with the height of the two small droplets given by the reciprocals $h_S=h_L^{-1}$ and $h_s=h_l^{-1}$, respectively. The corresponding volumes are denoted by $V_L=v(h_L)$, $V_S=v(h_S)$, $V_l=v(h_l)$ and $V_s=v(h_s)$. The only other possible equilibrium consists of three equal, small droplets of volume $\overline{V}$.
We will confirm in our work that for $\overline{V}>1$ only one type of stable equilibrium is present, as seen in the Newtonian case $s=1$  \cite{LeEA1} and  in the inviscid, symmetric case   \cite{SlSt}. If $\overline{V}$ is restricted to an interval of the form $l_0<\overline{V}<1$,  two types of stable equilibria will be shown to arise side-by-side, resembling the inviscid case in \cite{SlSt}. The limiting value $l_0>0$ will turn out to depend on $N$. In the situation discussed above ($N=3$)  $l_0\approx 0.957$, and the stable equilibria will be the ones described by \cref{LSS,sss}. If, however, $0<\overline{V}<l_0$, the only possible equilibrium will be the one consisting of droplets of equal size.  This equilibrium is stable.
\begin{snugshade}
\noindent {\em Socio-economic view}:  The parameter ranges $\overline{V}>1$ (``abundant resource") and $0<\overline{V}<l_0$ (``scarcest resource") correspond to the winner-take-all and egalitarian outcomes, respectively. The parameter range $l_0<\overline{V}<1$ (``scarce resource")  exhibits both possible outcomes in dependence on  initial conditions.
\end{snugshade}

We will study our physics-based model based on analytical techniques that are independent of the underlying network and of the power-law parameter $s$.
A key tool for our study of the system \cref{sys} will be provided by the following energy functional:
\begin{equation}\label{defW}
	{\mathcal W}({\mathbf V}) = {1\over N}\,\sum_{j=1}^N \int_0^{V_j} P(V)\,dV\quad \text{for ${\mathbf V} =  \left(V_1,\ldots,V_N\right)\in {\mathbb R}^N$.}
\end{equation}
$\mathcal W$ is a minor modification of the equivalent functional introduced in \cite{LeEA1}. It is a measure of the pressure-volume work due to the capillary  pressure acting on each spher\-ical-cap droplet.
\begin{snugshade}
\noindent {\em Socio-economic view}: The objective functional $\mathcal W$ is the total cost of resource accumulation on the level of the population, which reflects the aggregate drive to accumulate.
\end{snugshade}

Since for any $V_0\in \mathbb R$ and $h_0 = v^{-1}\left(V_0\right)$
\begin{equation}\label{intP}
	\int_0^{V_0} P(V)\,dV = \int_0^{h_0} p(h)\,v^\prime(h)\,dh = 3\,\int_0^{h_0} h\,dh={3\over 2}\,h_0^2,
\end{equation}
by Equation \cref{Adefn}, $\mathcal W$ is related to the total surface area $A({\mathbf V})=\sum_{j=1}^N S(V_j)$ through
\begin{equation}
	{\mathcal W}({\mathbf V}) = {1\over N}\,A({\mathbf V})-{3\over 2}.
\end{equation}
Let us record:
\begin{proposition}\label{propA}
${\mathcal W}({\mathbf V})\geq 0$ for all\, ${\mathbf V} \in {\mathbb R}^N$, and ${\mathcal W}({\mathbf V})\rightarrow \infty$ if and only if\, $\|{\mathbf V}\|\rightarrow \infty$.
\end{proposition}
The following result states the key  property of $\mathcal W$ for every $s>0$:
\begin{lemma}\label{Ader}
	Suppose\, ${\mathbf V} = {\mathbf V}(t) =  \left(V_1(t),\ldots,V_N(t)\right)$ is any solution of the system \cref{sys} on an open interval $I\subset \mathbb R$. Then
\begin{equation}
	{d\over {dt}}\, {\mathcal W}({\mathbf V}(t))=-{1\over {2\,N}}\,\sum_{i,j=1}^N c_{i,j}\,\left|P(V_i(t))-P(V_j(t))\right|^{s+1}\quad \text{for all $t\in I$.}
\end{equation}
In particular,
$
	{d\over {dt}}\, {\mathcal W}({\mathbf V}(t))\leq 0$ for all $t\in I$,
and
$
	{d\over {dt}}\, {\mathcal W}({\mathbf V}(t_0))= 0$  for some $t_0\in I$
if and only if \,  ${\mathbf V}(t_0)$ is an equilibrium of the system \cref{sys}.
\end{lemma}
\begin{snugshade}
\noindent {\em Socio-economic view}: Lowering the total cost $\mathcal W$ drives the dynamics of the competition.
\end{snugshade}
\begin{proof}
By direct calculation we obtain
\begin{align}
	{d\over {dt}}\, {\mathcal W}({\mathbf V})&= {1\over {2\,N}}\,\sum_{j=1}^N P(V_j)\,V_j^\prime+{1\over {2\,N}}\,\sum_{i=1}^N P(V_i)\,V_i^\prime\\
&= {1\over {2\,N}}\,\sum_{i,j=1}^N c_{i,j}\,  \left(P(V_j)\,\Delta_s P_{i,j} + P(V_i)\,\Delta_s P_{j,i}\right).
\end{align}
This implies the claim.
\end{proof}
Hence  solutions stay bounded on ``any'' time interval of existence starting at $t=0$. Consequently,
we obtain a global existence and uniqueness result.

\begin{proposition}\label{gloex}
Solutions of the initial value problem \cref{sys,ini} exist for all time $t\geq 0$. They are unique if $s\geq 1$.
\end{proposition}

\section{Preserved During Competition: Forward Invariant Sets}

\label{sec_equ}

We have seen in  \cref{masscon} that, by conservation of mass,  the set ${\mathbb V}\left(\overline{V}\right)$ is forward invariant. In this section we study two more, physically relevant sets.

\begin{definition}
	For given $\overline{V}$ and $k\in\mathbb N$ with $1\leq k\leq N$,   define the sets
$
	{\mathbb V}_+\left(\overline{V}\right) ={\mathbb V}\left(\overline{V}\right)\cap [0,\infty)^N$ and $
{{\mathbb V}}^k_+\left(\overline{V}\right) = \left\{{\mathbf v} = (v_1,\ldots,v_N)\in {\mathbb V}_+(\overline{V})\,\Biggr|\, v_k\leq 1\right\}$.
\end{definition}
Because  of \cref{stat,masscon} we readily have:
\begin{proposition}\label{VV+}
	Let ${\mathbf V}^* = \left(V^*_1,\dots, V^*_N\right)$ be an equilibrium of the system \cref{sys} with $\overline{V} = {1\over N}\,\sum_{j=1}^N V^*_j >0$. Then ${\mathbf V}^*$  is contained in  ${\mathbb V}_+\left(\overline{V}\right)$. In fact, $V^*_j>0$ for all $1\leq j\leq N$.
\end{proposition}

Next we will show that the set ${\mathbb V}_+(\overline{V})$ is forward invariant when $s\geq 1$.  To obtain this result, we choose $\epsilon>0$ and let
\begin{equation}\label{adeps}
	c^\epsilon_{i,j} = \left\{\begin{matrix} c_{i,j} & \text{if $c_{i,j}>0$,}\\
							\epsilon & \text{if $c_{i,j} = 0$ and $i\not=j$,}\\
							0&\text{otherwise.} \end{matrix}\right.
\end{equation}

\begin{theorem}\label{invariant}
	Suppose that $s\geq 1$, or that $s>0$ and the network is complete. Then
	any solution  ${\mathbf V}={\mathbf V}(t)$ of the system \cref{sys} with ${\mathbf V}(0)\in {\mathbb V}_+(\overline{V})$ satisfies
$
	 {\mathbf V}(t)\in {\mathbb V}_+(\overline{V})$   for all $t\geq  0$.
\end{theorem}
\begin{proof}
Let ${\mathbf V}={\mathbf V}(t)=\left(V_1(t),\ldots,V_N(t)\right)$ be a solution originating in ${\mathbb V}_+(\overline{V})$, and let us assume for the moment that the network is complete. It suffices to show that $V_j(t)\geq 0$ for all $t\geq 0$ and $1\leq j\leq N$. The map $V\mapsto P(V)$ is positive for $V>0$, increasing on $(0,1)$ and decreasing on $(1,\infty)$. It  satisfies $P(0) = \lim_{V\rightarrow \infty} P(V) = 0$.
 Suppose now that $t_0\geq 0$ is a fixed time and $j\in \{1,\ldots, N\}$ an index  such that
$
	0= V_j(t_0)\leq V_i(t_0)$,  $1\leq i\leq N$.
Since $\overline{V}>0$,
there is $k\in \{1,\ldots, N\}$ such that
$V_k(t_0)>0$. Consequently,
$0 = P(V_j(t_0))< P(V_k(t_0))$ and $0=P(V_j(t_0))\leq P(V_i(t_0))$, $1\leq i\leq N$.
Hence we have from \cref{defDs} at $t_0$ that
$ \Delta_s P_{kj}>0$ and $\Delta_s P_{i,j}\geq 0$ for $1\leq i\leq N$.
Completeness of the network gives then
${d\over {dt}}\,V_j(t_0)> 0$.
However, this result implies that
$V_j(t)\geq 0$ for $t\geq 0$,  $1\leq j\leq N$. 
Together with \cref{masscon}, this proves the claim for the case of a complete network with any $s>0$.

Let us now turn to the case $s\geq 1$. For ${\mathbf v}\in {\mathbb V}_+(\overline{V})$ and $\epsilon>0$, denote the solution of the system \cref{sys} with initial data ${\mathbf V}^\epsilon(0)= {\mathbf v}$  and with adjacency matrix $\left(c^\epsilon_{i,j}\right)$, given in
\cref{adeps}, by ${\mathbf V}^\epsilon={\mathbf V}^\epsilon (t) = \left(V^\epsilon_1(t),\ldots, V^\epsilon_N(t)\right)$. The corresponding network is complete. Hence our argument above applies and proves that
$V^\epsilon_j(t)\geq 0$ for $t\geq  0$, $\leq j\leq N$.
As we pass to the limit $\epsilon\rightarrow 0^+$, continuous dependence of solutions on the parameter $\epsilon$ shows that
the pointwise limit ${\mathbf V}(t) = \lim_{\epsilon\rightarrow 0^+} {\mathbf V}^\epsilon(t)$ exists for every $t\geq 0$ and solves the initial value problem \cref{sys,ini} with initial data ${\mathbf V}(0) = \mathbf v$ and with the original adjacency matrix $\left(c_{i,j}\right)$. Hence
$V_j(t)\geq 0$ for $t\geq 0$,  $1\leq j\leq N$.
\end{proof}

Finally let us turn to a physically important observation:

\begin{theorem}\label{invariant2}
	Suppose  that $s\geq 1$ and $k\in\{1,\ldots, N\}$. Then any solution ${\mathbf V}={\mathbf V}(t)$   of the system \cref{sys} with ${\mathbf V}(0)\in {\mathbb V}^k_+(\overline{V})$ satisfies
$
	 {\mathbf V}(t)\in {\mathbb V}^k_+(\overline{V})$  for all $t\geq  0$.
\end{theorem}
Hence for $s\geq 1$,  a small droplet will remain small.
\begin{snugshade}
\noindent {\em Socio-economic view}:  \Cref{invariant2} describes the ``once down-and-out, always down-and-out" outcome: When a competitor's resource falls below $1$, that competitor cannot recover. Hence upward mobility is impossible for individuals whose resource shrinks below this threshold.
\end{snugshade}
\begin{proof}
Suppose first that the network is complete. Given $k$, we may assume that there exists $t_0\geq 0$ and $j\in\{1,\ldots,N\}$ such that
$V_k(t_0)=1$  and $V_j(t_0)\not=1$.
 For otherwise, we would either have $V_k(t)<1$ for all $t\geq 0$, or $V_i(t) =1$, $1\leq i\leq N$, whenever $V_k(t) = 1$. The latter situation implies that $V_i(t)=1$, $1\leq i\leq N$,  for all $t\geq 0$ by uniqueness of solutions. In either case, the claim would follow.
Since for every $V>0$ with $V\not=1$, we have $P(V)<P(1)=2$, we obtain  that at time $t_0$,
$ \Delta_s P_{j k}<0 $ and $\Delta_s P_{i k}\leq 0$ for all $1\leq i\leq N$.
Consequently, by completeness of the network,
${d\over {dt}} V_k(t_0)< 0$.
Hence as soon as $V_k$ reaches the value $1$, it must decrease again. In conclusion,
$V_k(t)\leq 1$ for all $t\geq 0$ and $k\in\{1,\dots,N\}$.

If the network is not complete,  let us proceed as before: For ${\mathbf v}\in  {\mathbb V}^k_+(\overline{V})$ and $\epsilon>0$, let ${\mathbf V}^\epsilon={\mathbf V}^\epsilon (t) = \left(V^\epsilon_1(t),\ldots, V^\epsilon_N(t)\right)$ be the solution with initial data ${\mathbf V}^\epsilon(0)= {\mathbf v}$  and with adjacency matrix $\left(c^\epsilon_{i,j}\right)$, introduced  in
\cref{adeps}. Since the corresponding network is complete, our result above applies.
As we pass to the limit $\epsilon\rightarrow 0^+$, we obtain  the sought-for estimate as before.
\end{proof}

The condition $s\geq 1$ was used in this proof in two instances: First, we appealed to the uniqueness of the solution of \cref{sys,ini} with initial data $V_i(0)=1$, $1\leq i\leq N$, in the case of a complete network. Second, we exploited the uniqueness of solutions for continuous dependence on parameters to obtain the claim for non-complete networks. When we have a complete network, it is thus sufficient to replace the condition $s\geq 1$ by the requirement that the initial value problem \cref{sys,ini} with initial data $V_i(0)=1$, $1\leq i\leq N$, only permit the solution $V_i(t)=1$, $1\leq i\le N$, $t\geq 0$. This is indeed satisfied for $0<s<1$ in the case $N=2$, as can be seen by direct calculation.

\section{Rest States and End States: Equilibria and Semiflows}

Let us now return to the study of equilibria of the system \cref{sys} in ${\mathbb V}_+(\overline{V})$ for fixed $\overline{V}>0$. In light of \cref{stat}, we consider the question whether an equilibrium can consist of  $n$ large droplets (i.e. droplets of height $h>1$) and $N-n$ small droplets (i.e. droplets of height $h^{-1}\leq 1$). Noting the volume-height relation \cref{randv}, we can thus cast mass conservation, given by \cref{masscon}, in  the form
${n\over N}\,v(h)+\left(1-{n\over N}\right)\,v\left(h^{-1}\right) = {n\over N}\,\left({{h^3+3\,h}\over 4}\right) + \left(1-{n\over N}\right)\,\left({{h^{-3}+3\,h^{-1}}\over 4}\right) =\overline{V}$,  or
\begin{equation}\label{volconstr}
	\alpha\,v(h)+(1-\alpha)\,v\left(h^{-1}\right)=\overline{V}\quad \text{with $\alpha={n\over N}$.}
\end{equation}
This equation has a simple symmetry:
\begin{equation}\label{symeq}
	\text{$h>0$ solves \cref{volconstr} with $\alpha=\beta$ if and only if $h^{-1}$  solves \cref{volconstr} with  $\alpha=1-\beta$.}
\end{equation}
It is convenient to rewrite  Equation \cref{volconstr} as
$
{\mathcal P}_\alpha(h)=0
$
where the \emph{mass polynomial} ${\mathcal P}_\alpha = {\mathcal P}_\alpha(h;\overline{V})$ is  given by
\begin{equation}\label{defpa}
	{\mathcal P}_\alpha(h) \equiv \alpha\,h^6+3\,\alpha\,h^4-4\,\overline{V}\,h^3+3\,(1-\alpha)\,h^2+1-\alpha.
\end{equation}
\begin{snugshade}
\noindent {\em Applied mathematics aside}: We will allow the parameter $\alpha$ to take on all values in the interval $[0,1]$ instead of discrete values only. This  parametrization makes it possible to study the zeros of ${\mathcal P}_\alpha$  in continuous dependence on $\alpha$ and without regard to particular values of $n$ and $N$.
\end{snugshade}
Let us now address the {\em uniform} equilibrium.

\begin{proposition}\label{unieq}
	For every $\overline{V}>0$, there is an equilibrium consisting of $N$ droplets of equal height (and volume).
\end{proposition}
The uniform equilibrium arises for $\overline{V}>1$ with $n=N$ and for $0<\overline{V}\leq 1$ with $n=0$.

Before we continue with our analysis, let us briefly adapt an approach taken by Slater and Steen \cite{SlSt} in their study of $N$ inviscid spherical droplets with $S_N$ symmetry. To this end, we note that for constant $\overline{V}$, equilibria can be
 identified with each other if they share  the same number of large droplets of the same height (or, equivalently, the same number of small droplets of the same height). Because of the symmetry expressed in \cref{symeq}
 we may assume that  $ \alpha={n\over N}\leq {{\left\lfloor{{N/ 2}}\right\rfloor}\over N}$. Using ideas proposed in \cite{ThEA},  we  now introduce the quantity
\begin{equation}\label{deftheta}
	\theta = \theta(h) \equiv v(h)-v\left({1\over h}\right),
\end{equation}
defined for all positive $h$ such that  ${\mathcal P}_\alpha(h)=0$ for  some $\overline{V}>0$ and some $\alpha={n\over N}$, $1\leq n\leq {{\left\lfloor{{N/ 2}}\right\rfloor}}$.  The function $\theta$ is clearly invertible on its domain.
The case  $\theta>0$ arises exactly when there exists $h>1$ such that ${\mathcal P}_\alpha(h)=0$. Hence in this situation we have an equilibrium with exactly  $n=\alpha\,N$ large droplets of height $h>1$. The negative case $\theta<0$ corresponds to exactly $n=\alpha\,N$ small droplets of height $0<h<1$ as is apparent from the symmetry \cref{symeq}. $\theta=0$ occurs exactly when $h=1$. Hence the corresponding equilibrium is uniform and $\overline{V}=1$. It will become evident from \cref{volconstr,symeq} and the later developments that,  in this way, $\theta$ is defined for all $h>0$. To include all uniform equilibria, we extend the definition of $\theta$ by setting $\theta=0$ if there exists $h>0$ such that ${\mathcal P}_0(h)=0$ for  some $\overline{V}>0$.  Since ${\mathcal P}_0(h) = -4\,\overline{V}\,h^3+3\,h^2+1$, we have exactly one positive zero $h$ for every choice of $\overline{V}$. The zero $h$ is in $(0,1)$ if $\overline{V}>1$, it is  $1$ if $\overline{V}=1$, and it lies in $(1,\infty)$ if $0<\overline{V}<1$. Therefore, as we appeal again to \cref{symeq}, we conclude that $\theta=0$ corresponds either to the uniform equilibrium with  $N$ large droplets of height $h^{-1}$ if $\overline{V}>1$, or to the uniform equilibrium with $N$ small droplets of height $h^{-1}$ if $0<\overline{V}\leq 1$. Hence  after this extended definition of $\theta$, $\theta=0$ represents an equilibrium in the $(\overline{V}, \theta)$-plane for every $\overline{V}>0$. Having such a constant equilibrium for all parameter values is commonly assumed in bifurcation theory \cite{GuHo}. 

Now let $a\equiv h+h^{-1}$ and $b \equiv h-h^{-1}$. Then $a^2-b^2=4$. Since \cref{deftheta} becomes
\begin{equation}
	\theta = {1\over 4}\,b\,(b^2+6),
\end{equation}
we have
\begin{equation}\label{defab}
    b={\mathcal B}(\theta) \equiv { {
-2^{2/3}+2^{1/3}\,\left(\theta+\sqrt{2+\theta^2}\right)^{2/3} }\over {\left(\theta+\sqrt{2+\theta^2}\right)^{1/3}}   }\quad\text{and}\quad
a = {\mathcal A}(\theta) \equiv  \left(4+  {\mathcal B}(\theta)^2\right)^{1/2}.
\end{equation}
Equation \cref{volconstr} can now be written in the form 
$\overline{V} = {1\over 8}\,a^3+{{2\,\alpha-1}\over 8}\,b\,(b^2+6)$.
Hence we recover a single equation relating $\theta$ and $\overline{V}$:
\begin{equation}\label{Aeq}
	\overline{V} = {1\over 8}\,{\mathcal A}(\theta)^3+ \left(\alpha-{1\over 2}\right)\,\theta.
\end{equation}
Equation \cref{Aeq} is - up to a rescaling - the same as in \cite{SlSt}. This might seem curious since the work \cite{SlSt} is concerned with inviscid spherical-cap droplets subject to $S_N$ symmetry. On second thought,  equilibria in \cite{SlSt} are determined by the same volume constraint together with the Young-Laplace relation as here. \Cref{ecurves} displays the generic situation: curves of equilibria with $n$ large droplets determined by Equation \cref{Aeq} plus the uniform equilibrium curve $\theta=0$ (labeled $n=0$ for $\overline{V}\leq 1$ and $n=N$ for $\overline{V}>1$).
\begin{figure}[tbhp]
\centering
    \includegraphics[width=0.7\textwidth]{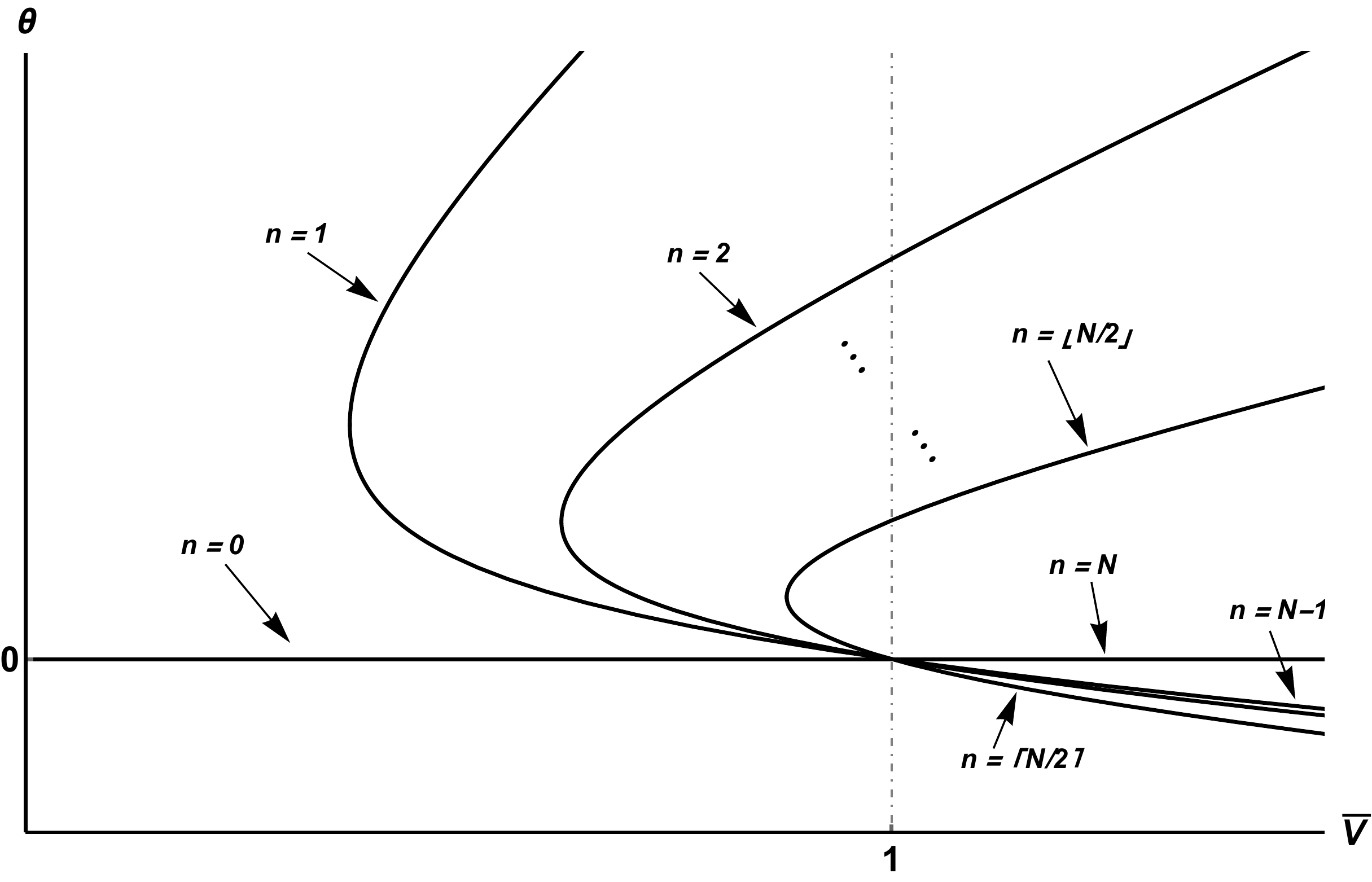}
\caption{Equilibrium curves in the $(\overline{V}, \theta)$-plane for equilibria with $n$ large droplets (odd $N$) }
 \label{ecurves}
\end{figure}

 Equation \cref{Aeq} allows an interesting first insight in the location of equilibria in dependence on $\overline{V}$. To obtain more detailed information about the location and nature of equilibria, we proceed by studying the zeros of the mass polynomial ${\mathcal P}_{\alpha}$ directly. As we will see, this approach will prove useful: Firstly, we gather details about the size of equilibria which have an immediate bearing on their stability. Secondly, we can easily  identify the turning points of the equilibrium curves in \cref{ecurves}. Thirdly, our results lay the foundation of our study of hierarchies among the equilibria.

From here onwards it will be convenient  to allow any real $\alpha$ with $0\leq \alpha\leq 1$ in the definition of ${\mathcal P}_\alpha$, keeping in mind that $\alpha={n\over N}$ describes the situation of $n$ large and $N-n$ small droplets.
Moreover,  to find non-uniform  equilibria (i.e.\,equilibria consisting of both large and small droplets), it suffices to assume that $1\leq n\leq N-1$ and  to study the zeros of ${\mathcal P}_\alpha$ for $0<\alpha<1$.
It is tacitly understood in all later developments that zeros are counted according to their multiplicity. The key result on zeros of ${\mathcal P}_\alpha$ is the following:

\begin{theorem}\label{poszero}
 Let $\overline{V}>0$. Then for every $0<\alpha<1$, the mass polynomial ${\mathcal P}_\alpha$,  has either no positive zeros or exactly two. It has exactly two positive zeros $h_1=h_1(\alpha)$ and $h_2=h_2(\alpha)$ with $h_1\leq h_2$ if and only if
\begin{equation}\label{KVcon}
		\overline{V}\geq \alpha^{3/4}\,\left(1-\alpha\right)^{1/4}+\alpha^{1/4}\,\left(1-\alpha\right)^{3/4}.
\end{equation}
In this case the zeros $h_1$ and $h_2$ are such that either
\begin{align}
	& h_1=\left({1\over \alpha}-1\right)^{1/4}=h_2\quad \text{if\quad  $\overline{V}=\alpha^{3/4}\,\left(1-\alpha\right)^{1/4}+\alpha^{1/4}\,\left(1-\alpha\right)^{3/4}$, or}\\
 	& h_1<\left({1\over \alpha}-1\right)^{1/4}<h_2\quad \text{if\quad  $\overline{V}>\alpha^{3/4}\,\left(1-\alpha\right)^{1/4}+\alpha^{1/4}\,\left(1-\alpha\right)^{3/4}$.}
\end{align}
\end{theorem}
\begin{proof}
Consider the polynomial $q_\alpha$, defined  by
\begin{multline}
	q_\alpha(h) \equiv \alpha\,\biggl(h^4+2\,\left({1\over \alpha}-1\right)^{1/4}\,h^3+3\,\left(1+\left({1\over \alpha}-1\right)^{1/2}\right)\,h^2+\\
	2\,\left({1\over \alpha}-1\right)^{1/4}\,h+\left({1\over \alpha}-1\right)^{1/2}\biggr).
\end{multline}
Evidently, $q_\alpha(h)>0$ for $h>0$.
Then for $h>0$,
\begin{equation}\label{gseq}
	{\mathcal P}_\alpha(h)\gtreqqless
	q_\alpha(h)\,\left(h-\left({1\over \alpha}-1\right)^{1/4}\right)^2\  \text{if   $\overline{V}\lesseqqgtr \alpha^{3/4}\,\left(1-\alpha\right)^{1/4}+\alpha^{1/4}\,\left(1-\alpha\right)^{3/4}$,}\\
\end{equation}
respectively. The claim follows now immediately.
\end{proof}
Condition \cref{KVcon} in \cref{poszero} is both  necessary and sufficient for positive zeros of the polynomial ${\mathcal P}_\alpha$. This surprising result hinges on the factorizability of  ${\mathcal P}_\alpha$, a polynomial of degree 6, when   $\overline{V}= \alpha^{3/4}\,\left(1-\alpha\right)^{1/4}+\alpha^{1/4}\,\left(1-\alpha\right)^{3/4}$. The symmetry \cref{symeq} hints at this result. For later use let us introduce
the \emph{limiting  function}
\begin{equation}
	\alpha\longmapsto L(\alpha)\equiv \alpha^{3/4}\,\left(1-\alpha\right)^{1/4}+\alpha^{1/4}\,\left(1-\alpha\right)^{3/4},
\end{equation}
appearing on the right of \cref{KVcon}. $L$ is strictly increasing on $(0,{1\over 2})$ and strictly decreasing on 
$({1\over 2}, 1)$. It has the maximum 1 at $\alpha={1\over 2}$. The function is graphed in \cref{limfunc}.
\begin{figure}[tbhp]
  \begin{center}
    \includegraphics[width=0.35\textwidth]{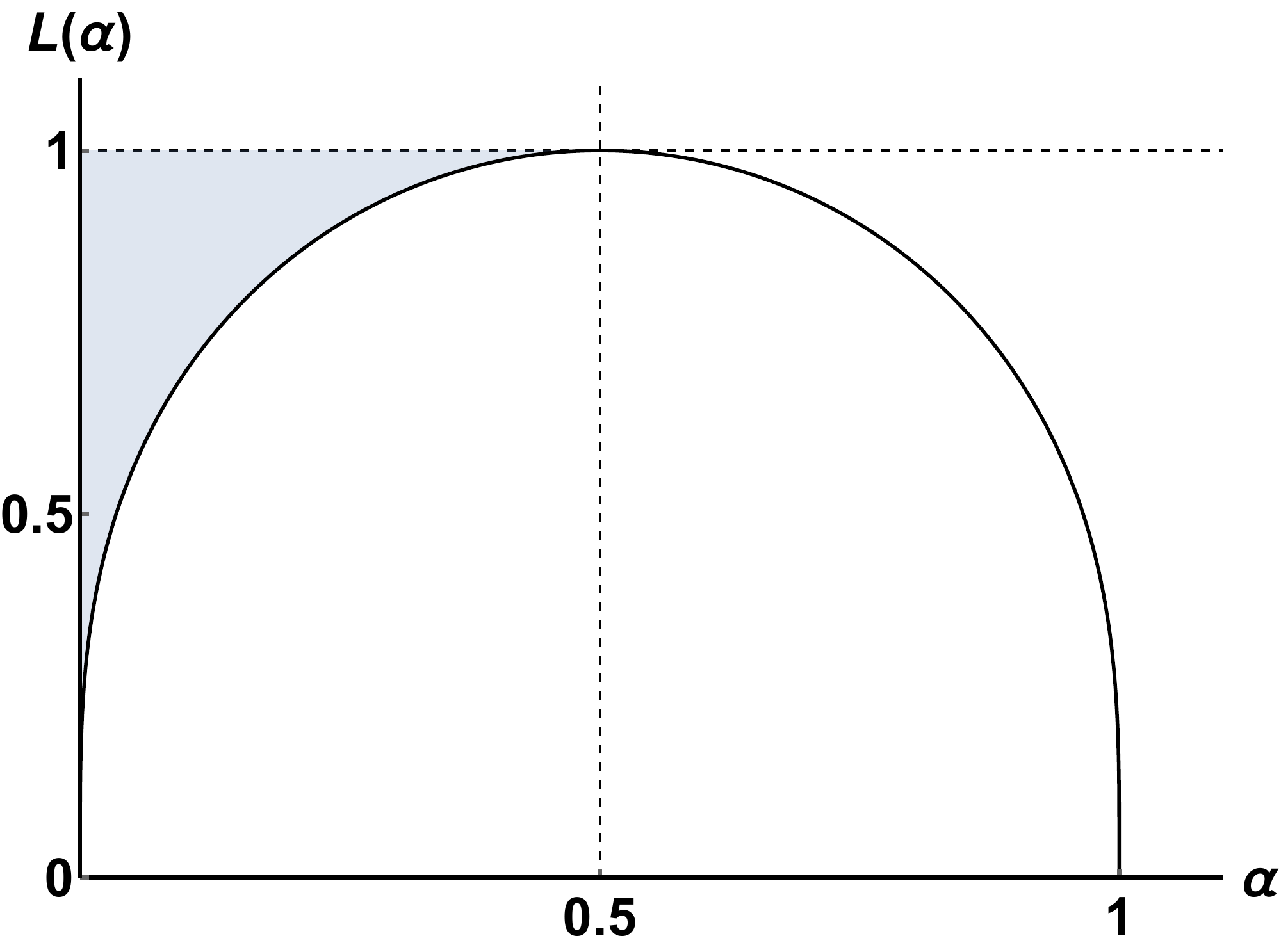}
  \end{center}
\vspace*{-0.3cm}
 \caption{The limiting function $L$}
  \label{limfunc}
\end{figure}
Next we record some additional properties  of the mass polynomial ${\mathcal P}_\alpha$.

\begin{proposition}\label{propp} For every $\overline{V}>0$ and $0<\alpha< 1$,
\begin{align}
	 & \text{${\mathcal P}_\alpha(h)\rightarrow \infty$ as $h\rightarrow \infty$,}\quad
	{\mathcal P}_\alpha(0) = 1-\alpha,\quad {\mathcal P}_\alpha(1) = 4\,\left(1-\overline{V}\right),\label{zeroone}\\
&{\mathcal P}_\alpha\left(\left({{1}\over \alpha}-1\right)^{1/4}\right) = {4}\,\left({1\over \alpha}-1\right)^{3/4}\,\left(L(\alpha)-\overline{V}\right).\label{critval}
\end{align}
\end{proposition}
With this information we obtain:

\begin{theorem}\label{pVL1}
Let $\overline{V}>1$. For every\,  $0<\alpha<1$, the mass polynomial ${\mathcal P}_\alpha$ has exactly one zero $h_L=h_L(\alpha)$ larger than 1. In particular,
\begin{equation}\label{estLa1}
	h_L>\left({1\over \alpha}-1\right)^{1/4}.
\end{equation}
\end{theorem}
Note that condition \cref{KVcon} is vacuously satisfied. The existence and uniqueness of $h_L>1$ is a consequence of \cref{zeroone} when $\overline{V}>1$. 
The estimate \cref{estLa1} follows from \cref{critval} since for $\overline{V}>1$, $0<\alpha<1$, we have  ${\mathcal P}_\alpha\left(\left({1\over \alpha}-1\right)^{1/4}\right)<0$. Of course, \cref{estLa1} is useful only if $0<\alpha<{1\over 2}$.

\begin{theorem}\label{pVE1}
Let $\overline{V} = 1$.
\begin{enumerate}
\item[\mbox{(a)}]  If\, $0<\alpha<{1\over 2}$,  the mass polynomial ${\mathcal P}_\alpha$ has exactly one zero $h_L=h_L(\alpha)$ larger than 1. In this case,
\begin{equation}\label{estE1}
	h_L>\left({1\over \alpha}-1\right)^{1/4}.
\end{equation}
\item[\mbox{(b)}] If\, ${1\over 2}\leq \alpha<1$, the mass polynomial ${\mathcal P}_\alpha$ has no zero larger than 1.
\end{enumerate}
\end{theorem}
\begin{proof}
Since  ${\mathcal P}_\alpha(1) = 0$ and ${\mathcal P}_\alpha^\prime(1) = 6\,(2\,\alpha-1)$, it is clear in light of \cref{zeroone} that ${\mathcal P}_{\alpha}$ has a zero larger than 1 if and only if $0<\alpha< {1\over 2}$. If so, this is the only zero larger than 1 since ${\mathcal P}_\alpha$ cannot have more than two positive zeros. The estimate \cref{estE1} follows again from \cref{critval} since ${\mathcal P}_\alpha\left(\left({1\over \alpha}-1\right)^{1/4}\right)<0$ if $\overline{V}=1$ and $0<\alpha<{1\over 2}$.
\end{proof}

\begin{theorem}\label{pVS1}
Let  $0<\overline{V}<1$.
\begin{enumerate}
\item[\mbox{(a)}] Suppose\, $0<\alpha<{1\over 2}$ and $\overline{V}\geq L(\alpha)$.
Then the mass polynomial ${\mathcal P}_\alpha$  has exactly two  zeros $h_l=h_l(\alpha)$ and $h_L=h_L(\alpha)$ such that
\begin{equation}\label{h12}
	1<h_l\leq \left({1\over \alpha}-1\right)^{1/4}\leq h_L.
\end{equation}
The two zeros $h_l$ and $h_L$ are such that either
\begin{align}
	& h_l=\left({1\over \alpha}-1\right)^{1/4}=h_L\quad \text{if\quad  $\overline{V}=L(\alpha)$, or}\\
 	& h_l<\left({1\over \alpha}-1\right)^{1/4}<h_L\quad \text{if\quad  $\overline{V}> L(\alpha)$.}
\end{align}
\item[\mbox{(b)}] If ${1\over 2}\leq \alpha<1$ or $\overline{V}< L(\alpha)$, the mass polynomial ${\mathcal P}_\alpha$ has no zero larger than 1.
\end{enumerate}
\end{theorem}
\begin{proof}
By \cref{poszero}, ${\mathcal P}_\alpha$, $0<\alpha<1$, has either no positive zeros or exactly two. The latter occurs if and only if  \cref{KVcon} holds true. If so,  since ${\mathcal P}_\alpha(0)>0$,  ${\mathcal P}_\alpha(1)>0$ and ${\mathcal P}_\alpha\left(\left({1\over \alpha}-1\right)^{1/4}\right)\leq 0$, \cref{propp} implies that either the number $\left({1\over \alpha}-1\right)^{1/4}$ and  the positive zeros of ${\mathcal P}_\alpha$  lie in the interval $(0,1)$, or they all lie in the interval $(1,\infty)$. Now $\left({1\over \alpha}-1\right)^{1/4}$ lies
in $(1,\infty)$ if and only if $0<\alpha<{1\over 2}$. The remaining claims follow from \cref{poszero}.
\end{proof}
The shaded area in \cref{limfunc} displays permissible pairs  $(\alpha, \overline{V})$ to which part (a) of  \cref{pVS1} applies.

All possible non-uniform equilibria are now completely described by \cref{pVL1,pVE1,pVS1} 
when we note that each such equilibrium is -- up to permutation -- given by the number of its large droplets $n$ of height $h>1$ and the number of its small droplets $N-n$ of height $h^{-1}$.
We note, in particular, that, for $0<\overline{V}\leq 1$, the condition $0<\alpha<{1\over 2}$ translates to $N\geq 3$.  Hence non-uniform equilibria with $0<\overline{V}\leq 1$ require $N\geq 3$.  For $\overline{V}>1$, every equilibrium contains at least one large droplet.  If $N\geq 2$ and $0<\overline{V}<L\left({1\over N}\right)$, then only the uniform equilibrium consisting of $N$ droplets of equal size is possible.

Finally, we end this section with some important remarks about the structure of the dynamical system,  given by \cref{sys,ini}: In the case $s\geq 1$, \cref{gloex} together with \cref{masscon} shows that the initial value problem \cref{sys,ini} defines a semiflow on ${\mathbb V}(\overline{V})$.
By \cref{stat} and \cref{Ader}, this semiflow  is a gradient dynamical system with Lyapunov function $\mathcal W$, see \cite{Ha}.
When $0<s<1$, the situation is more delicate. Here,  \cref{stat} and \cref{Ader} imply  that the initial value problem \cref{sys,ini} defines a generalized semiflow (in the sense of Ball) on ${\mathbb V}(\overline{V})$ with $\mathcal W$ serving again as a Lyapunov function, see \cite{Ba}.

Our discussion above shows that ${\mathbb V}_+(\overline{V})$ (and hence ${\mathbb V}(\overline{V})$ as well) contains only finitely many equilibria for fixed $\overline{V}$. Hence we can appeal to classical results about gradient dynamical systems in the case $s\geq 1$ and to analogous results for generalized semiflows with Lyapunov function in the case $0<s<1$ (see \cite{Ba,Ha}) to conclude:

\begin{proposition}
	The (generalized) semiflow on ${\mathbb V}(\overline{V})$, given by the initial value problem \cref{sys,ini} with ${\mathbf V}(0)\in {{\mathbb V}(\overline{V})}$, has a global attractor which consists of all  equilibria. Specifically, every  semi-trajectory  ${\mathbf V}={\mathbf V}(t)$ of the system \cref{sys} with ${\mathbf V}(0)\in {{\mathbb V}(\overline{V})}$ converges to an equilibrium of \cref{sys} in ${\mathbb V}_+(\overline{V})$ as $t\rightarrow\infty$.
\end{proposition}

\begin{snugshade}
\noindent {\em Socio-economic view}:  While we have accounted for all possible rest states (equilibria) in the competition, only some of them will be end states, i.e.\,reachable rest states (stable equilibria). This observation motivates the following section. 
\end{snugshade}

\section{Reachable Rest States: Stability of Equilibria}

\label{sec_stab}

The stability of equilibria was discussed for the Newtonian case $s=1$ with $\overline{V}>1$  in \cite{LeEA1}. The arguments there were based on the observation that
the Lyapunov function ${\mathcal W}$ is independent of the particular network of fluid channels. The authors exploited this observation by  working with a network of their choice (star network) and
combining local information about equilibria (obtained by linearization)  with the Lyapunov Stability Theorem. Instead of pursuing an approach as in \cite{LeEA1}, we will give an argument which is independent of any particular network (except the number $N$), and is also independent of the power-law parameter $s>0$ since the energy functional $\mathcal W$ is as well. Our discussion contains the stability findings   in \cite{LeEA1}. At the core of our approach lies the following simple equivalence:

\begin{proposition}\label{Astab}
	A point\, ${\mathbf V}^*\in {\mathbb V}(\overline{V})$ is a stable equilibrium of the system \cref{sys}  if and only if it is a local minimizer of the Lyapunov function ${\mathcal W}$ on ${\mathbb V}(\overline{V})$.
\end{proposition}
Note that a local minimizer of $\mathcal W$ is actually a strict local minimizer by \cref{Ader}. Also, in light of \cref{VV+} we may replace ${\mathbb V}(\overline{V})$ by ${\mathbb V}_+(\overline{V})$ here.

\begin{snugshade}
\noindent {\em Applied mathematics aside}: \Cref{Astab} characterizes stable equilibria as minimizers of an optimization problem subject to constraints. A single Lagrange multiplier suffices for the volume constraint.
\end{snugshade}

Consider the Lagrangian ${\mathcal L}: {\mathbb R}^N\times {\mathbb R}\rightarrow {\mathbb R}$, given for ${\mathbf V} = (V_1,\ldots,V_N)$ by
\begin{equation}
	{\mathcal L}({\mathbf V}, \lambda) \equiv{\mathcal W}({\mathbf V})-\lambda\,F({\mathbf V})\quad
\text{with}\quad
	F({\mathbf V}) \equiv {1\over N}\,\sum_{k=1}^N V_k -\overline{V}.
\end{equation}
If ${\mathcal W}$ assumes a local minimum on ${\mathbb V}(\overline{V})$ at ${\mathbf V}^* = (V^*_1,\ldots,V^*_N)$, then there exists $\lambda^*\in \mathbb R$ such that  $({\mathbf V}^*,\lambda^*)$ is a local minimizer of ${\mathcal L}$  on ${\mathbb R}^N\times {\mathbb R}$.
 For ${\mathcal L}$ to have a local minimum on ${\mathbb R}^N\times {\mathbb R}$ at $({\mathbf V}^*, \lambda^*)$, it is necessary that
\begin{align}
	& D_{\mathbf V} {\mathcal L}({\mathbf V}^*,\lambda^*) = {\mathbf 0},\quad F({\mathbf V}^*) = 0,\quad \text{and}\label{condmin1}\\
	& {\boldsymbol \phi}\,  D_{\mathbf V}^2 \,{\mathcal L}({\mathbf V}^*,\lambda^*)\, {{\boldsymbol \phi}}^T\geq 0\quad \text{for all ${\boldsymbol \phi}\in {\mathbb R}^N$ with $\left(D  F({\mathbf V}^*)\right)\,{\boldsymbol \phi}^T = 0$.}\label{condmin2}
\end{align}
Here, $D F$ denotes the total derivative (row vector) of $F$, while $D_{\mathbf V} {\mathcal L}$ and $D_{\mathbf V}^2 \,{\mathcal L} $ denote the total derivative  and the Hessian of $\mathcal L$ with regard to $\mathbf V$, respectively. Clearly, in our situation we have  $D_{\mathbf V}^2 \,{\mathcal L} = D^2\,{\mathcal W}$, the Hessian of ${\mathcal W}$.
Conversely, if $({\mathbf V}^*, \lambda^*)\in {\mathbb V}(\overline{V})\times {\mathbb R}$ is such that the condition \cref{condmin1} holds  and such that
\begin{equation}\label{condmin3}
	  {\boldsymbol \phi}\,  D_{\mathbf V}^2 \,{\mathcal L}({\mathbf V}^*,\lambda^*)\, {{\boldsymbol \phi}}^T> 0\quad \text{for all nonzero $ {\boldsymbol \phi}\in {\mathbb R}^N$ with $\left(D F({\mathbf V}^*)\right)\,{ {\boldsymbol \phi}}^T = 0$,}
\end{equation}
${\mathcal W}$ assumes a strict local minimum on ${\mathbb V}(\overline{V})$ at ${\mathbf V}^*$.
Conditions \cref{condmin1,condmin2,condmin3}, are classical, see e.g.\,\cite{LuYe}. \Cref{condmin1} just states the requirement  that ${\mathbf V}^*$ be an equilibrium of  \cref{sys} in ${\mathbb V}(\overline{V})$.

If ${\mathbf V}^*=(V^*_1,\ldots,V^*_N) \in {\mathbb V}(\overline{V})$ is an equilibrium of \cref{sys}, ${\mathbf V}^*$ consists of $n$ entries for large droplets of volume $V_L>1$ (and height $h_L>1$) and $N-n$  entries for small droplets of corresponding volume $0<V_S\leq 1$ (and height $h_S$). Note that $V_S=1$ immediately requires $n=0$ and $\overline{V}=1$.
Let us define
\begin{equation}
	\delta_S \equiv \left\{\begin{matrix} P^\prime(V_S) & \text{if $0\leq n\leq N-1$,}\\[1ex]
						0 &{\text{if $n=N$}}\end{matrix}\right.\quad \text{and}\quad \delta_L \equiv \left\{\begin{matrix} P^\prime(V_L) & \text{if $1\leq n\leq N$,}\\[1ex]
						0 & \text{if $n=0$.}\end{matrix}\right.
\end{equation}
Then  the necessary condition \cref{condmin2} for a local minimizer can be cast in the form
\begin{equation}\label{conalt}
		\delta_L\,\left(\sum_{k=1}^n \phi_k^2\right)+\delta_S\,\left(\sum_{k=n+1}^N \phi_k^2\right)\geq 0\quad \text{for all $(\phi_1,\ldots,\phi_N)$ with $\sum_{j=1}^N \phi_j = 0$.}
\end{equation}
The sufficient condition \cref{condmin3} for a strict local minimizer reduces to
 \begin{equation}\label{conalt2}
		\delta_L\,\left(\sum_{k=1}^n \phi_k^2\right)+\delta_S\,\left(\sum_{k=n+1}^N \phi_k^2\right)> 0\quad \text{for all nonzero $(\phi_1,\ldots,\phi_N)$ with $\sum_{j=1}^N \phi_j = 0$.}
\end{equation}
Next we have  for $V=v(h)$,
\begin{equation}\label{Pder}
	P^\prime(V) ={{16\,(1-h^2)}\over {3\,(h^2+1)^3}}.
\end{equation}
Hence $\delta_L<0$ if $1\leq n\leq N$; $\delta_S> 0$ if $1\leq n\leq N-1$, or  if $n=0$ and $0<V<1$ ; and
\begin{equation}\label{pspl}
	\delta_S=-h_L^4\,\delta_L\quad \text{if $1\leq n\leq N-1$.}
\end{equation}
  Consequently, we have immediately from \cref{conalt} that
${\mathbf V}^*$ is unstable if there exist $i, j$, $1\leq i<j\leq N$ such that $V^*_i=V^*_j=V_L>1$. Hence we have shown:

\begin{theorem}\label{stabeq}
	An equilibrium of the system \cref{sys} in ${\mathbb V}(\overline{V})$ is unstable if it contains two or more  large droplets.
\end{theorem}

For
$\overline{V}>1$, any equilibrium ${\mathbf V}^*$ in ${\mathbb V}(\overline{V})$  must contain at least one entry $V^*_i>1$. Moreover, since $\mathcal W$ assumes a global minimum on ${{\mathbb V}(\overline{V})}$, the minimizer must be an equilibrium of the system \cref{sys} by \cref{condmin1}.
 Hence we conclude:

\begin{corollary}\label{V>1stab}
An  equilibrium of the system \cref{sys} in ${\mathbb V}(\overline{V})$ with $\overline{V}>1$ is stable if and only if it contains exactly one large droplet.
\end{corollary}
\begin{snugshade}
\noindent {\em Applied mathematics aside}: 
\Cref{stabeq} is reminiscent of related results in \cite{We} for  two and three spher\-i\-cal-cap droplets and in \cite{SlSt} for $N$ coupled inviscid droplet oscillators with $S_N$ symmetry, while \cref{V>1stab} was proved in \cite{LeEA1} for $s=1$, exploiting the hyperbolicity of equilibria in the Newtonian regime.
\end{snugshade}

Let us now focus on the case $n=1$. Since for $(\phi_1,\dots,\phi_N)$ with  $\sum_{j=1}^N \phi_j = 0$
\begin{equation}
	\phi_1^2=\left(\sum_{j=2}^N \phi_j\right)^2\leq (N-1)\,\sum_{j=2}^N \phi_j^2,
\end{equation}
we can use \cref{conalt2,pspl} to obtain the sufficient condition for a strict local minimizer
\begin{equation}
	\delta_L\,\left(N-1-h_L^4\right)>0.
\end{equation}

\begin{theorem}\label{V<=1stab}
An  equilibrium of the system \cref{sys} in ${\mathbb V}(\overline{V})$ 
is stable if it contains exactly one large droplet of height $h>\left(N-1\right)^{1/4}$.
 \end{theorem}
It follows from \cref{pVE1,pVS1}  with $\alpha={1\over N}$ that such equilibria arise for $\overline{V}=1$ if $N\geq 3$, and for $0<\overline{V}<1$ if $\overline{V}>L\left({1\over N}\right)$.  For $\overline{V}>1$ this result is already contained in \cref{V>1stab} by \cref{pVL1}.

Next, in the case $n=1$, let us choose $\phi_1=1$ and $\phi_j = -(N-1)^{-1}$, $2\leq j\leq N$. Then \cref{conalt}  reduces to the following  necessary condition for a local minimizer
\begin{equation}
	\delta_L\,\left(N-1-h_L^4\right)\geq 0.
\end{equation}
Hence an equilibrium is unstable  if it contains a large droplet of height less than $\left(N-1\right)^{1/4}$. \Cref{pVS1} shows that such equilibria with $n=1$ occur for $0<\overline{V}<1$ if $\overline{V}>L\left({1\over N}\right)$. There the height of the  large droplet for this unstable equilibrium  was denoted by $h_l(\alpha)$ with $\alpha={1\over N}$.

 The borderline case $n=1$ and $h_L = \left(N-1\right)^{1/4}$ with $0<V<1$ is left to be discussed. By \cref{pVS1}, this situation arises when $\overline{V}=L\left({1\over N}\right)$ with $N\geq 3$ .
To obtain a stability result in this case, let us argue directly. First
\begin{equation}\label{P2D}
	P^{\prime\prime}(V) = {{256}\over 9}\,{{h\,(h^2-2)}\over {(h^2+1)^5}}\quad\text{where $V=v(h)$.}
\end{equation}
Next let ${\boldsymbol \phi}={\boldsymbol \phi}(t)=\left(\phi_1(t),\ldots,\phi_N(t)\right)$ be a smooth curve. Then
\begin{align}
	&{{d^2}\over {dt^2}} {\mathcal W}\left({{\boldsymbol \phi}} \right)={1\over N}\,\sum_{k=1}^N \left(P^{\prime}(\phi_k)\,\left(\phi_k^{\prime}\right)^2+P(\phi_k)\,\phi_k^{\prime\prime}\right),\quad \text{and}\label{A2D}\\
&
	{{d^3}\over {dt^3}}{\mathcal W}\left({{\boldsymbol \phi}} \right)={1\over N}\, \sum_{k=1}^N \left(P^{\prime\prime}(\phi_k)\,\left(\phi_k^{\prime}\right)^3+3\,P^\prime(\phi_k)\,\phi_k^\prime\,\phi_k^{\prime\prime}+P(\phi_k)\,\phi_k^{\prime\prime\prime}\right).\label{A3D}
\end{align}
Let $V_L=v(h_L)$ and $V_S = v\left(h_L^{-1}\right)$ with   $h_L = \left(N-1\right)^{1/4}$. For fixed $\phi_0\in \mathbb R$, we set
\begin{equation}
	 \phi_k(t) \equiv \left\{\begin{matrix} V_L+(N-1)\,\phi_0\,t & \text{if $k=1$,}\\
						V_S-\phi_0\,t & \text{otherwise.}\end{matrix}\right.
\end{equation}
Then ${{\boldsymbol \phi}}(t) = \left(\phi_1(t),\ldots,\phi_N(t)\right)$ is a smooth curve with trace in ${\mathbb V}({\overline{V}})$. Moreover,
\begin{align}
	&{{d^2}\over {dt^2}} {\mathcal W}\left({\mathbf v}(t) \right)\Big|_{t=0} = 0,\quad \text{and}\\
	&{{d^3}\over {dt^3}}{\mathcal W}\left({\mathbf v}(t) \right)\Big|_{t=0} =\left(P^{\prime\prime}(V_L)\,(N-1)^2-P^{\prime\prime}(V_S)\right)\,{{N-1}\over N}\,\phi_0^3\\
				&\hspace*{2.35cm}={{256}\over 9}\,{{\left((N-1)^{1/2}-1\right)\,(N-1)^{7/4}}\over {\left((N-1)^{1/2}+1\right)^4}}\,{{N-1}\over N}\,\phi_0^3.
\end{align}
Here we have made use of \cref{P2D}. Since ${\mathbf V}^*={{\boldsymbol \phi}}(0)$ is an equilibrium of the desired form with  $\phi_0$  arbitrary,  ${\mathcal W}$ does not assume a local minimum at ${\mathbf V}^*$. Hence we have found:

\begin{theorem}\label{V<1instab}
An  equilibrium of the system \cref{sys} in ${\mathbb V}(\overline{V})$
is unstable if it contains a large droplet of height $1<h\leq \left(N-1\right)^{1/4}$.
\end{theorem}

Finally, let us turn to the case $n=0$. The corresponding  equilibria are uniform  of height $0<h_S\leq 1$. Hence we have $0<\overline{V}\leq 1$.  \Cref{conalt2} immediately gives the sufficient condition for a strict local minimizer in the form
\begin{equation}
	\delta_S>0.
\end{equation}
Clearly,  this condition holds true for the uniform equilibrium when $0<\overline{V}<1$. This condition  fails, however, for $\overline{V}=1$ since then $\delta_S=0$.
To  obtain a stability result for $n=0$, $\overline{V}=1$, we proceed differently. First, if $N=2$, ${\mathbf V}^* = \left(1,1\right)$ is the only possible equilibrium of \cref{sys} by \cref{unieq,pVE1}. Hence  the only choice for a global minimizer of ${\mathcal W}$ on  ${\mathbb V}(\overline{V})$     is ${\mathbf V}^*$ which must be stable.

Next  assume $N\geq 3$. Let  ${\mathbf V}^*=(1,\ldots,1)$ and consider ${{\boldsymbol \phi}}={{\boldsymbol \phi}}(t)\equiv {\mathbf V}^*+t\,\left(\phi_1,\ldots,\phi_N\right)$ where $\phi_1$, $\phi_2\in \mathbb R$ are fixed, $\phi_3 = -\left(\phi_1+\phi_2\right)$ and $\phi_k=0$ for $3<k\leq N$.  Since $\sum_{k=1}^N \phi_k=0$,  ${{\boldsymbol \phi}}(t)$ is a curve with trace in ${\mathbb V}(1)$.
From \cref{A2D,A3D} we obtain
\begin{align}
	&{{d^2}\over {dt^2}} {\mathcal W}\left({{\boldsymbol \phi}}(t) \right)\Big|_{t=0} = 0,\quad\text{and}\\
	&  {{d^3}\over {dt^3}} {\mathcal W}\left({{\boldsymbol \phi}}(t) \right)\Big|_{t=0} = {1\over N}\,P^{\prime\prime}(1)\,\sum_{k=1}^N \phi_k^3 = -{3\over N}\,P^{\prime\prime}(1)\,\phi_1\,\phi_2\,\left(\phi_1+\phi_2\right).\label{indef}
\end{align}
Since $P^{\prime\prime}(1)\not=0$ by \cref{P2D} and since the right-hand side of Equation \cref{indef} can take positive and negative values for appropriate choices of $\phi_1$ and $\phi_2$, ${\mathbf V}^*$ is not a local minimizer of ${\mathcal W}$ on ${\mathbb V}(1)$. For $\overline{V}>1$ the uniform equilibrium in ${\mathbb V}(\overline{V})$ is always unstable by \cref{stabeq}. Hence we can summarize:

\begin{theorem}\label{uniV=1}
	The uniform equilibrium  of the system \cref{sys} in ${\mathbb V}({\overline{V}})$ is stable if and only if either $0<\overline{V}<1$ and $N\geq 3$, or $0<\overline{V}\leq 1$ and $N=2$.
\end{theorem}
An immediate consequence of this result is that for $N=2$ and any $s>0$, the initial value problem \cref{sys,ini} with initial data $V_1(0) = 1=V_2(0)$ has the unique solution $V_1(t) = 1=V_2(t)$, $t\geq 0$. Similarly, if $0<\overline{V}<1$, $N\geq 2$ and $s>0$, the uniform equilibrium is the unique solution of \cref{sys,ini} with initial data $V_i(0) = \overline{V}$, $1\leq i\leq N$. These findings for the case $N=2$ can be confirmed independently by direct arguments.

We have now determined the stability of all equilibria in ${\mathbb V}(\overline{V})$ and, by \cref{VV+},  in ${\mathbb V}_+(\overline{V})$ for every $\overline{V}>0$. The stability results are new for the case $s\not=1$. They are also new for the case $0<\overline{V}\leq 1$ when $s=1$.

\section{Rest State Orderings: Energy Hierarchies of Equilibria}

\label{sec_hier}

In the following  our objective is to determine an ordering of equilibria according to the size  of the energy functional $\mathcal W$. 

\begin{snugshade}
\noindent {\em Quantum-mechanical aside}: The equilibria in volume scavenging take on discrete values of the energy functional $\mathcal W$ (or total surface area $A$). We expect them to be ordered from lowest energy  (or lowest total surface area) to largest energy  (or largest total surface area) in analogy with the discrete energy levels of electrons in an atom. Hence unstable equilibria in volume scavenging  are equivalent  to excited states in a quantum-mechanical system, while a stable equilibrium minimizing the energy functional  corresponds to the ground state.
\end{snugshade}

\begin{snugshade}
\noindent {\em Socio-economic view}: The rest states of the competition  are naturally ordered in hierarchies of 
more versus less costly (or energetic) outcomes. An outcome is more costly  (energetic) if its total cost (energy) ${\mathcal W}$ is larger. 
\end{snugshade}

Let us first motivate the problem by pursuing an approach based on the variable $\theta$, introduced in \cref{deftheta}.

If $\mathbf V$ is an equilibrium of \cref{sys} consisting of $n$ droplets of height $h>0$ and $N-n$ droplets of height $h^{-1}$, $1\leq n\leq \lfloor {N\over 2}\rfloor$, then by \cref{defW,intP} with $\alpha={n\over N}$, we obtain
$	{\mathcal W}({\mathbf V}) = {3\over 2}\,\left(\alpha\,h^2+(1-\alpha)\,{1\over {h^2}}\right)$.
Using again $a=h+{1\over h}$ and $b=h-{1\over h}$, we find
$	{\mathcal W}({\mathbf V})= {3\over 2}+{3\over 4}\,b^2+{3\over 2}\,\left(\alpha-{1\over 2}\right)\,a\,b$.
Consequently, by \cref{defab}, we have a function in $\theta$
\begin{equation}\label{weq}
	 w(\theta) \equiv {3\over 2}+{3\over 4}\,{\mathcal B}(\theta)^2+{3\over 2}\,\left(\alpha-{1\over 2}\right)\,{\mathcal A}(\theta)\,{\mathcal B}(\theta)
\end{equation}
such that ${\mathcal W}({\mathbf V})=w(\theta)$ with $h={1\over 2}\left({\mathcal A}(\theta)+{\mathcal B}(\theta)\right)$. 
This {\em reduced} energy functional is displayed in \cref{wred}.
\begin{figure}[tbhp]\vspace*{-0.7cm}
\begin{center}
    \includegraphics[width=0.45\textwidth]{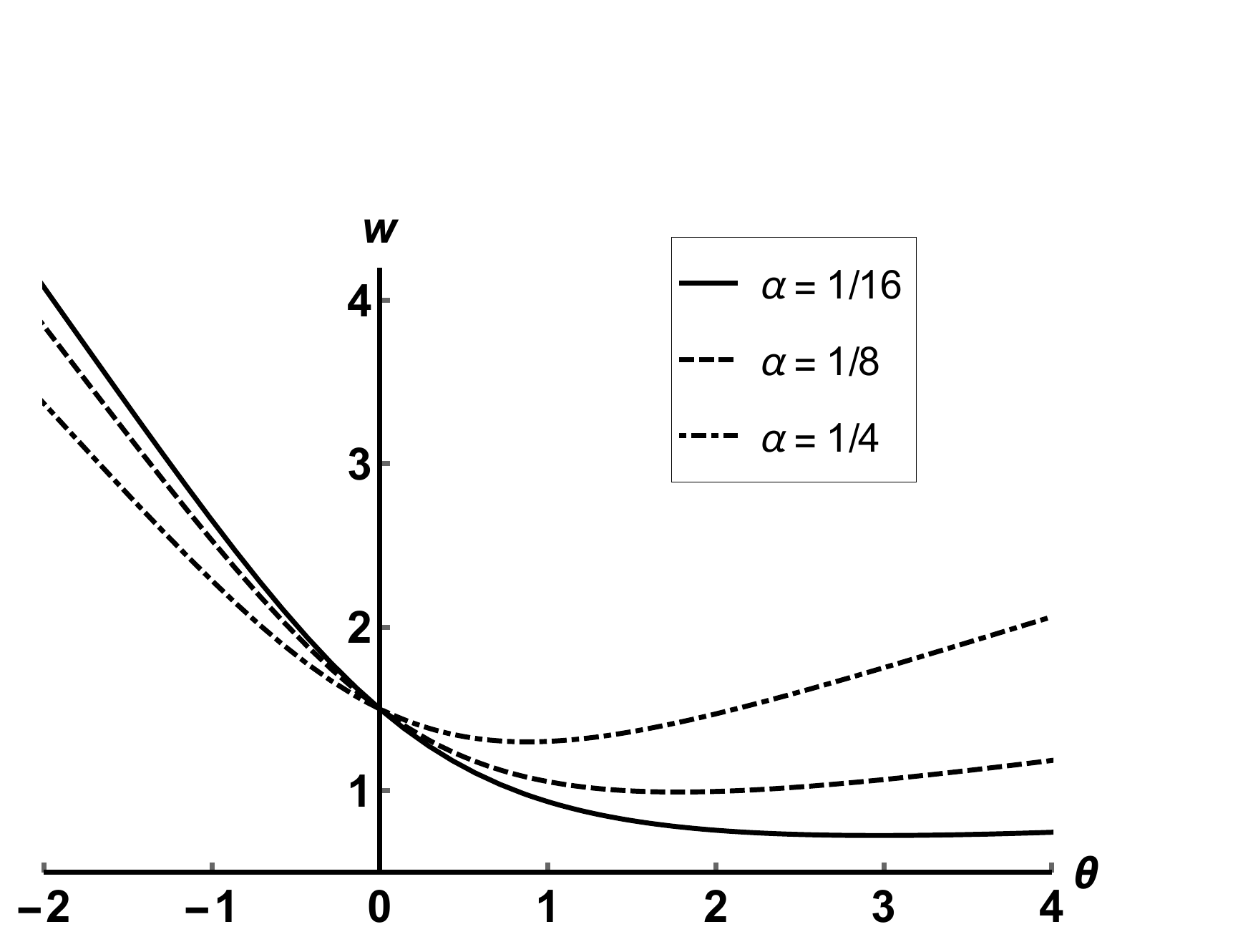}
\caption{The reduced energy functional $w$ for various values of $\alpha$}
   \label{wred}
\end{center}
\end{figure}
Let us consider an example. We choose $\overline{V}=1.2$ and determine $\theta$ and $w(\theta)$ from Equations \cref{Aeq,weq} for $\alpha\in \left\{{1\over {16}},{1\over 8},{1\over 4}\right\}$. Noting that Equation \cref{Aeq} has both a positive $\left(\oplus\right)$ and negative $\left(\ominus\right)$ solution for $\theta$, we display the results in the following tables:\\

\begin{center}
\begin{minipage}{0.45\textwidth}
\begin{tabular}{c|c|c|c}
    $\oplus$ & $\alpha={1\over {16}}$ & $\alpha={1\over {8}}$ &  $\alpha={1\over {4}}$\\[1ex] \hline
  $\theta$ & 15.936 & 7.374 & 3.199\\[1ex] \hline
   $w(\theta)$ & 1.427 & 1.646& 1.814
 \end{tabular}
 \end{minipage}
 \begin{minipage}{0.45\textwidth}
\begin{tabular}{c|c|c|c}
  $\ominus$ & $\alpha={1\over {16}}$ & $\alpha={1\over {8}}$ &  $\alpha={1\over {4}}$\\[1ex] \hline
  $\theta$ & -0.398 & -0.447 & -0.582\\[1ex] \hline
   $w(\theta)$ & 1.899 & 1.898& 1.897
 \end{tabular}
 \end{minipage}
\end{center}
\medskip
For $\theta>0$ we have  exactly $n=\alpha\,N$ large droplets, while for $\theta<0$ there are exactly $n=(1-\alpha)\,N$ large droplets. Hence the example above suggests (choosing $N=16$ for instance) that for the same average volume $\overline{V}$, the larger the energy functional $\mathcal W$ (or $w$), the more large droplets are contained in an equilibrium. 

\begin{snugshade}
\noindent {\em Physics aside}: While this claim appears intuitive and obvious, it is neither:  For simplicity take $\overline{V}>1$ and  compare two different equilibria. One have $n_1$ large droplets, the other one $n_2$. If $n_1<n_2$, each of the $n_1$ large droplets in the first equilibrium  is expected to have  larger volume (hence larger surface area) than each of the $n_2$ large droplets  in the second equilibrium. Yet there are only $n_1$ large droplets in the first equilibrium, while there are $n_2>n_1$ in the second. A similar picture arises for the small droplets. As the table with negative ($\ominus$) solutions for $\theta$ illustrates, the actual differences in the energy functional (or surface area)  can be tiny.
\end{snugshade}

We will prove below that the claim above is indeed correct if $\overline{V}>1$. For $0<\overline{V}<1$, the situation is, however, more subtle since for constant $\alpha$, Equation \cref{Aeq} generally permits two distinct positive solutions for $\theta$ and no negative solution as seen in \cref{ecurves}.

To obtain more precise information about the ordering of equilibria in terms of  $\mathcal W$, we put the approach above aside and instead examine the zeros of the mass polynomial ${\mathcal P}_\alpha$, given by \cref{defpa}, in more detail.  For $0<\overline{V}\leq 1$, we define $\alpha^*\left(\overline{V}\right)$ to be the unique value of $\alpha\in \left(0,{1\over 2}\right]$ such that $\overline{V}=L(\alpha)$.
Now consider the map $\alpha\mapsto h_L(\alpha)$, defined on $(0,1)$  for  $\overline{V}>1$ by \cref{pVL1} and on $\left(0,\alpha^*(\overline{V})\right)$ for  $0<\overline{V}\leq 1$  by \cref{pVE1,pVS1} with $h_L(\alpha)>\left({1\over \alpha}-1\right)^{1/4}$. In each case, $h_L(\alpha)$ is a simple root of the equation ${\mathcal P}_\alpha(h) = 0$. Hence it is elementary to conclude that $\alpha\mapsto h_L(\alpha)$ is a smooth map and ${d\over {d\alpha}} h_L(\alpha)\not= 0$. Since $h_L(\alpha)>\left({1\over \alpha}-1\right)^{1/4}$, we have $h_L(\alpha)\rightarrow\infty$ as $\alpha\rightarrow 0^+$. Consequently, $ {d\over {d\alpha}} h_L(\alpha)<0$. Since $h_L(\alpha)$ is also bounded from below as  $\alpha$ approaches the right endpoint of its interval of definition, we obtain:

\begin{proposition}\label{h_Lprop}
	Let\, $\alpha_0 = 1$ if $\overline{V}>1$, and  $\alpha_0= \alpha^*(\overline{V})$ if\, $0<\overline{V}\leq 1$. Then the map $\alpha\mapsto
h_L(\alpha)$ is smooth and strictly decreasing on $\left(0,\alpha_0\right)$. Moreover, it extends continuously to $\left(0,\alpha_0\right]$ such that
\begin{enumerate}
	\item[\mbox{(a)}]  in case\, $\overline{V}>1$, $h_L(\alpha_0)>1$ is the unique positive root  of  $h^3+3\,h-4\,\overline{V} = 0$,
	\item[\mbox{(b)}]  in case\, $\overline{V}=1$, $h_L(\alpha_0) = 1$, and
	\item[\mbox{(c)}]  in case\, $0<\overline{V}<1$, $h_L(\alpha_0) = \left({1\over \alpha_0} -1\right)^{1/4}>1$.
\end{enumerate}
\end{proposition}

Next let us investigate the map  $\alpha\mapsto h_l(\alpha)$, defined  on $\left(0,\alpha^*(\overline{V})\right)$ for  $0<\overline{V}<1$  by \cref{pVS1} with $1<h_l(\alpha)<\left({1\over \alpha}-1\right)^{1/4}$. Again, $h_l(\alpha)$ is a simple root of the equation ${\mathcal P}_\alpha(h)=0$, and thus $\alpha\mapsto h_l(\alpha)$ is a smooth map and ${d\over {d\alpha}} h_l(\alpha)\not= 0$. Now let $h_l(0)$ be the unique positive root of the equation ${\mathcal P}_0(h) = -4\,\overline{V}\,h^3+3\,h^2+1=0$. It follows readily that this definition extends $h_l$ to a smooth function on $\left[0,\alpha^*(\overline{V})\right)$. Since
$
	  {\mathcal P}_0(1) = -4\,\overline{V}+4>0$, $
	  {\mathcal P}_0\left({1\over {2\,\overline{V}}}\right) = {{1\over {4\,\overline{V}^2}}}+1$, and
$ \lim_{h\rightarrow \infty} {\mathcal P}_0(h) = -\infty$,
we obtain
\begin{equation}
	h_l(0)>\max\left\{1, {1\over {2\,\overline{V}}}\right\}.
\end{equation}
This estimate together with the equation ${d\over {d\alpha}} \left({\mathcal P}_\alpha(h_l(\alpha))\right)=0$ implies that
\begin{equation}
	{d\over {d\alpha}} h_l (0) = {{h_l(0)^6+3\,h_l(0)^4-3\,h_l(0)^2-1}\over {6\,h_l(0)\,\left(2\,\overline{V}\,h_l(0)-1\right)}}>0.
\end{equation}
Hence by continuity,  ${d\over {d\alpha}} h_l(\alpha)>0$ for all $\alpha$.
Since $h_l$ is increasing and $h_l(\alpha)<\left({1\over \alpha}-1\right)^{1/4}$ on $\left(0,\alpha^*(\overline{V})\right)$, the limit  $\lim_{\alpha\rightarrow \alpha^*(\overline{V})^-} h_l(\alpha)$ exists. In fact, we obtain:

\begin{proposition}\label{h_lprop}
	Let\, $0<\overline{V}<1$. Then the map $\alpha\mapsto
h_l(\alpha)$ is smooth and strictly increasing on $\left(0,\alpha^*(\overline{V})\right)$. It extends continuously to $\left[0,\alpha^*(\overline{V})\right]$ such that
$h_l(0)>1$ is the unique positive root of $-4\,\overline{V}\,h^3+3\,h^2+1=0$ and
$h_l\left(\alpha^*(\overline{V})\right)=\left({1\over {\alpha^*(\overline{V})}}-1\right)^{1/4}$.
\end{proposition}

Now we define the function $\kappa=\kappa(h)$ for $h>1$ by
\begin{equation}\label{defkappa}
	\kappa(h) \equiv{{4\,\overline{V}\,h^3-3\,h^2-1}\over {h^6+3\,h^4-3\,h^2-1}}.
\end{equation}
If $\overline{V}=1$, this definition reduces to
\begin{equation}\label{defkappa2}
	\kappa(h) = {{4\,h^2+h+1}\over {h^5+h^4+h^3+4\,h^2+h+1}},
\end{equation}
which extends $\kappa$  smoothly to all $h\geq 1$ with $\kappa(1) = {2\over 3}$.
The function $\kappa$  is trivially related to the polynomial ${\mathcal P}_\alpha$, given by \cref{defpa}, in the following way:

\begin{proposition}\label{kappap}
	Suppose $h_0>1$ is a root of the equation ${\mathcal P}_\alpha(h) = 0$ for some $0\leq \alpha\leq 1$. Then
$\kappa(h_0) = \alpha$.
\end{proposition}

Finally we introduce the functional ${\mathcal S} = {\mathcal S}(h)$ for $h>1$ by
\begin{equation}
	{\mathcal S}(h)\equiv\kappa(h)\,\int_0^{v(h)} P(V)\,dV+\left(1-\kappa(h)\right)\,\int_0^{v(1/h)} P(V)\,dV.
\end{equation}
Note that by Equation \cref{intP},
\begin{equation}\label{Sbetter}
	{\mathcal S}(h) = {3\over 2}\,\left(\kappa(h)\,h^2+\left(1-\kappa(h)\right)\,{1\over {h^2}}\right).
\end{equation}
\begin{snugshade}
\noindent {\em Applied mathematics aside}: 
For integers $N\geq 2$ and $0\leq n\leq N$,  let $h_0>1$ be a root of the equation ${\mathcal P}_\alpha(h) = 0$ with $\alpha= {n\over N}$.
Then ${\mathcal S}(h_0) = {\mathcal W}({\mathbf V}_0)$ where ${\mathbf V}_0$ is any equilibrium consisting of exactly $n=\alpha\,N$ large droplets of height $h_0$. Hence $\mathcal S$ interpolates  the values of $\mathcal W$ along the equilibria.
\end{snugshade}
Using  \cref{defkappa}, we find
\begin{equation}
	{{d\mathcal S}\over {dh}}(h) =  -6\,{{\left(h^2-1\right)\,\left(\overline{V}\,h^4-h^3-h+\overline{V}\right)}\over
{\left(h^4+4\,h^2+1\right)^2}}.
\end{equation}

\begin{theorem}\label{incr}
	Let $\alpha_0 = 1$ if $\overline{V}>1$, and  $\alpha_0= \alpha^*(\overline{V})$ if\, $0<\overline{V}\leq 1$.  Then the maps
$\alpha\mapsto {\mathcal S}(h_L(\alpha))$ and $\alpha\mapsto {\mathcal S}(h_l(\alpha))$ are  strictly increasing on $(0,\alpha_0)$.
\end{theorem}

\begin{proof}	
Let $h=h_L$ or $h=h_l$ on $(0,\alpha_0)$. Since $h(\alpha)>1$ and $\kappa(h(\alpha)) = \alpha>0$ on $(0,\alpha_0)$, the definition \cref{defkappa} of $\kappa$ shows that $4\,\overline{V}\,h(\alpha)^3-3\,h(\alpha)^2-1>0$. Hence by \cref{h_Lprop,h_lprop},
\begin{equation}
	\left(4\,\overline{V}\,h^3-3\,h^2-1\right)\,{d\over {d\alpha}} h<0 \  \text{ if $h=h_L$,\  and}\ \left(4\,\overline{V}\,h^3-3\,h^2-1\right)\,{d\over {d\alpha}} h>0\ \text{if $h=h_l$.}
\end{equation}
Consequently, the map $\alpha\mapsto {\mathcal F}(\alpha)\equiv \overline{V}\,h(\alpha)^4-h(\alpha)^3-h(\alpha)+\overline{V}$ is  strictly decreasing on $(0,\alpha_0)$ if $h=h_L$, and  strictly increasing on $(0,\alpha_0)$ if $h=h_l$. In the first case, we obtain from  \cref{h_Lprop} that $\lim_{\alpha\rightarrow \alpha_0^+} {\mathcal F}(\alpha)$ exists and is non-negative. Thus ${\mathcal F}(\alpha)>0$ on $(0,\alpha_0)$. In the second case, \cref{h_lprop} proves that $\lim_{\alpha\rightarrow \alpha_0^+} {\mathcal F}(\alpha)$ exists and equals $0$. Hence ${\mathcal F}(\alpha)<0$ on $(0,\alpha_0)$. Finally, the claim follows since $h(\alpha)>1$ on $(0,\alpha_0)$ and
\begin{equation}
	{d\over {d\alpha}}\left({\mathcal S}(h(\alpha))\right) = -6\,{{\left(h(\alpha)^2-1\right)\,{\mathcal F}(\alpha)}\over
{\left(h(\alpha)^4+4\,h(\alpha)^2+1\right)^2}}\,{d\over {d\alpha}} h(\alpha).
\end{equation}
\end{proof}

Now, for fixed $N$, we let $N^*(\overline{V})$  be the largest integer $K$ such that $K\leq \alpha^*(\overline{V})\,N$ if  $0<\overline{V}< 1$. We also define $N^*(1)$ to be the largest integer $K$ such that $K< {1\over 2}\,N$. Specifically, we have
\begin{equation}
	N^*(\overline{V}) =\left\{\begin{matrix}
				\left\lceil{{1\over 2}\,N}-1\right\rceil& \text{for  $\overline{V}= 1$,}\\
				\left\lfloor{\alpha^*(\overline{V})\,N}\right\rfloor& \text{for  $0<\overline{V}< 1$.}
			\end{matrix}\right.
\end{equation}
Note that for $0<\overline{V}<1$, the condition $N^*(\overline{V}) \geq 1$ is equivalent to $L\left({1\over N}\right)\leq \overline{V}<1$.
Now we set:
\begin{enumerate}
	\item[\mbox{(i)}] If $\overline{V}>1$, let $\beta_n \equiv {\mathcal W}({\mathbf V}_n)$,  $1\leq n\leq N$,  where ${\mathbf V}_n$   is an equilibrium of the system \cref{sys} with exactly $n$ large droplets.
	\item[\mbox{(ii)}]   If $\overline{V}=1$,  let $\gamma_n \equiv  {\mathcal W}({\mathbf V}_n)$,  $0\leq n\leq N^*(1)$,
where ${\mathbf V}_n$   is an equilibrium of the system \cref{sys} with exactly $n$ large droplets.
\item[\mbox{(iii)}] If $0<\overline{V}<1$ and $N^*(\overline{V})\geq 1$, let $\lambda_n \equiv {\mathcal W}({\mathbf V}_n)$,   $1\leq n\leq N^*(\overline{V})$,  where  ${\mathbf V}_n$   is an equilibrium of the system \cref{sys} with exactly $n$ large droplets of height   $h\geq {\left({N\over n}-1\right)^{1/4}}$.
	\item[\mbox{(iv)}] If  $0<\overline{V}<1$, let $\sigma_n \equiv {\mathcal W}({\mathbf V}_n)$,  $0\leq n\leq N^*(\overline{V})$,  where  ${\mathbf V}_n$   is an equilibrium of the system \cref{sys} with exactly $n$ large droplets of height   $h\leq {\left({N\over n}-1\right)^{1/4}}$.
\end{enumerate}
Observe, in particular, that $\beta_N$, $\gamma_0$ and $\sigma_0$ denote the value of the functional ${\mathcal W}$ at the uniform equilibrium when $\overline{V}>1$, $\overline{V} = 1$,  and $0<\overline{V}<1$, respectively.

\begin{theorem}\label{hierarchies}
\mbox{\ }
\begin{enumerate}
	\item[\mbox{(a)}] If\ \  $\overline{V}>1$, then $\beta_n<\beta_{n+1}$, $1\leq n\leq N-1$.
	\item[\mbox{(b)}] If\ \ $\overline{V}=1$, then $\gamma_n<\gamma_{n+1}$, $1\leq n \leq N^*(1)-1$, and $\gamma_n<\gamma_0$, $1\leq n\leq N^*(1)$.
\item[\mbox{(c)}] If\ \ $0<\overline{V}<1$, then  $\lambda_n<\lambda_{n+1}$, $1\leq n\leq N^*(\overline{V})-1$, and $\sigma_n<\sigma_{n+1}$, $0\leq n\leq N^*(\overline{V})-1$.
\end{enumerate}
\end{theorem}

\begin{proof}
Let $\alpha_0$ be given as in \cref{incr} and let $h=h_L$ or $h=h_l$ on $(0,\alpha_0)$. Then with $V_L = v(h(\alpha))$ and $V_S = v(1/h(\alpha))$,
${\mathcal S}(h(\alpha)) = \alpha\,\int_0^{V_L} P(V)\,dV + (1-\alpha)\,\int_0^{V_S} P(V)\,dV$.
Hence setting $\alpha={n\over N}$, we obtain
\begin{equation}
	{\mathcal S}\left(h\left({n\over N}\right)\right) = \left\{\begin{matrix} \beta_n& \text{if $\overline{V}>1$, $h=h_L$,  and $1\leq n\leq N-1$,}\\[1ex]
						\gamma_n& \text{if $\overline{V}=1$, $h=h_L$,  and $1\leq n\leq N^*(1)$,}\\[1ex]
						\lambda_n& \text{if $0<\overline{V}<1$, $h=h_L$,  and $1\leq n\leq N^*(\overline{V})-1$},\\[1ex]				
						\sigma_n& \text{if $0<\overline{V}<1$, $h=h_l$, and $1\leq n\leq N^*(\overline{V})-1$.}\end{matrix}\right.
\end{equation}
If either $\overline{V}>1$ or $0<\overline{V}<1$, $h_L$ extends to a continuous function on $(0,\alpha_0]$ such that
$h_L(\alpha)>1$ for $0<\alpha\leq \alpha_0$. Consequently,
	${\mathcal S}\left(h_L(1)\right) = \beta_N$  if $\overline{V}>1$, and 
	${\mathcal S}\left(h_L\left({{N^*(\overline{V})}/ N}\right)\right) = \lambda_{N^*(\overline{V})}$ if $0<\overline{V}<1$,
provided that $N^*(\overline{V})\geq 1$.
If $\overline{V} = 1$, $h_L$ extends to a continuous function on $\left(0,{1\over 2}\right]$ such that $h_L(\alpha)>1$ for $0<\alpha<{1\over 2}$ and $h_L\left({1\over 2}\right) = 1$. Consequently, by \cref{defkappa2,Sbetter}, the map $\alpha\mapsto {\mathcal S}\left(h_L(\alpha)\right)$ extends continuously to $\left(0,{1\over 2}\right]$ such that
\begin{equation}
	{\mathcal S}\left(h_L\left({1\over 2}\right)\right) = {\mathcal S}(1) = {2\over 3}\,\int_0^1 P(V)\,dV+ \left(1-{2\over 3}\right)\,\int_0^1 P(V)\,dV = \int_0^1 P(V)\,dV = \gamma_0.
\end{equation}
Finally, if  $0<\overline{V}<1$, $h_l$ extends to a continuous function on $[0,\alpha_0]$ such that
$h_l(\alpha)>1$ for $0\leq \alpha\leq \alpha_0$. Hence
	${\mathcal S}\left(h_l(0)\right) = \sigma_0$ and 
	${\mathcal S}\left(h_l\left({{N^*(\overline{V})}/N}\right)\right) =  \sigma_{N^*(\overline{V})}$.
The ordering of the quantities $\beta_n$, $\gamma_n$, $\lambda_n$ and $\sigma_n$ follows now immediately from the strict monotonicity of the maps $\alpha\mapsto {\mathcal S}(h_L(\alpha))$ and
$\alpha\mapsto {\mathcal S}(h_l(\alpha))$.
\end{proof}

In the case $\overline{V}<1$ with $N^*(\overline{V})\geq 1$,  the system \cref{sys}  potentially exhibits two types of stable equilibria: the uniform equilibrium  with droplets of height $h_0=h_l(0)$ and equilibria consisting of exactly one large droplet of height $h_N = h_L\left({1\over N}\right)$. It is natural to ask for which stable equilibrium the functional $\mathcal W$ is minimal.
Instead of  a comprehensive answer we can easily give a partial one: Since $\kappa(h)\sim h^{-3}$ as $h\rightarrow \infty$ and $h_N=h_L\left({1\over N}\right)\rightarrow \infty$ as $N\rightarrow\infty$, we obtain from  \cref{Sbetter} that
$
{\mathcal S}\left(h_N\right) \rightarrow 0$ as $N\rightarrow\infty$, while ${\mathcal S}\left(h_0\right) = {3\over 2}\,{{h_0^{-2}}}>0$.
Hence for sufficiently large $N$ (with constant $0<\overline{V}<1$), the stable equilibria with one large droplet are the global minimizers of $\mathcal W$. If, on the other hand,
$
	\overline{V} =L\left({1\over N}\right)
$,
then $\alpha^*(\overline{V}) = {1\over N}$ and,  by \cref{pVS1}, $h_l\left({1\over N}\right) = h_L\left({1\over N}\right)$. By \cref{V<1instab}, the corresponding equilibria with exactly one large droplet are unstable. Hence they cannot be   minimizers of $\mathcal W$. Indeed, since 
$
	\sigma_0<\sigma_1= \lambda_1
$  by
\cref{hierarchies},
$\mathcal W$ takes its global minimum value at the uniform equilibrium. Clearly, there exists $\epsilon>0$ (depending on $N$) such that this conclusion remains valid if
$
	L\left({1\over N}\right)\leq \overline{V}<L\left({1\over N}\right)+\epsilon$.
\begin{snugshade}
\noindent {\em Engineering aside}: The possible occurrence  of two types of stable equilibria for  $\overline{V}<1$   is an interesting observation which lends support to the switching mechanism of the adhesion device presented by Vogel and Steen in \cite{VoSt}. The device can switch between two states: attached and detached. Both states are equilibria of the droplet competition. The attached state corresponds to the uniform equilibrium where each droplet forms a liquid bridge with the flat substrate. The detached state arises when one droplet is much larger than the other ones. Liquid bridges between the small droplets and the substrate will be broken. Our idealized fluid flow model is able to describe and explain this grab-and-release mechanism in qualitative terms.
\end{snugshade}

\begin{snugshade}
\noindent{\em Material science aside}: Since for $L\left({1\over N}\right)<\overline{V}<1$ $(N\geq 3$) volume scavenging is ``bistable'' (i.e.\,it exhibits two different types of stable equilibria),  an initial droplet configuration can evolve towards the uniform equilibrium. In this case,  volume scavenging does not lead to  Ostwald ripening. The uniform equilibrium is still a local minimizer of the energy functional $\mathcal W$, and yet no coarsening occurs.
\end{snugshade}

\section{Concluding Remarks}

\label{sec_concl}

In the preceding sections we concentrated on the occurrence of equilibria, their stability and their ordering for fixed $\overline{V}>0$. It is now instructive to shed light on the equilibria and their stability when $\overline{V}$ varies. We make use of the equilibrium curves obtained via Equation \cref{Aeq} and identify equilibria as stable or unstable. We classify as follows:
\begin{itemize}
	\item $L^N$, $S^N$: uniform equilibria, consisting of $N$ large ($L$) or $N$ small ($S$) droplets, respectively
	\item $L^k\,S^m$: non-uniform equilibria, consisting of $k$ large ($L$) droplets of height $h> \left({N\over k}-1\right)^{1/4}$ and $m=N-k$ small ($S$) droplets
	\item $l^k\,s^m$: non-uniform equilibria, consisting of $k$ large ($l$) droplets of height $h\leq \left({N\over k}-1\right)^{1/4}$ and $m=N-k$ small ($s$) droplets
\end{itemize}
\begin{figure}[tbhp]
\centering
  \subfloat[$N=2$]{
    \includegraphics[width=0.48\textwidth]{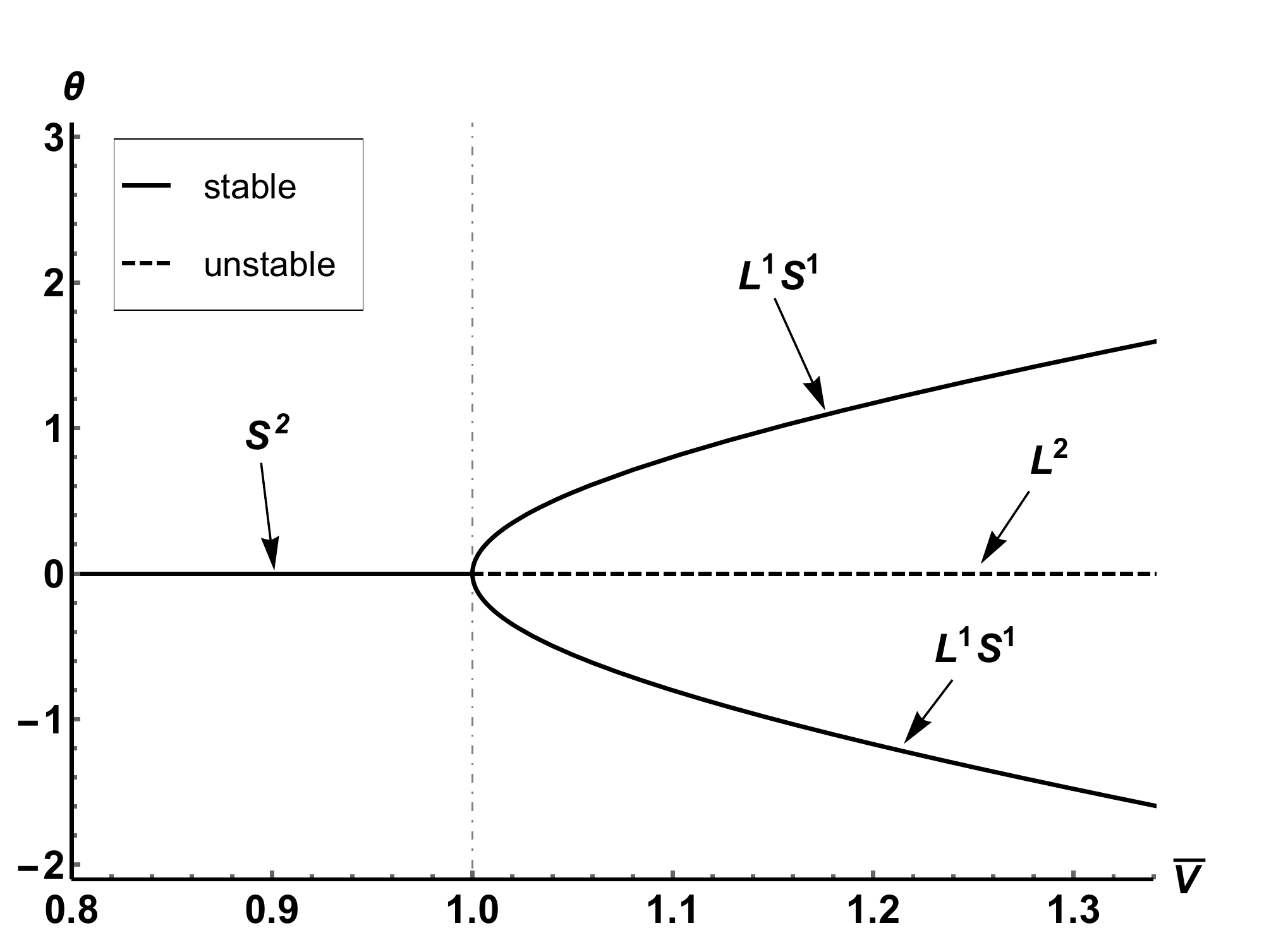}
    \label{N2bf}}
\subfloat[$N=3$]{
    \includegraphics[width=0.48\textwidth]{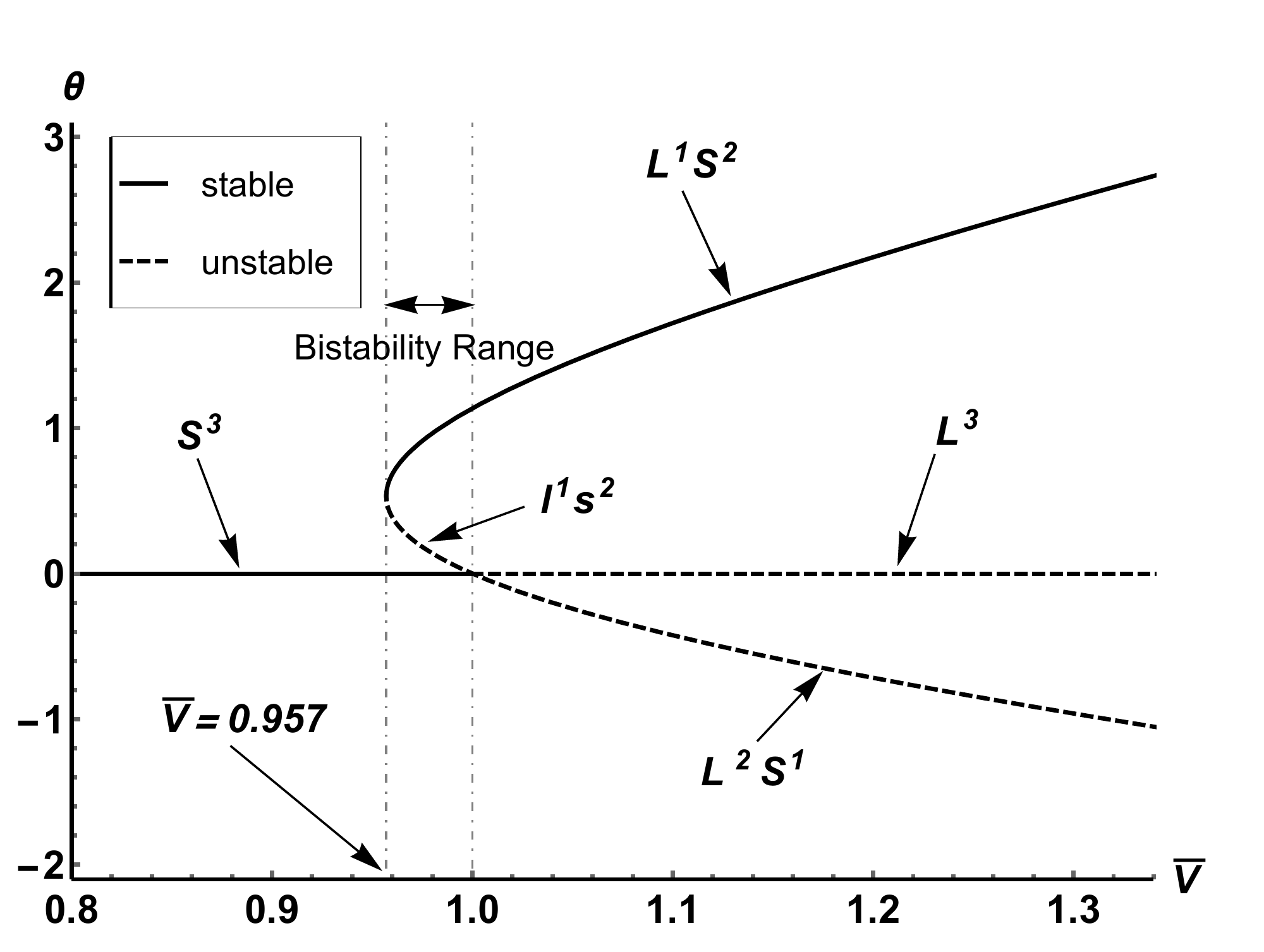}
    \label{N3bf}}
\caption{Bifurcation diagrams for $N=2, 3$}
\label{Bfd23}
\end{figure}
\begin{figure}[tbhp]
\centering
    \includegraphics[width=0.95\textwidth]{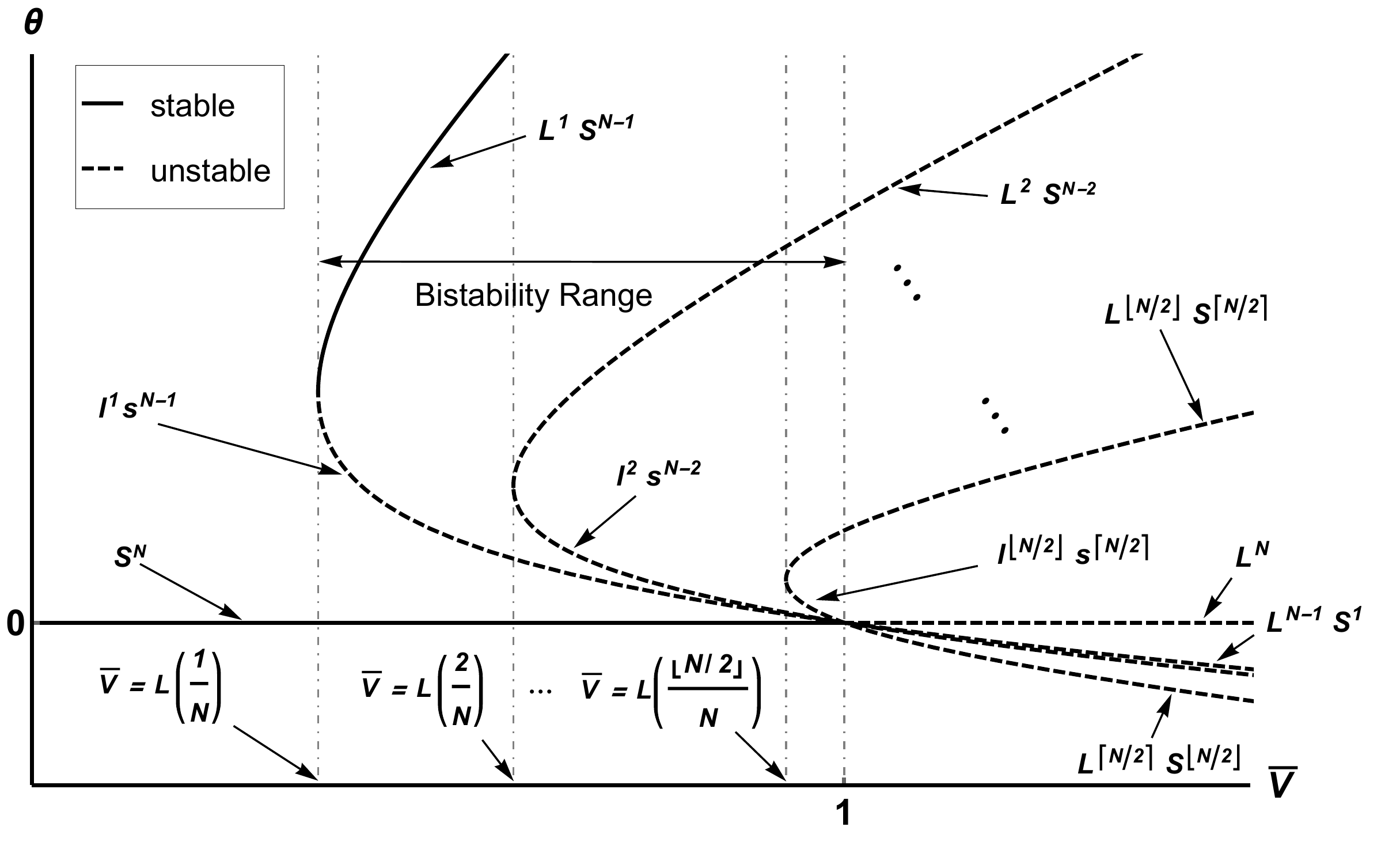}
\caption{Generic bifurcation diagram for odd $N$}
\label{Bfd}
\end{figure}

\Cref{Bfd23} displays the situation $N=2, 3$. A generic bifurcation diagram (for odd $N$) is shown in \cref{Bfd}:
Equilibria of the type $L^k\,S^m$ and $l^k\,s^m$ are both present on equilibrium curves with turning point in the half-plane  $\overline{V}<1$. Such equilibria  lie immediately above and below the turning point, respectively.  The turning point itself occurs for $\overline{V} = L\left({n\over N}\right) = {1\over N}\,\left(\left({N\over n}-1\right)^{1/4}+\left({N\over n}-1\right)^{3/4}\right)$.  Up to a rescaling,  this value of $\overline{V}$ was  given by Slater and Steen \cite{SlSt} in the case $n=1$ and has been reconfirmed here. At this turning point a saddle-node bifurcation occurs with equilibria of the type $L^1\,S^{N-1}$ stable and all other equilibria on the same solution curve unstable. Our results on the limiting function $L$ furnish the location of all turning points with $n\in \left\{1,\ldots,\left\lfloor{N\over 2}\right\rfloor\right\}$.  Equilibria of the type $l^k\,s^m$ extend all the way to the point $(\overline{V},\theta) = (1,0)$ and then switch to  the type $L^m S^k$. No equilibria of the type $l^k\,s^m$ are present on the equilibrium curve with  $\alpha={1\over 2}$, i.e.\,for $N$ even and $n={N\over 2}$. This curve consists of two solution branches symmetric about the axis $\theta=0$. In the case $N=2$ this situation corresponds to a supercritical pitchfork bifurcation  as indicated in  \cref{N2bf}. Our results are in agreement with Wente's work for $N=2$ and $3$ \cite{We}. The {\em bistability range} $L\left({1\over N}\right)<\overline{V}<1$, i.e.\,the interval of $\overline{V}$-values where two different types of stable equilibria can occur,  is indicated in \cref{N3bf,Bfd}. 

When $\overline{V}>1$,  the basins of attraction of stable equilibria are sensitive to changes in the initial data (and model parameters) as seen in \Cref{sec_mot}. A similar behavior is expected when $\overline{V}<1$. We leave it to future work to study the basins of attraction, especially  for 
$\overline{V}$ in the bistability range.

\begin{snugshade}
\noindent {\em Socio-economic view}: The volume scavenging competition of $N$ individuals exhibits three distinct behaviors: For  scarcest resource, $0<\overline{V}<L\left({1\over N}\right)$, the only possible outcome is egalitarian. In contrast, for abundant resource, $\overline{V}>1$, the competition exhibits the winner-take-all outcome. Finally, for scarce resource, $L\left({1\over N}\right)<\overline{V}<1$, both the winner-take-all and the all-share-evenly outcomes are possible. 
As $N\rightarrow \infty$, 
$L\left({1\over N}\right)\rightarrow 0$. Hence, as the number of competitors grows,  the range of resource where
bistability occurs increases.  At the same time, the range where only the egalitarian outcome is possible becomes smaller.
\end{snugshade}
\begin{figure}[tbhp]
\centering
  \subfloat[$N=10$]{
    \includegraphics[width=0.31\textwidth]{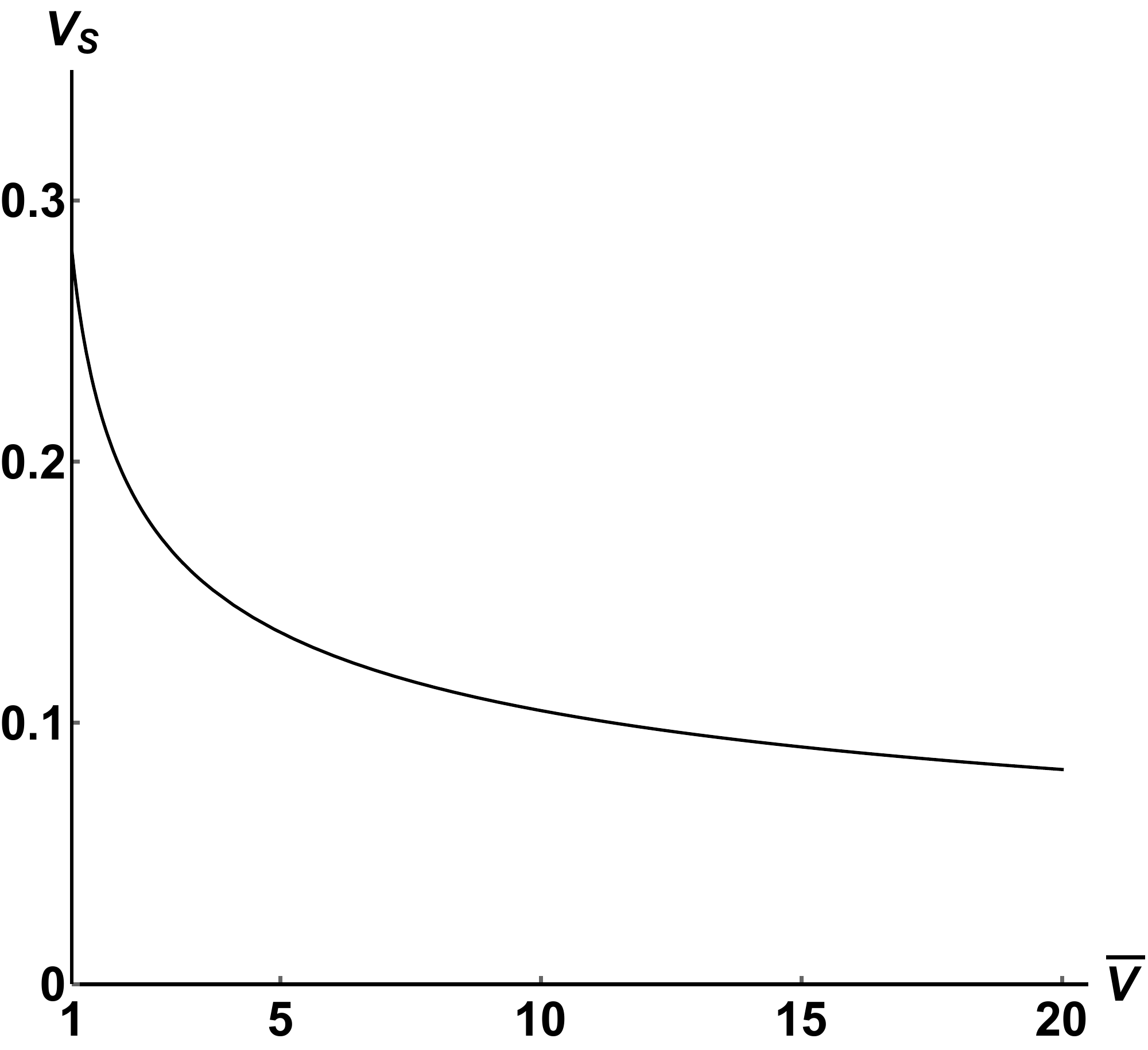}
    \label{SocEco10}}
 \subfloat[$N=30$]{
    \includegraphics[width=0.31\textwidth]{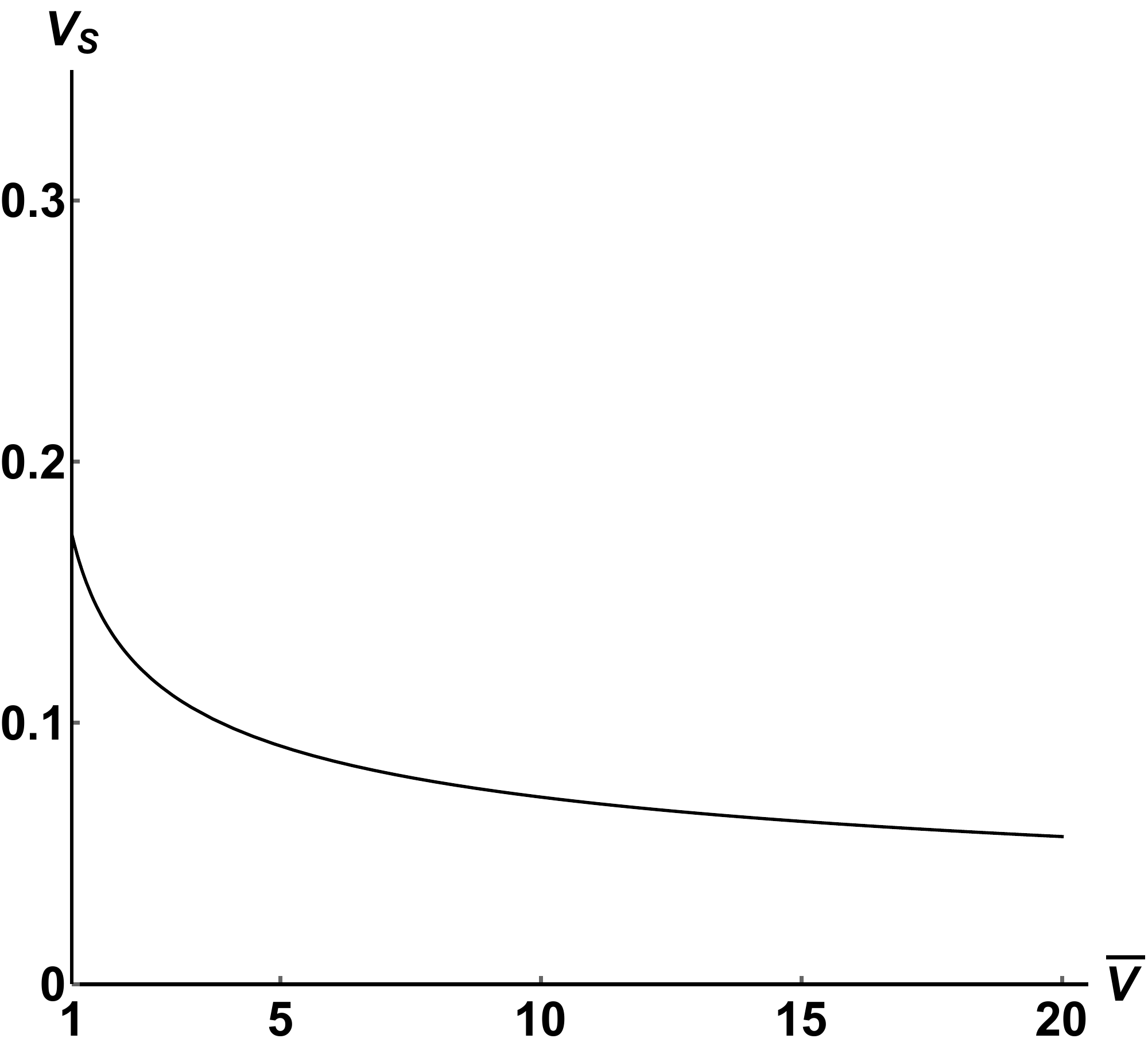}
    \label{SocEco30}}
\subfloat[$N=100$]{
    \includegraphics[width=0.31\textwidth]{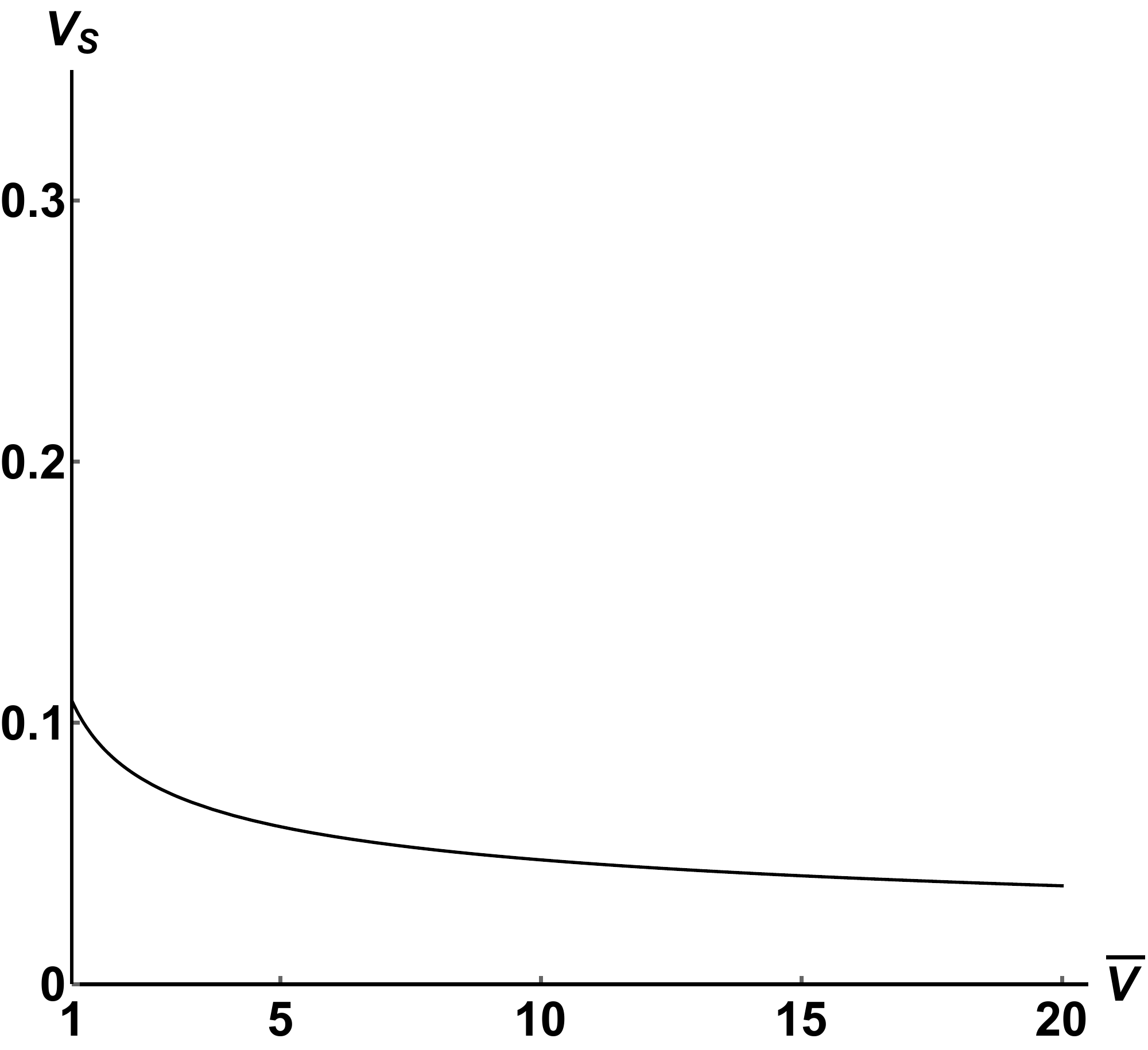}
    \label{SocEco100}}
\caption{Small droplet volume $V_S$ for stable equilibria as a function of  ${\overline V}>1$}
\label{EcoSoc}
\end{figure}
\begin{snugshade}
\noindent {\em  Socio-economic view}: 
The physics-based  model we studied motivates many questions of socio-eco\-nom\-ic interest. 
Among other things we might ask whether the ``macroeconomic objective'' of increasing  average resource per individual 
$\overline{V}$ actually results in a better outcome for the majority of  individuals in a population of size $N$. 
Since, for abundant resource $\overline{V}>1$, individuals experience the winner-take-all outcome, all but one competitor will earn an equal,  small share 
of total resource, corresponding to the volume  $V_S=V_S\left(\overline{V}\right)$ of a
small droplet in a stable equilibrium with exactly one large droplet. Hence $V_S\left(\overline{V}\right)<1$ is the end state resource for every non-winning individual in the winner-take-all market with abundant resource $\overline{V}>1$. Graphs of this function are depicted in \cref{EcoSoc} for populations of size $N=10$, $N=30$ and $N=100$. Surprisingly, it appears that an increase in average resource per individual $\overline{V}$ results in a smaller share $V_S\left(\overline{V}\right)$. Hence almost all individuals fare better  when resource $\overline{V}$ is less abundant.
\end{snugshade}

We have now given a complete characterization of all  equilibria and their stability for volume scavenging of Newtonian and power-law fluids.   While our stability and bifurcation results bear resemblance with the case of $N$ inviscid droplets with $S_N$ symmetry \cite{SlSt}, there are notable differences: The inviscid flow in \cite{SlSt} exhibits periodic and quasi-periodic solutions as well as chaotic behavior.  In the viscous situation discussed here, the fluid rheology (for any $s>0$) renders stable equilibria  asymptotically stable and restricts $\omega$-limit sets to a finite number of equilibria. This viscous behavior for $\overline{V}<1$ allows us to rationalize  the mechanism for the adhesion device discussed in \cite{VoSt}.  It is, however, remarkable that, in turn, our  results on the hierarchical ordering of equilibria in terms of the pressure-volume work functional $\mathcal W$  (and hence total surface area)  apply to the inviscid case as well. Our description of the turning points of the equilibrium curves as special values of the limiting function $L$ carry over directly. Most of these findings are new, both for the viscous and inviscid case.

\section*{Acknowledgement}

The authors are grateful  to the Institute of Mathematics and Its Applications at the University of Minnesota for its support. This work has its roots in scientific activities offered at the IMA during the thematic year on ``Complex Fluids and Complex Flows'' in 2009/2010.
The authors also thank Profs.\,\,Robert Frank (Cornell University) and Philip Cook (Duke University) for relevant references in sociology and economics.

\bibliographystyle{s-plain}

\begin{thebibliography}{10}

\bibitem{Ba}
{\sc J.~M. Ball}, {\em Continuity properties and global attractors of
  generalized semiflows and the {N}avier-{S}tokes equations}, in Mechanics:
  From Theory to Computation, Springer-Verlag, New York, 2000, pp.~447--474.

\bibitem{BiEA}
{\sc R.~B. Bird, R.~C. Armstrong, and O.~Hassager}, {\em Dynamics of Polymeric
  Liquids}, vol.~1, Wiley Interscience, New York, 1987.

\bibitem{Bo}
{\sc C.~V. Boys}, {\em Soap Bubbles: Their Colors and Forces Which Mold Them},
  Dover Publications, New York, 1958.

\bibitem{DrEA}
{\sc W.~Dreyer, I.~M\"uller, and P.~Strehlow}, {\em A study of equilibria of
  interconnected balloons}, Quart. J. Mech. Appl. Math., 35 (1982),
  pp.~419--440.

\bibitem{EiAn}
{\sc T.~Eisner and J.~Aneshansley}, {\em Defense by foot adhesion in a beetle
  (hemisphaerota cyanea)}, Proc. Nat. Acad. Sci. USA, 97 (2000),
  pp.~6568--6573.

\bibitem{FrCo-art}
{\sc R.~H. Frank and P.~J. Cook}, {\em Winner-take-all markets}, Papers
  Political Econ.,  (1991), 18.
\newblock Political Economy Research Group, University of Western Ontario.

\bibitem{FrCo}
{\sc R.~H. Frank and P.~J. Cook}, {\em The Winner-Take-All Society: Why the Few
  at the Top Get So Much More Than the Rest of Us}, The Free Press, New York,
  1995.

\bibitem{GuHo}
{\sc J.~Guckenheimer and P.~Holmes}, {\em Nonlinear Oscillations, Dynamical
  Systems, and Bifurcations of Vector Fields}, Springer-Verlag, New York, 1983.

\bibitem{Ha}
{\sc J.~K. Hale}, {\em Stability and gradient dynamical systems}, Rev. Mat.
  Complut., 17 (2004), pp.~7--57.

\bibitem{LuYe}
{\sc D.~G. Luenberger and Y.~Ye}, {\em Linear and Nonlinear Programming},
  Springer-Verlag, New York, third~ed., 2008.

\bibitem{RaVo}
{\sc L.~Ratke and P.~W. Voorhees}, {\em Growth and Coarsening. Ostwald Ripening
  in Material Processing}, Springer-Verlag, Berlin, 2002.

\bibitem{Sc-art}
{\sc T.~C. Schelling}, {\em Dynamic models of segregation}, J. Math. Sociol., 1
  (1971), pp.~143--186.

\bibitem{Sc}
{\sc T.~C. Schelling}, {\em Micromotives and Microbehavior}, Norton, New York,
  1978.

\bibitem{SlSt}
{\sc D.~M. Slater and P.~H. Steen}, {\em Bifurcation and stability of {$n$}
  coupled droplet oscillators with {$S_n$} symmetry}, SIAM J. Appl. Math., 71
  (2011), pp.~1204--1219.

\bibitem{Ta}
{\sc R.~I. Tanner}, {\em Engineering Rheology}, Oxford Univ. Press, New York,
  second~ed., 2000.

\bibitem{ThEA}
{\sc E.~A. Theisen, M.~J. Vogel, C.~A. L\'opez, A.~H. Hirsa, and P.~H. Steen},
  {\em Capillary dynamics of coupled spherical-cap droplets}, J. Fluid Mech.,
  580 (2007), pp.~495--505.

\bibitem{LeEA1}
{\sc H.~B. van Lengerich, M.~J. Vogel, and P.~H. Steen}, {\em Dynamics and
  stability of volume-scavenging drop arrays: Coarsening by capillarity},
  Physica D, 238 (2009), pp.~531--539.

\bibitem{LeEA2}
{\sc H.~B. van Lengerich, M.~J. Vogel, and P.~H. Steen}, {\em Coarsening of
  capillary drops coupled by conduit networks}, Phys. Rev. E, 82 (2010), 066312
  (11~pages).

\bibitem{VoSt}
{\sc M.~J. {Vogel} and P.~H. {Steen}}, {\em Capillarity-based switchable
  adhesion}, Proc. Nat. Acad. Sci. USA, 107 (2010), pp.~3377--3381.

\bibitem{WeBa}
{\sc F.~Weinhaus and W.~Barker}, {\em On the equilibrium states of
  interconnected bubbles or balloons}, Am. J. Phys., 46 (1978), pp.~978--982.

\bibitem{We}
{\sc H.~C. Wente}, {\em A surprising bubble catastrophe}, Pacific J. Math., 189
  (1999), pp.~339--375.

\end{thebibliography}

\end{document}